\documentclass{lmcs}
\usepackage{expl3}
\usepackage{bussproofs}
\usepackage{bm}
\usepackage{stmaryrd}
\ExplSyntaxOn
\cs_generate_variant:Nn\tl_if_empty:nTF{f}
\newcommand*\ifpresent[3]{\tl_if_empty:fTF{#1}{#3}{#2}}
\ExplSyntaxOff
\newcommand*\AXM[1]{\AxiomC{$#1$}}
\newcommand*\UIM[1]{\UnaryInfC{$#1$}}
\newcommand*\BIM[1]{\BinaryInfC{$#1$}}

\newcommand*\RLM[1]{\RightLabel{$\scriptstyle#1$}}
\newcommand*\DP\DisplayProof
\newcommand*\Def{\mathrel{\overset{\Delta}{=}}}
\newcommand*\GramDef{\;\;\mathrel{::=}\;\;}
\newcommand*\Entails{\mathrel{\vdash}}
\newcommand*\BarSep{\mathrel{|}}

\newcommand*\FV[1]{\text{FV}\left(#1\right)}
\newcommand*\SetSuch[2]{{\left\{#1\ifpresent{#2}{\!\;\middle|\;#2}{}\right\}}}
\newcommand*\Subst[2]{\left\{#2/#1\right\}}
\newcommand*\FormalSubst[2]{\left[\ifpresent{#1}{#1\mapsto}{}#2\right]}
\newcommand*\FormalShift[1]{\raisebox{.15em}{$\uparrow$}\ifpresent{#1}{\raisebox{-.2em}{$\mkern-3mu\scriptstyle#1$}}{}\mkern2mu}
\newcommand*\ValuatA{v}
\newcommand*\FTypeA{T}
\newcommand*\FTypeB{U}
\newcommand*\FTypeC{V}
\newcommand*\FRed{\mathrel{\succ}}
\newcommand*\LogVarA{\epsilon}
\newcommand*\LogRelVarA{\eta}
\newcommand*\LogNatA{m}
\newcommand*\LogNatVarA{i}
\newcommand*\LogNatZ{0}
\newcommand*\LogNatS[1]{S\,#1}
\newcommand*\LogTermA{M}
\newcommand*\LogTermB{N}
\newcommand*\LogTermC{P}
\newcommand*\LogTermVarA{t}
\newcommand*\LogTermVarB{u}
\newcommand*\LogTermVarC{v}
\newcommand*\LogTermVar[1]{\underline{#1}}
\newcommand*\LogTermLam[1]{\lambda.#1}
\newcommand*\LogTermApp[2]{{#1}\,{#2}}
\newcommand*\LogTermSubst[2]{#1\FormalSubst{}{#2}}
\newcommand*\LogTermListA{\Pi}
\newcommand*\LogTermListVarA{\pi}
\newcommand*\LogTermListNil{\left\langle\right\rangle}
\newcommand*\LogTermListCons[2]{\left\langle#1,#2\right\rangle}
\newcommand*\LogTermListAbbrv[1]{\left\langle#1\right\rangle}
\newcommand*\LogPredVarA{X}

\newcommand*\LogPredVarForm[1]{\overline{#1}}
\newcommand*\LogBoolA{\Phi}
\newcommand*\LogBoolVarA{b}
\newcommand*\LogBoolT{t\mkern-2mut}
\newcommand*\LogBoolF{f\mkern-6muf}
\newcommand*\LogBoolIn[2]{#1\in#2}
\newcommand*\LogNorm[1]{#1\mkern-5mu\downarrow}
\newcommand*\LogNNorm[2]{#1\backslash\mkern-14mu\downarrow^{#2}}
\newcommand*\LogNormForm{\Downarrow}
\newcommand*\LogFormA{A}
\newcommand*\LogFormB{B}
\newcommand*\LogFormC{C}
\newcommand*\LogFormD{D}
\newcommand*\LogRedCand[1]{\mathcal{R}ed\mathcal{C}and\left(#1\right)}
\newcommand*\LogRC[1]{RC_{#1}}
\newcommand*\LT{\Lambda T}
\newcommand*\LTbbc{\Lambda T_{bbc}}
\newcommand*\LTRed{\leadsto}
\newcommand*\LTTypeA{{\bm{\sigma}}}
\newcommand*\LTTypeB{{\bm{\tau}}}
\newcommand*\LTTypeNat{{\bm{\iota}}}
\newcommand*\LTTypeLam{{\bm{\lambda}}}
\newcommand*\LTTypeLamList{{\LTTypeLam^*}}
\newcommand*\LTTypeFinFun[1]{{#1}^\dagger}
\newcommand*\LTCst{\mathcal{C}st}
\newcommand*\LTConstA{\mathtt{c}}
\newcommand*\LTVarA{x}
\newcommand*\LTVarB{y}
\newcommand*\LTVarC{z}
\newcommand*\LTVarD{u}
\newcommand*\LTTermA{\mathtt{M}}
\newcommand*\LTTermB{\mathtt{N}}
\newcommand*\LTTermC{\mathtt{P}}
\newcommand*\LTTermD{\mathtt{Q}}
\newcommand*\LTTermE{\mathtt{R}}
\newcommand*\LTTermF{\mathtt{S}}
\newcommand*\LTValA{\mathtt{U}}
\newcommand*\LTValB{\mathtt{V}}
\newcommand*\LTEnvA[1]{E\left[#1\right]}
\newcommand*\LTProj{\mathtt{p}}
\newcommand*\LTPair[2]{\left\langle#1,#2\right\rangle}
\newcommand*\LTNatZ{\mathtt{z}}
\newcommand*\LTNatS{\mathtt{s}}
\newcommand*\LTNatIt{\mathtt{it}_\LTTypeNat}
\newcommand*\LTLamVar{\mathtt{var}}
\newcommand*\LTLamAbs{\mathtt{abs}}
\newcommand*\LTLamApp{\mathtt{app}}
\newcommand*\LTLamListApp{\mathtt{app}^*}
\newcommand*\LTLamIt{\mathtt{it}_\LTTypeLam}
\newcommand*\LTLamSubst{\mathtt{subst}}
\newcommand*\LTLamListShift{\mathtt{shift}^*}
\newcommand*\LTLamListNil{\mathtt{nil}}
\newcommand*\LTLamListCons{\mathtt{cons}}
\newcommand*\LTLamListIt{\mathtt{it}_\LTTypeLamList}
\newcommand*\LTcan{\mathtt{can}}
\newcommand*\LTFinFunEmpty{\left\{\right\}}
\newcommand*\LTFinFunExtend[3]{#1\cup\left\{#2\mapsto#3\right\}}
\newcommand*\LTFinFunComplete[2]{#1\mathbin{|}#2}
\newcommand*\LTBBC{\mathtt{bbc}}
\newcommand*\LTdne{\mathtt{dne}}
\newcommand*\LTexf{\mathtt{exf}}
\newcommand*\LTcomp{\mathtt{comp}}
\newcommand*\LTrepl{\mathtt{repl}}
\newcommand*\LTelim{\mathtt{elim}}
\newcommand*\LTnormrc{\mathtt{normrc}}
\newcommand*\LTisrc{\mathtt{isrc}}
\newcommand*\LTadeq{\mathtt{adeq}}
\newcommand*\LTnorm{\mathtt{norm}}
\newcommand*\LTred{\mathtt{red}}
\newcommand*\CPOA{D}
\newcommand*\CPOB{E}
\newcommand*\CPOelA{\varphi}
\newcommand*\CPOelB{\psi}
\newcommand*\CPOelC{\theta}
\newcommand*\CPOelD{\xi}
\newcommand*\CPOelE{\zeta}

\newcommand*\CPOdir{\Delta}
\newcommand*\CPObot{\bot}
\newcommand*\CPOsup{\sqcup}
\newcommand*\CPOinterp[1]{\left\llbracket#1\right\rrbracket}

\newcommand*\CPOproj{\pi}
\newcommand*\LTinterp[1]{{#1}^\diamond}
\newcommand*\RealVal[1]{\left|#1\right|}
\newcommand*\RealBot{\bot\mkern-11mu\bot}
\newcommand*\RealNatA{\mathfrak{n}}
\newcommand*\RealTermA{\mathfrak{M}}
\newcommand*\RealTermB{\mathfrak{N}}
\newcommand*\RealTermC{\mathfrak{P}}
\newcommand*\RealTermListA{\mathfrak{p}}
\newcommand*\RealPredA{\mathfrak{X}}
\newcommand*\RealBoolA{\mathfrak{b}}
\newcommand*\RealBoolT{\mathfrak{t\mkern-3mut}}
\newcommand*\RealBoolF{\mathfrak{f\mkern-3muf}}

\usepackage{hyperref}
\begin{document}
\title{An interpretation of system F through bar recursion}
\titlecomment{This article is an extended version of~\cite{BlotF}}
\author{Valentin Blot}
\address{LRI, Universit\'e Paris Sud, CNRS, Universit\'e Paris-Saclay, France}
\keywords{polymorphism, system F, bar recursion, BBC functional, realizability}
\thanks{This research was supported by the Labex DigiCosme (project ANR11LABEX0045DIGICOSME) operated by ANR as part of the program ``Investissements d'Avenir'' Idex ParisSaclay (ANR11IDEX000302).}
\begin{abstract}
There are two possible computational interpretations of second-order arithmetic: Girard's system F or Spector's bar recursion and its variants. While the logic is the same, the programs obtained from these two interpretations have a fundamentally different computational behavior and their relationship is not well understood. We make a step towards a comparison by defining the first translation of system F into a simply-typed total language with a variant of bar recursion. This translation relies on a realizability interpretation of second-order arithmetic. Due to G\"odel's incompleteness theorem there is no proof of termination of system F within second-order arithmetic. However, for each individual term of system F there is a proof in second-order arithmetic that it terminates, with its realizability interpretation providing a bound on the number of reduction steps to reach a normal form. Using this bound, we compute the normal form through primitive recursion. Moreover, since the normalization proof of system F proceeds by induction on typing derivations, the translation is compositional. The flexibility of our method opens the possibility of getting a more direct translation that will provide an alternative approach to the study of polymorphism, namely through bar recursion.
\end{abstract}
\maketitle
\section{Introduction}
Second-order $\lambda$-calculus~\cite{GirardF,ReynoldsPolymorphism} is a poweful type system in which terms such as $\lambda x.x\,x$ can be typed. The language obtained is still strongly normalizing, but so far all proofs of this fact rely on the notion of reducibility candidates (RCs): sets of $\lambda$-terms satisfying some axioms. In these proofs, every type has an associated RC and every typed term belongs to the RC associated to its type. Normalization is then a consequence of the axioms of RCs. An important aspect of these proofs is that they are impredicative: the RC associated to a universally quantified type is the intersection over all RCs, which includes the intersection itself. Our translation reduces the termination of system F to the termination of a variant of bar recursion that is proved with an instance of Zorn's lemma, thus avoiding the direct use of impredicative RCs.\par
In 1962, Spector used bar recursion~\cite{Spector} to interpret the axiom scheme of comprehension and therefore extend G\"odel's Dialectica interpretation of arithmetic into an interpretation of analysis. Variants of bar recursion have then been used in Kreisel's modified realizability to interpret the axioms of countable and dependent choice in a classical setting. Among these variants, modified bar recursion~\cite{BergerOlivaChoice} relies on the continuity of one of its arguments to ensure termination, rather than on the explicit termination condition of Spector's original version. Krivine used this variant in untyped realizability for set theory~\cite{KrivineBarRec}. We use here the BBC functional~\cite{BerardiBezemCoquand}, another variant of bar recursion that builds the elements of the choice sequence when they are needed, rather than sequentially. Our proof of correctness of this operator is adapted from the semantic poof of~\cite{BergerBBCDomains} that relies on Zorn's lemma. We extend the usual realizability interpretation of first-order arithmetic into an interpretation of its second-order counterpart by interpreting the axiom scheme of comprehension with the BBC functional.\par
For any single term of system F there exists a proof in second-order arithmetic that it terminates. This mapping from terms of system F to proofs of second-order arithmetic is closely related to Reynolds' abstraction theorem~\cite{ReynoldsParametricity} which, as explained in~\cite{WadlerIsomorphism}, relies on an embedding of system F into second-order arithmetic. We use our interpretation of second-order arithmetic to extract the normal form of the system F term from its termination proof. Our technique is similar to Berger's work in the simply-typed case~\cite{BergerNbe} and is closely related to normalization by evaluation, extended to system F in~\cite{AHSNbe,AbelNbeF}. We define a multi-sorted first-order logic with a sort for $\lambda$-terms with de Bruijn indices to avoid an encoding of $\lambda$-terms as natural numbers. Our logic is also equipped with a sort for sets of $\lambda$-terms so we can formalize the notion of reducibility candidates. Since these sets are first-order elements of the logic, we cannot instantiate a set variable with an arbitrary formula as we would in second-order logic. Nevertheless we get back this possibility through our interpretation of the axiom scheme of comprehension with the BBC functional.\par
In a second step we fix the target programming language of the translation. This language, that we call system $\LTbbc$, is purely functional with a type of $\lambda$-terms, primitive recursion, and the BBC functional. System $\LTbbc$ is in particular simply-typed and total. We also describe the sound and computationally adequate semantics of this language in the category of complete partial orders.\par
The last step is the definition of a realizability semantics for our logic. To each formula we associate a type of system $\LTbbc$ and a set of realizers in the complete partial order interpreting that type. Defining realizers as elements of the model rather than syntactic programs simplifies the correctness proof for the BBC functional since we have non-computable functions on discrete types in the model. We interpret classical logic through an encoding of existential quantifications in terms of the universal ones and negation. The BBC functional interprets a variant of the axiom of countable choice which, combined with our interpretation of classical logic, provides a realizer of the axiom scheme of comprehension. Using this realizer, we interpret the instantiation of set variables with arbitrary formulas and therefore full second-order arithmetic. Finally, each program of system F is translated into a program of system $\LTbbc$ that computes the normal form of the initial term of system F through the realizability interpretation of its proof of termination for weak head reduction.
\section{Normalization of system F}
\label{NormProof}
We give here the proof of normalization of system F that we will interpret through realizability in section~\ref{Realizability}. In particular, we introduce a formal syntax suited to the formalization of the proof. Our notion of reducibility candidates is a simplified version of Tait's saturated sets~\cite{TaitRealizability} that also appears in~\cite{KrivineBook} and is sufficient for weak head reduction. We could use Girard's reducibility candidates~\cite{GirardPhD} but the corresponding normalization proof performs induction on the length of reduction of subterms in the arrow case of lemma~\ref{RCisRC} and the interpretation would be much more complicated.
\subsection{Terms and substitutions}
\label{Subst}
We describe here the formal syntax for $\lambda$-terms that we use throughout the paper. In particular this syntax will be part of our logic in section~\ref{Logic} $\alpha$-conversion can complicate the use of binders in logic and we avoid this issue by using de Bruijn indices so we have a canonical representation of $\lambda$-terms up to $\alpha$-equivalence. The formal syntax for the set $\Lambda$ of all $\lambda$-terms is given by the following grammar:
\begin{equation*}
\LogTermA,\LogTermB\GramDef\LogTermVar{\LogNatA}\BarSep\LogTermLam{\LogTermA}\BarSep\LogTermApp{\LogTermA}{\LogTermB}
\end{equation*}
where $\LogNatA$ is a natural number. We suppose that the reader is familiar with de Bruijn indices and do not recall here the translations between usual $\lambda$-terms and $\lambda$-terms with de Bruijn indices. We only give an example: the $\lambda$-term $\lambda xyz.y\left(\lambda u.x\right)$ is written with de Bruijn indices as $\LogTermLam{\LogTermLam{\LogTermLam{\LogTermApp{\LogTermVar{1}}{\left(\LogTermLam{\LogTermVar{3}}\right)}}}}$. Since we use Tait's style of reducibility candidate, we will have to manipulate $\lambda$-terms applied to an arbitrary number of arguments. We therefore also consider lists of $\lambda$-terms, for which we use the notation $\LogTermListA=\LogTermListAbbrv{\LogTermA_0,\ldots,\LogTermA_{n-1}}$. We write $\LogTermApp{\LogTermA}{\LogTermListA}$ for $\LogTermA\,\LogTermA_0\,\ldots\,\LogTermA_{n-1}$ in $\Lambda$. Parallel substitution with de Bruijn indices requires the definition of a shift operation $\FormalShift{k}$ on terms. $\FormalShift{k}\LogTermA$ is the result of incrementing the value of all variables of $\LogTermA$ with an outer index $\geq k$, that is, the variables $\LogTermVar{\LogNatA}$ such that $\LogNatA\geq k+l$ where $l$ is the number of $\lambda$-abstractions above the variable $\LogTermVar{\LogNatA}$ in $\LogTermA$. $\FormalShift{k}\LogTermA$ is defined as follows:
\begin{align*}
\FormalShift{k}\LogTermVar{\LogNatA}&\Def\left\{\begin{aligned}&\LogTermVar{\LogNatA+1}&&\text{if }\LogNatA\geq k\\&\LogTermVar{\LogNatA}&&\text{otherwise}\end{aligned}\right.
&&
\begin{aligned}
\FormalShift{k}\left(\LogTermLam{\LogTermA}\right)&\Def\LogTermLam{\left(\FormalShift{k+1}\LogTermA\right)}\\
\FormalShift{k}\left(\LogTermApp{\LogTermA}{\LogTermB}\right)&\Def\LogTermApp{\left(\FormalShift{k}\LogTermA\right)}{\left(\FormalShift{k}\LogTermB\right)}
\end{aligned}
\end{align*}
This operation is extended to lists of terms:
\begin{equation*}
\FormalShift{k}\LogTermListAbbrv{\LogTermA_0,\ldots,\LogTermA_{n-1}}\Def\LogTermListAbbrv{\FormalShift{k}\LogTermA_0,\ldots,\FormalShift{k}\LogTermA_{n-1}}
\end{equation*}
We write $\FormalShift{}\LogTermA$ (resp. $\FormalShift{}\LogTermListA$) for $\FormalShift{\LogNatZ}\LogTermA$ (resp. $\FormalShift{\LogNatZ}\LogTermListA$), the result of incrementing all the free variables of $\LogTermA$. Using the shift operation, we define parallel substitution $\LogTermB\FormalSubst{k}{\LogTermListA}$ where $\LogTermListA=\left\langle\LogTermA_0,\ldots,\LogTermA_{n-1}\right\rangle$. The result of the parallel substitution $\LogTermB\FormalSubst{k}{\LogTermListA}$ is obtained by substituting $\LogTermA_i$ for variables of outer index $i$ such that $k\leq i<k+n$ in $\LogTermB$, and subtracting $n$ to variables of outer index $i\geq k+n$:
\begin{gather*}
\LogTermVar{\LogNatA}\FormalSubst{k}{\LogTermListAbbrv{\LogTermA_0,\ldots,\LogTermA_{n-1}}}\Def\left\{\begin{aligned}&\LogTermVar{\LogNatA}&&\text{if }\LogNatA<k\\&\LogTermA_{\LogNatA-k}&&\text{if }k\leq\LogNatA<k+n\\&\LogTermVar{\LogNatA-n}&&\text{otherwise}\end{aligned}\right.
\\
\begin{flalign*}
\left(\LogTermLam{\LogTermA}\right)\FormalSubst{k}{\LogTermListA}&\Def\LogTermLam{\left(\LogTermA\FormalSubst{k+1}{\FormalShift{}\LogTermListA}\right)}
&
\left(\LogTermApp{\LogTermA}{\LogTermB}\right)\FormalSubst{k}{\LogTermListA}&\Def\LogTermApp{\left(\LogTermA\FormalSubst{k}{\LogTermListA}\right)}{\left(\LogTermB\FormalSubst{k}{\LogTermListA}\right)}
\end{flalign*}
\end{gather*}
Substitution of a single term is defined as:
\begin{equation*}
\LogTermA\FormalSubst{k}{\LogTermB}\Def\LogTermA\FormalSubst{k}{\LogTermListAbbrv{\LogTermB}}
\end{equation*}
and we write $\LogTermA\FormalSubst{}{\LogTermListA}$ (resp. $\LogTermA\FormalSubst{}{\LogTermB}$) for $\LogTermA\FormalSubst{\LogNatZ}{\LogTermListA}$ (resp. $\LogTermA\FormalSubst{\LogNatZ}{\LogTermB}$). The usual $\beta$-reduction of $\lambda$-calculus is therefore:
\begin{equation*}
\LogTermApp{\left(\LogTermLam{\LogTermA}\right)}{\LogTermB}\FRed\LogTermA\FormalSubst{}{\LogTermB}
\end{equation*}
The following substitution lemma will be used in the proof of normalization:
\begin{lem}
\label{SubstLemma}
We have the following equality:
\begin{equation*}
\LogTermA\FormalSubst{k}{\LogTermListAbbrv{\LogTermB,\LogTermListA}}=\LogTermA\FormalSubst{k+1}{\FormalShift{k}\LogTermListA}\FormalSubst{k}{\LogTermB}
\end{equation*}
where $\LogTermListAbbrv{\LogTermB,\LogTermListA}$ is the result of prepending $\LogTermB$ to $\LogTermListA$.
\end{lem}
\begin{proof}
By induction on $\LogTermA$, using $\FormalShift{}\left(\FormalShift{k}\LogTermListA\right)=\FormalShift{k+1}\left(\FormalShift{}\LogTermListA\right)$ for the case of a $\lambda$-abstraction.
\end{proof}
\subsection{The normalization theorem}
We prove here the normalization of system F in a formal way so we can interpret it through realizability in section~\ref{Realizability}. As explained before we choose a simplified version of the usual proof that proves only weak head reduction so the interpretation is relatively simple, but any other proof couldd be used since our realizability model interprets full second-order arithmetic.\par
First, we recall the typing rules of system F in figure~\ref{TypingSystF}, where types are defined by the following grammar:
\begin{equation*}
\FTypeA,\FTypeB\GramDef\LogPredVarA\BarSep\FTypeA\to\FTypeB\BarSep\forall\LogPredVarA\,\FTypeA
\end{equation*}
where $\LogPredVarA$ ranges over a countable set of type variables.
\begin{figure}
\begin{gather*}
\begin{flalign*}
&\AXM{}
\RLM{0\leq\LogNatA<n}\UIM{\FTypeA_{n-1},\ldots,\FTypeA_0\Entails\LogTermVar{\LogNatA}:\FTypeA_\LogNatA}
\DP
&
&\AXM{\Gamma,\FTypeA\Entails\LogTermA:\FTypeB}
\UIM{\Gamma\Entails\LogTermLam{\LogTermA}:\FTypeA\to\FTypeB}
\DP
&
&\AXM{\Gamma\Entails\LogTermA:\FTypeA\to\FTypeB}
\AXM{\Gamma\Entails\LogTermB:\FTypeA}
\BIM{\Gamma\Entails\LogTermApp{\LogTermA}{\LogTermB}:\FTypeB}
\DP
\end{flalign*}
\\
\begin{align*}
&\AXM{\Gamma\Entails\LogTermA:\FTypeA}
\RLM{\LogPredVarA\notin\FV{\Gamma}}\UIM{\Gamma\Entails\LogTermA:\forall\LogPredVarA\,\FTypeA}
\DP
&
&\AXM{\Gamma\Entails\LogTermA:\forall\LogPredVarA\,\FTypeA}
\UIM{\Gamma\Entails\LogTermA:\FTypeA\Subst{\LogPredVarA}{\FTypeB}}
\DP&
\end{align*}
\end{gather*}
\caption{Typing rules of system F}
\label{TypingSystF}
\end{figure}
Since we work with de Bruijn indices, contexts are ordered lists of types (and the order is important). We use a Curry presentation (without type abstractions and applications within the terms) since it simplifies the syntax and we are not interested into type checking or inference. As explained above, we only consider weak head reduction:
\begin{equation*}
\LogTermApp{\LogTermApp{\left(\LogTermLam{\LogTermA}\right)}{\LogTermB}}{\LogTermListA}\FRed\LogTermApp{\LogTermA\FormalSubst{}{\LogTermB}}{\LogTermListA}
\end{equation*}
and write $\LogNorm{\LogTermA}$ if $\LogTermA$ normalizes for the above reduction.\par
The normalization proof goes as follows: first, we define the set $\mathcal{RC}\subseteq\mathcal{P}\left(\Lambda\right)$ of reducibility candidates and we prove that the set of normalizing terms is a reducibility candidate. Then, we associate a set $\LogRC{\FTypeA,\ValuatA}\subseteq\Lambda$ to each type $\FTypeA$ of system F with valuation $\ValuatA:\FV{\FTypeA}\to\mathcal{RC}$, and we prove that $\LogRC{\FTypeA,\ValuatA}$ is a reducibility candidate. Finally, we prove that if a closed term $\LogTermA$ is of closed type $\FTypeA$, then $\LogTermA\in\LogRC{\FTypeA,\emptyset}$. Since one of the properties of reducibility candidates is that they contain only normalizing terms, we can then conclude that $\LogTermA$ normalizes.\par
We now give the proof in more details. First, define the set $\mathcal{RC}$ of reducibility candidates:
\begin{defi}[Reducibility candidate]
$\RealPredA\subseteq\Lambda$ is in $\mathcal{RC}$ if:
\begin{itemize}
\item For any list of terms $\LogTermListA$, we have $\LogTermApp{\LogTermVar{\LogNatZ}}{\LogTermListA}\in\RealPredA$
\item If $\LogTermA\in\RealPredA$, then $\LogNorm{\LogTermA}$
\item If $\LogTermApp{\LogTermA\FormalSubst{}{\LogTermB}}{\LogTermListA}\in\RealPredA$, then $\LogTermApp{\LogTermApp{\left(\LogTermLam{\LogTermA}\right)}{\LogTermB}}{\LogTermListA}\in\RealPredA$
\end{itemize}
\end{defi}
In particular, the set of normalizing terms is a reducibility candidate:
\begin{lem}
\label{NormRC}
$\SetSuch{\LogTermA\in\Lambda}{\LogNorm{\LogTermA}}\in\mathcal{RC}$
\end{lem}
\begin{proof}
We prove the three properties of reducibility candidates:
\begin{itemize}
\item For any $\LogTermListA$, $\LogTermApp{\LogTermVar{\LogNatZ}}{\LogTermListA}$ is in head normal form so $\LogNorm{\LogTermApp{\LogTermVar{\LogNatZ}}{\LogTermListA}}$
\item If $\LogNorm{\LogTermA}$, then $\LogNorm{\LogTermA}$
\item If $\LogNorm{\LogTermApp{\LogTermA\FormalSubst{}{\LogTermB}}{\LogTermListA}}$, then $\LogNorm{\LogTermApp{\LogTermApp{\left(\LogTermLam{\LogTermA}\right)}{\LogTermB}}{\LogTermListA}}$ because $\LogTermApp{\LogTermApp{\left(\LogTermLam{\LogTermA}\right)}{\LogTermB}}{\LogTermListA}\FRed\LogTermApp{\LogTermA\FormalSubst{}{\LogTermB}}{\LogTermListA}$\qedhere
\end{itemize}
\end{proof}
As explained above, in the second step we define a set $\LogRC{\FTypeA,\ValuatA}\subseteq\Lambda$ for each type $\FTypeA$ with valuation $\ValuatA$:
\begin{defi}
If $\FTypeA$ is a type of system F and if $\ValuatA:\FV{\FTypeA}\to\mathcal{RC}$, we define $\LogRC{\FTypeA,\ValuatA}$ inductively:
\begin{itemize}
\item$\LogRC{\LogPredVarA,\ValuatA}=\ValuatA\left(\LogPredVarA\right)$
\item$\LogRC{\FTypeA\to\FTypeB,\ValuatA}=\SetSuch{\LogTermA}{\forall\LogTermB\in\LogRC{\FTypeA,\ValuatA},\LogTermA\,\LogTermB\in\LogRC{\FTypeB,\ValuatA}}$
\item$\LogRC{\forall\LogPredVarA\,\FTypeA,\ValuatA}=\bigcap\SetSuch{\LogRC{\FTypeA,\ValuatA\uplus\SetSuch{\LogPredVarA\mapsto\RealPredA}{}}}{\RealPredA\in\mathcal{RC}}$
\end{itemize}
\end{defi}
These sets are indeed reducibility candidates:
\begin{lem}
\label{RCisRC}
If $\FTypeA$ is a type and $\ValuatA:\FV{\FTypeA}\to\mathcal{RC}$, then $\LogRC{\FTypeA,\ValuatA}\in\mathcal{RC}$
\end{lem}
\begin{proof}
By induction on $\FTypeA$:
\begin{itemize}
\item Since $\ValuatA\left(\LogPredVarA\right)\in\mathcal{RC}$, we have $\LogRC{\LogPredVarA,\ValuatA}=\ValuatA\left(\LogPredVarA\right)\in\mathcal{RC}$
\item Suppose $\LogRC{\FTypeA,\ValuatA}\in\mathcal{RC}$ and $\LogRC{\FTypeB,\ValuatA}\in\mathcal{RC}$. We prove $\LogRC{\FTypeA\to\FTypeB,\ValuatA}\in\mathcal{RC}$:
\begin{itemize}
\item If $\LogTermListA$ is a list of terms and $\LogTermA\in\LogRC{\FTypeA,\ValuatA}$ then $\LogTermListCons{\LogTermListA}{\LogTermA}$ (result of appending $\LogTermA$ to $\LogTermListA$) is a list of terms so:
\begin{equation*}
\LogTermApp{\left(\LogTermApp{\LogTermVar{\LogNatZ}}{\LogTermListA}\right)}{\LogTermA}=\LogTermApp{\LogTermVar{\LogNatZ}}{\LogTermListCons{\LogTermListA}{\LogTermA}}\in\LogRC{\FTypeB,\ValuatA}
\end{equation*}
by induction hypothesis on $\FTypeB$, and therefore:
\begin{equation*}
\LogTermApp{\LogTermVar{\LogNatZ}}{\LogTermListA}\in\LogRC{\FTypeA\to\FTypeB,\ValuatA}
\end{equation*}
\item Let $\LogTermA\in\LogRC{\FTypeA\to\FTypeB,\ValuatA}$. We have $\LogTermVar{\LogNatZ}=\LogTermApp{\LogTermVar{\LogNatZ}}{\LogTermListNil}\in\LogRC{\FTypeA,\ValuatA}$ by induction hypothesis on $\FTypeA$, so $\LogTermApp{\LogTermA}{\LogTermVar{\LogNatZ}}\in\LogRC{\FTypeB,\ValuatA}$ by definition of $\LogRC{\FTypeA\to\FTypeB,\ValuatA}$ and $\LogNorm{\LogTermApp{\LogTermA}{\LogTermVar{\LogNatZ}}}$ by induction hypothesis on $\FTypeB$. Since every reduction sequence from $\LogTermA$ can be turned into a reduction sequence from $\LogTermApp{\LogTermA}{\LogTermVar{\LogNatZ}}$ with same length, we get $\LogNorm{\LogTermA}$.
\item Suppose $\LogTermApp{\LogTermA\FormalSubst{}{\LogTermB}}{\LogTermListA}\in\LogRC{\FTypeA\to\FTypeB,\ValuatA}$. Then for any $\LogTermC\in\LogRC{\FTypeA,\ValuatA}$ we have by definition of $\LogRC{\FTypeA\to\FTypeB,\ValuatA}$:
\begin{equation*}
\LogTermApp{\LogTermA\FormalSubst{}{\LogTermB}}{\LogTermListCons{\LogTermListA}{\LogTermC}}=\LogTermApp{\LogTermApp{\LogTermA\FormalSubst{}{\LogTermB}}{\LogTermListA}}{\LogTermC}\in\LogRC{\FTypeB,\ValuatA}
\end{equation*}
and therefore:
\begin{equation*}
\LogTermApp{\LogTermApp{\LogTermApp{\left(\LogTermLam{\LogTermA}\right)}{\LogTermB}}{\LogTermListA}}{\LogTermC}=\LogTermApp{\LogTermApp{\left(\LogTermLam{\LogTermA}\right)}{\LogTermB}}{\LogTermListCons{\LogTermListA}{\LogTermC}}\in\LogRC{\FTypeB,\ValuatA}
\end{equation*}
by induction hypothesis on $\FTypeB$. This proves $\LogTermApp{\LogTermApp{\left(\LogTermLam{\LogTermA}\right)}{\LogTermB}}{\LogTermListA}\in\LogRC{\FTypeA\to\FTypeB,\ValuatA}$.
\end{itemize}
\item Suppose $\LogRC{\FTypeA,\ValuatA\uplus\SetSuch{\LogPredVarA\mapsto\RealPredA}{}}\in\mathcal{RC}$ for every $\RealPredA\in\mathcal{RC}$.
\begin{itemize}
\item If $\LogTermListA$ is a list of terms, then $\LogTermApp{\LogTermVar{\LogNatZ}}{\LogTermListA}\in\LogRC{\FTypeA,\ValuatA\uplus\SetSuch{\LogPredVarA\mapsto\RealPredA}{}}$ for every $\RealPredA\in\mathcal{RC}$ by induction hypothesis on $\FTypeA$, and therefore $\LogTermApp{\LogTermVar{\LogNatZ}}{\LogTermListA}\in\LogRC{\forall\LogPredVarA\,\FTypeA,\ValuatA}$.
\item If $\LogTermA\in\LogRC{\forall\LogPredVarA\,\FTypeA,\ValuatA}$, then $\LogTermA\in\LogRC{\FTypeA,\ValuatA\uplus\SetSuch{\LogPredVarA\mapsto\SetSuch{\LogTermB\in\Lambda}{\LogNorm{\LogTermB}}}{}}$ since $\SetSuch{\LogTermB\in\Lambda}{\LogNorm{\LogTermB}}\in\mathcal{RC}$ by lemma~\ref{NormRC}, and therefore $\LogNorm{\LogTermA}$ by induction hypothesis on $\FTypeA$.
\item If $\LogTermApp{\LogTermA\FormalSubst{}{\LogTermB}}{\LogTermListA}\in\LogRC{\forall\LogPredVarA\,\FTypeA,\ValuatA}$ and $\RealPredA\in\mathcal{RC}$, then in particular:
\begin{equation*}
\LogTermApp{\LogTermA\FormalSubst{}{\LogTermB}}{\LogTermListA}\in\LogRC{\FTypeA,\ValuatA\uplus\SetSuch{\LogPredVarA\mapsto\RealPredA}{}}
\end{equation*}
and so:
\begin{equation*}
\LogTermApp{\LogTermApp{\left(\LogTermLam{\LogTermA}\right)}{\LogTermB}}{\LogTermListA}\in\LogRC{\FTypeA,\ValuatA\uplus\SetSuch{\LogPredVarA\mapsto\RealPredA}{}}
\end{equation*}
by induction hypothesis on $\FTypeA$. Therefore $\LogTermApp{\LogTermApp{\left(\LogTermLam{\LogTermA}\right)}{\LogTermB}}{\LogTermListA}\in\LogRC{\forall\LogPredVarA\,\FTypeA,\ValuatA}$.\qedhere
\end{itemize}
\end{itemize}
\end{proof}
In the last step of the normalization proof, we prove that each term of system F belongs to the reducibility candidate associated to its type:
\begin{lem}
\label{FNorm}
If $\FTypeA_{n-1},\ldots,\FTypeA_0\Entails\LogTermB:\FTypeB$ in system F and if $\ValuatA:\FV{\FTypeA_{n-1},\ldots\FTypeA_0,\FTypeB}\to\mathcal{RC}$ and $\LogTermListA=\LogTermListAbbrv{\LogTermA_0,\ldots,\LogTermA_{n-1}}$ are such that $\LogTermA_i\in\LogRC{\FTypeA_i,\ValuatA}$ for $0\leq i<n$, then $\LogTermB\FormalSubst{}{\LogTermListA}\in\LogRC{\FTypeB,\ValuatA}$
\end{lem}
\begin{proof}
By induction on the typing derivation:
\begin{itemize}
\item$\FTypeA_{n-1},\ldots,\FTypeA_0\Entails\LogTermVar{\LogNatA}:\FTypeA_\LogNatA$. We have $\LogTermVar{\LogNatA}\FormalSubst{}{\LogTermListA}=\LogTermA_\LogNatA\in\LogRC{\FTypeA_\LogNatA,\ValuatA}$ as an hypothesis.
\item$\FTypeA_{n-1},\ldots,\FTypeA_0\Entails\LogTermLam{\LogTermB}:\FTypeB\to\FTypeC$. If $\LogTermC\in\LogRC{\FTypeB,\ValuatA}$ then:
\begin{equation*}
\LogTermB\FormalSubst{1}{\FormalShift{}\LogTermListA}\FormalSubst{}{\LogTermC}=\LogTermB\FormalSubst{}{\LogTermListAbbrv{\LogTermC,\LogTermListA}}\in\LogRC{\FTypeC,\ValuatA}
\end{equation*}
by lemma~\ref{SubstLemma} and induction hypothesis on $\LogTermB$, so:
\begin{equation*}
\LogTermApp{\left(\LogTermLam{\LogTermB}\right)\FormalSubst{}{\LogTermListA}}{\LogTermC}=\LogTermApp{\LogTermLam{\left(\LogTermB\FormalSubst{1}{\FormalShift{}\LogTermListA}\right)}}{\LogTermC}\in\LogRC{\FTypeC,\ValuatA}
\end{equation*}
by definition of parallel substitution and the third property of reducibility candidates, since $\LogRC{\FTypeC,\ValuatA}\in\mathcal{RC}$ by lemma~\ref{RCisRC}.
\item$\FTypeA_{n-1},\ldots,\FTypeA_0\Entails\LogTermApp{\LogTermB}{\LogTermC}:\FTypeC$. We have:
\begin{equation*}
\left(\LogTermApp{\LogTermB}{\LogTermC}\right)\FormalSubst{}{\LogTermListA}=\LogTermApp{\left(\LogTermB\FormalSubst{}{\LogTermListA}\right)}{\left(\LogTermC\FormalSubst{}{\LogTermListA}\right)}\in\LogRC{\FTypeC,\ValuatA}
\end{equation*}
because $\LogTermB\FormalSubst{}{\LogTermListA}\in\LogRC{\FTypeB\to\FTypeC,\ValuatA}$ and $\LogTermC\FormalSubst{}{\LogTermListA}\in\LogRC{\FTypeB,\ValuatA}$ by induction hypotheses on $\LogTermB$ and $\LogTermC$.
\item$\FTypeA_{n-1},\ldots,\FTypeA_0\Entails\LogTermB:\forall\LogPredVarA\,\FTypeB$. If $\RealPredA\in\mathcal{RC}$ then $\LogTermA_i\in\LogRC{\FTypeA_i,\ValuatA\uplus\SetSuch{\LogPredVarA\mapsto\RealPredA}{}}$ because $\LogPredVarA\notin\FV{\FTypeA_i}$, and therefore $\LogTermB\FormalSubst{}{\LogTermListA}\in\LogRC{\FTypeB,\ValuatA\uplus\SetSuch{\LogPredVarA\mapsto\RealPredA}{}}$ by induction hypothesis on $\LogTermB:\FTypeB$.
\item$\FTypeA_{n-1},\ldots,\FTypeA_0\Entails\LogTermB:\FTypeB\Subst{\LogPredVarA}{\FTypeC}$. We have $\LogRC{\FTypeC,\ValuatA}\in\mathcal{RC}$ by lemma~\ref{RCisRC}, so:
\begin{equation*}
\LogTermB\FormalSubst{}{\LogTermListA}\in\LogRC{\FTypeB,\ValuatA\uplus\SetSuch{\LogPredVarA\mapsto\LogRC{\FTypeC,\ValuatA}}{}}=\LogRC{\FTypeB\Subst{\LogPredVarA}{\FTypeC},\ValuatA}
\end{equation*}
by induction hypothesis on $\LogTermB:\forall\LogPredVarA\,\FTypeB$. The equality between the two reducibility candidates is proved by induction on $\FTypeB$.\qedhere
\end{itemize}
\end{proof}
We can now conclude our normalization proof of system F:
\begin{thm}
If a closed term $\LogTermA$ has closed type $\FTypeA$ in system F, then $\LogNorm{\LogTermA}$
\end{thm}
\begin{proof}
Lemma~\ref{FNorm} gives $\LogTermA\in\LogRC{\FTypeA,\emptyset}$ and we get $\LogNorm{\LogTermA}$ by lemma~\ref{RCisRC} and the second property of reducibility candidates.
\end{proof}
\section{A logic for \texorpdfstring{$\lambda$}{lambda}-terms}
\label{Logic}
This section is devoted to the definition of a first-order multi-sorted logic in which we can easily formalize the normalization proof of system F described in the previous section. The main feature is that our logic has a sort of $\lambda$-terms, avoiding the usual encoding into natural numbers.
\subsection{Definitions}
Since we define a simply-typed realizability interpretation, we represent second-order artihmetic as a multi-sorted first-order theory. In particular, sets of $\lambda$-terms (and reducibility candidates) are first-order elements and we cannot instantiate them with arbitrary formulas. We will however get back this possibility in the next section through an interpretation of the axiom scheme of comprehension with the BBC functional. Since we formalize normalization of system F in this logic, we need a sort for $\lambda$-terms and a sort for sets of $\lambda$-terms. Moreover, the formal definition of reducibility candidates needs quantifications on lists of $\lambda$-terms so we have a sort for these. We will also manipulate lengths of reduction sequences so we include a sort for natural numbers. Finally, quantification over booleans will be convenient when interpreting the axiom scheme of comprehension so we include a sort for these as well. We distinguish elements of different sorts by using different notations:
\begin{gather*}
\begin{flalign*}
\LogNatA&\GramDef\LogNatVarA\BarSep\LogNatZ\BarSep\LogNatS{\LogNatA}
&
\LogTermA&\GramDef\LogTermVarA\BarSep\LogTermVar{\LogNatA}\BarSep\LogTermLam{\LogTermA}\BarSep\LogTermApp{\LogTermA}{\LogTermListA}\BarSep\LogTermSubst{\LogTermA}{\LogTermListA}
&
\LogTermListA&\GramDef\LogTermListVarA\BarSep\LogTermListNil\BarSep\LogTermListCons{\LogTermListA}{\LogTermA}
\end{flalign*}
\\
\begin{align*}
\LogPredVarA&\GramDef\LogPredVarA
&
\LogBoolA&\GramDef\LogBoolVarA\BarSep\LogBoolT\BarSep\LogBoolF\BarSep\LogBoolIn{\LogTermA}{\LogPredVarA}\BarSep\LogNNorm{\LogTermA}{\LogNatA}
&
\end{align*}
\end{gather*}
where $\LogNatVarA$, $\LogTermVarA$, $\LogTermListVarA$, $\LogPredVarA$ and $\LogBoolVarA$ range over countable sets of sorted variables of the logic. Notations $\LogNatA$, $\LogTermA$, $\LogTermListA$ and $\LogBoolA$ are used as meta-variables ranging over the first-order elements of the logic. Since the only elements of sort ``set'' (ranged over with $\LogPredVarA$) are variables, the meta-variables of sort ``set'' are exactly the variables of the logic of sort ``set'' so we can use the same notation for both. $\LogTermApp{\LogTermA}{\LogTermListA}$ is the application of term $\LogTermA$ to the list of arguments $\LogTermListA$, whereas $\LogTermSubst{\LogTermA}{\LogTermListA}$ it the parallel substitution of $\LogTermListA$ into $\LogTermA$. The elements $\LogBoolA$ are booleans that reflect validity. Note that in $\LogBoolIn{\LogTermA}{\LogPredVarA}$ (resp. $\LogNNorm{\LogTermA}{\LogNatA}$), $\LogBoolIn{}{}$ (resp. $\LogNNorm{}{}$) is formally a binary function symbol taking a term $\LogTermA$ and a set $\LogPredVarA$ (resp. a term $\LogTermA$ and a natural number $\LogNatA$) and returning a boolean. $\LogNNorm{\LogTermA}{\LogNatA}$ means that $\LogTermA$ can reduce for $\LogNatA$ steps of weak head reduction without reaching a normal form. We abbreviate $\LogTermListCons{\LogTermListCons{\ldots\LogTermListCons{\LogTermListNil}{\LogTermA_0}}{\ldots}}{\LogTermA_{n-1}}$ as $\LogTermListAbbrv{\LogTermA_0,\ldots,\LogTermA_{n-1}}$ and $\LogTermSubst{\LogTermA}{\LogTermListAbbrv{\LogTermB}}$ as $\LogTermSubst{\LogTermA}{\LogTermB}$. Formulas are defined as follows:
\begin{equation*}
\LogFormA,\LogFormB\GramDef\LogBoolA\BarSep\LogFormA\Rightarrow\LogFormB\BarSep\LogFormA\wedge\LogFormB\BarSep\forall\LogVarA\,\LogFormA
\end{equation*}
where $\LogVarA$ ranges over variables of any sort: $\LogNatVarA$, $\LogTermVarA$, $\LogTermListVarA$, $\LogPredVarA$, $\LogBoolVarA$. Formally, $\LogBoolA$ is a unary predicate symbol taking a boolean ($\LogBoolA$ itself) and should be thought of as ``$\LogBoolA=\LogBoolT$''. We also define the following abbreviations:
\begin{flalign*}
\neg\LogFormA&\Def\LogFormA\Rightarrow\LogBoolF
&
\exists\LogVarA\,\LogFormA&\Def\neg\forall\LogVarA\neg\LogFormA
&
\LogNorm{\LogTermA}&\Def\neg\forall\LogNatVarA\,\LogNNorm{\LogTermA}{\LogNatVarA}
&
\LogFormA\Leftrightarrow\LogFormB&\Def\left(\LogFormA\Rightarrow\LogFormB\right)\wedge\left(\LogFormB\Rightarrow\LogFormA\right)
\end{flalign*}
where $\LogVarA$ ranges over variables of any sort. Note that our logic does not contain primitive existential quantifications. This is because we need classical logic to interpret the axiom scheme of comprehension, and therefore we choose to work in a subset of intuitionistic logic corresponding to the target of G\"odel's negative translation. This is to be contrasted with the dialectica-like interpretations that perform an explicit negative translation from classical to intuitionistic logic, before giving a computational interpretation of the target of the translation.\par
We also define the notion of dependent formulas that will be useful to our formalization of the normalization proof. A 1-formula is a formula depending on first-order elements of the logic. For example, $\LogFormA\left(\LogTermA,\LogBoolA\right)\equiv\forall\LogTermListVarA\left(\LogTermApp{\LogTermA}{\LogTermListVarA}\in\LogPredVarA\Rightarrow\LogBoolA\right)$ is a 1-formula depending on a term $\LogTermA$ and a boolean $\LogBoolA$ (containing moreover a free variable $\LogPredVarA$). We avoid the capture of bound variables, so $\LogFormA\left(\LogTermApp{\LogTermVarA}{\LogTermListVarA},\LogBoolIn{\LogTermVarA}{\LogPredVarA}\right)$ is $\forall\LogTermListVarA'\left(\LogBoolIn{\LogTermApp{\LogTermApp{\LogTermVarA}{\LogTermListVarA}}{\LogTermListVarA'}}{\LogPredVarA}\Rightarrow\LogBoolIn{\LogTermVarA}{\LogPredVarA}\right)$. We also consider 2-formulas: formulas depending on 1-formulas. The only 2-formulas that we consider depend on one 1-formula which itself depends on one term. An example of such a 2-formula is $\LogFormA\left(\LogFormB\right)\equiv\forall\LogTermListVarA\left(\LogFormB\left(\LogTermApp{\LogTermVarA}{\LogTermListVarA}\right)\Rightarrow\LogBoolIn{\LogTermVar{\LogNatZ}}{\LogPredVarA}\right)$. Again, we avoid the capture of bound variables: if for example $\LogFormB\left(\LogTermA\right)\equiv\LogBoolIn{\LogTermApp{\LogTermA}{\LogTermListVarA}}{\LogPredVarA}\Rightarrow\LogBoolIn{\LogTermA}{\LogPredVarA}$, then $\LogFormA\left(\LogFormB\right)$ is $\forall\LogTermListVarA'\left(\left(\LogBoolIn{\LogTermApp{\LogTermApp{\LogTermVarA}{\LogTermListVarA'}}{\LogTermListVarA}}{\LogPredVarA}\Rightarrow\LogBoolIn{\LogTermApp{\LogTermVarA}{\LogTermListVarA'}}{\LogPredVarA}\right)\Rightarrow\LogBoolIn{\LogTermVar{\LogNatZ}}{\LogPredVarA}\right)$.\par
For each variable $\LogPredVarA$ of sort set we define the 1-formula $\LogPredVarForm{\LogPredVarA}\left(\LogTermA\right)\equiv\LogBoolIn{\LogTermA}{\LogPredVarA}$. We also define the 1-formula $\LogNormForm\left(\LogTermA\right)\equiv\LogNorm{\LogTermA}$. If $\LogFormA$ is a formula and $\LogPredVarA$ is a variable of sort set, then we write $\LogPredVarForm{\LogPredVarA}\mapsto\LogFormA$ for the 2-formula such that $\left(\LogPredVarForm{\LogPredVarA}\mapsto\LogFormA\right)\left(\LogFormB\right)$ is $\LogFormA$ where every atom of the form $\LogBoolIn{\LogTermA}{\LogPredVarA}$ has been replaced with $\LogFormB\left(\LogTermA\right)$.\par
The 2-formula $\LogRedCand{\LogFormA}$ says that the set of $\LogTermA$ such that $\LogFormA\left(\LogTermA\right)$ holds is a reducibility candidate:
\begin{equation*}
\LogRedCand{\LogFormA}\Def\left(\forall\LogTermListVarA\,\LogFormA\left(\LogTermApp{\LogTermVar{\LogNatZ}}{\LogTermListVarA}\right)\wedge\forall\LogTermVarA\left(\LogFormA\left(\LogTermVarA\right)\Rightarrow\LogNorm{\LogTermVarA}\right)\right)\wedge\forall\LogTermVarA\,\forall\LogTermVarB\,\forall\LogTermListVarA\left(\LogFormA\left(\LogTermApp{\LogTermSubst{\LogTermVarA}{\LogTermVarB}}{\LogTermListVarA}\right)\Rightarrow\LogFormA\left(\LogTermApp{\LogTermApp{\left(\LogTermLam{\LogTermVarA}\right)}{\LogTermListAbbrv{\LogTermVarB}}}{\LogTermListVarA}\right)\right)
\end{equation*}
Finally, to each type $\FTypeA$ of system F built from variables $\LogPredVarA$ of our logic we associate the 1-formula $\LogRC{\FTypeA}$ defined inductively as:
\begin{gather*}
\begin{align*}
\LogRC{\LogPredVarA}&\Def\LogPredVarForm{\LogPredVarA}
&
\LogRC{\FTypeA\to\FTypeB}\left(\LogTermA\right)&\Def\forall\LogTermVarA\left(\LogRC{\FTypeA}\left(\LogTermVarA\right)\Rightarrow\LogRC{\FTypeB}\left(\LogTermApp{\LogTermA}{\LogTermListAbbrv{\LogTermVarA}}\right)\right)
&
\end{align*}
\\
\LogRC{\forall\LogPredVarA\,\FTypeA}\left(\LogTermA\right)\Def\forall\LogPredVarA\left(\LogRedCand{\LogPredVarForm{\LogPredVarA}}\Rightarrow\LogRC{\FTypeA}\left(\LogTermA\right)\right)
\end{gather*}
The free variables of sort set in $\LogRC{\FTypeA}\left(\LogTermA\right)$ are exactly the free variables of $\FTypeA$.
\subsection{Interpreting normalization of system F}
This section contains an overview of our interpretation from a purely logical point of view, the computational interpretation itself will be given in section~\ref{Realizability}. Since our logic is first-order, we have to interpret the instantiation of a set variable with an arbitrary formula. In order to do that we first interpret the axiom scheme of comprehension with the BBC functional. If $\LogFormA$ is a 1-formula depending on a term and if $\LogPredVarA$ is a set variable that is not free in $\LogFormA\left(\LogTermA\right)$, then the corresponding instance of comprehension is:
\begin{equation*}
\exists\LogPredVarA\,\forall\LogTermVarA\left(\LogBoolIn{\LogTermVarA}{\LogPredVarA}\Leftrightarrow\LogFormA\left(\LogTermVarA\right)\right)
\end{equation*}
Then, using comprehension, we interpret the first-order equivalent of the elimination of second-order quantification. If $\LogFormA$ is a 2-formula, $\LogFormB$ is a 1-formula depending on a term and if $\LogPredVarA\notin\FV{\LogFormA}$ (meaning that $\LogPredVarA\notin\FV{\LogFormA\left(\LogFormC\right)}$ whenever $\LogPredVarA\notin\FV{\LogFormC}$), then this is:
\begin{equation*}
\forall\LogPredVarA\LogFormA\left(\LogPredVarForm{\LogPredVarA}\right)\Rightarrow\LogFormA\left(\LogFormB\right)
\end{equation*}
Interpreting this family of implications from the axiom scheme of comprehension requires the definition of a realizer by induction on $\LogFormA$. The interpretation of the instantiation of set variables with arbitrary formulas provides us with an interpretation of full second-order arithmetic. Building on this, we then interpret the formalization of lemma~\ref{NormRC} in our logic:
\begin{equation*}
\LogRedCand{\LogNormForm}
\end{equation*}
As a second step, we interpret the formalization of lemma~\ref{RCisRC}. If $\FV{\FTypeA}\subseteq\SetSuch{\LogPredVarA_0,\ldots,\LogPredVarA_{n-1}}{}$ then this is:
\begin{equation*}
\forall\LogPredVarA_0(\LogRedCand{\LogPredVarForm{\LogPredVarA_0}}\Rightarrow\ldots\Rightarrow\forall\LogPredVarA_{n-1}(\LogRedCand{\LogPredVarForm{\LogPredVarA_{n-1}}}\Rightarrow\LogRedCand{\LogRC{\FTypeA}})\ldots)
\end{equation*}
As a last step we interpret the formalization of lemma~\ref{FNorm}. If $\FTypeA_0,\ldots,\FTypeA_{m-1},\FTypeB$ are types such that $\FV{\FTypeA_0,\ldots,\FTypeA_{m-1},\FTypeB}\subseteq\SetSuch{\LogPredVarA_0,\ldots,\LogPredVarA_{n-1}}{}$ and if $\FTypeA_{m-1},\ldots,\FTypeA_0\Entails\LogTermA:\FTypeB$ is the conclusion of a valid typing derivation in system F, then this is:
\begin{multline*}
\forall\LogPredVarA_0(\LogRedCand{\LogPredVarForm{\LogPredVarA_0}}\Rightarrow\ldots\Rightarrow\forall\LogPredVarA_{n-1}(\LogRedCand{\LogPredVarForm{\LogPredVarA_{n-1}}}\\*
\Rightarrow\forall\LogTermVarA_{m-1}(\LogRC{\FTypeA_{m-1}}\left(\LogTermVarA_{m-1}\right)\Rightarrow\ldots\Rightarrow\forall\LogTermVarA_0(\LogRC{\FTypeA_0}\left(\LogTermVarA_0\right)\\*
\Rightarrow\LogRC{\FTypeB}\left(\LogTermSubst{\LogTermA}{\LogTermListAbbrv{\LogTermVarA_0,\ldots,\LogTermVarA_{m-1}}}\right))\ldots))\ldots)
\end{multline*}
The interpretation of the formula above provides a realizer of $\LogRC{\FTypeA}\left(\LogTermA\right)$ for each closed term $\LogTermA$ of closed type $\FTypeA$ in system F, from which we extract a bound on the number of reduction steps needed for reaching a normal form. Finally, we use this extracted bound to compute the normal form of $\LogTermA$ using primitive recursion.
\section{A simply-typed programming language with the BBC functional}
In this section, we define the target of our translation of system F: a simply-typed functional programming language that we call system $\LTbbc$. This language has product types, basic types for natural numbers, $\lambda$-terms and lists of $\lambda$-terms, primitive recursion on these basic types and the BBC functional. We also give a domain-theoretic denotational semantics for this programming language that is sound and computationally adequate.
\subsection{Syntax of system \texorpdfstring{$\LTbbc$}{lambda-T-bbc}}
We first define system $\LT$, and then extend it to system $\LTbbc$ by adding the BBC functional together with its reduction rule. The programming language system $\LT$ is an extension of G\"odel's system $T$ with types for $\lambda$-terms and lists of $\lambda$-terms, together with primitive recursion on these. The types of system $\LT$ are defined by the following grammar:
\begin{equation*}
\LTTypeA,\LTTypeB\GramDef\LTTypeNat\BarSep\LTTypeLam\BarSep\LTTypeLamList\BarSep\LTTypeA\to\LTTypeB\BarSep\LTTypeA\times\LTTypeB
\end{equation*}
where $\LTTypeNat$ is the type of natural numbers, $\LTTypeLam$ is the type of $\lambda$-terms, $\LTTypeLamList$ is the type of lists of $\lambda$-terms, $\LTTypeA\to\LTTypeB$ is the type of functions from $\LTTypeA$ to $\LTTypeB$ and $\LTTypeA\times\LTTypeB$ is the product type of $\LTTypeA$ and $\LTTypeB$. The syntax of system $\LT$ is given along with its typing rules in figure~\ref{TypingLT} and its reduction rules are given in figure~\ref{RedLT}.
\begin{figure}
\begin{gather*}
\begin{align*}
&\AXM{}
\UIM{\Gamma,\LTVarA:\LTTypeA\Entails\LTVarA:\LTTypeA}
\DP
&
&\AXM{\Gamma,\LTVarA:\LTTypeA\Entails\LTTermA:\LTTypeB}
\UIM{\Gamma\Entails\lambda\LTVarA.\LTTermA:\LTTypeA\to\LTTypeB}
\DP
&
&\AXM{\Gamma\Entails\LTTermA:\LTTypeA\to\LTTypeB}
\AXM{\Gamma\Entails\LTTermB:\LTTypeA}
\BIM{\Gamma\Entails\LTTermA\,\LTTermB:\LTTypeB}
\DP
&
\end{align*}
\\[5pt]
\begin{align*}
&\AXM{}
\RLM{\left(\LTConstA:\LTTypeA\right)\in\LTCst}\UIM{\Gamma\Entails\LTConstA:\LTTypeA}
\DP
&
&\AXM{\Gamma\Entails\LTTermA:\LTTypeA}
\AXM{\Gamma\Entails\LTTermB:\LTTypeB}
\BIM{\Gamma\Entails\LTPair{\LTTermA}{\LTTermB}:\LTTypeA\times\LTTypeB}
\DP
&
&\AXM{\Gamma\Entails\LTTermA:\LTTypeA\times\LTTypeB}
\UIM{\Gamma\Entails\LTProj_1\,\LTTermA:\LTTypeA}
\DP
&
&\AXM{\Gamma\Entails\LTTermA:\LTTypeA\times\LTTypeB}
\UIM{\Gamma\Entails\LTProj_2\,\LTTermA:\LTTypeB}
\DP
&
\end{align*}
\\
\begin{flalign*}
&\text{where $\LTCst$ is:}&
\end{flalign*}
\\
\begin{flalign*}
&\begin{gathered}
\LTNatZ:\LTTypeNat
\hfill
\LTNatS:\LTTypeNat\to\LTTypeNat
\\
\LTNatIt:\LTTypeA\to\left(\LTTypeA\to\LTTypeA\right)\to\LTTypeNat\to\LTTypeA
\end{gathered}
&
&\begin{gathered}
\LTLamVar:\LTTypeNat\to\LTTypeLam
\hfill
\LTLamAbs:\LTTypeLam\to\LTTypeLam
\hfill
\LTLamApp:\LTTypeLam\to\LTTypeLam\to\LTTypeLam
\\
\LTLamIt:\left(\LTTypeNat\to\LTTypeA\right)\to\left(\LTTypeA\to\LTTypeA\right)\to\left(\LTTypeA\to\LTTypeA\to\LTTypeA\right)\to\LTTypeLam\to\LTTypeA
\end{gathered}
\end{flalign*}
\\
\begin{align*}
\LTLamListNil&:\LTTypeLamList
&
\LTLamListCons&:\LTTypeLamList\to\LTTypeLam\to\LTTypeLamList
&
\LTLamListIt&:\LTTypeA\to\left(\LTTypeA\to\LTTypeLam\to\LTTypeA\right)\to\LTTypeLamList\to\LTTypeA
&
\end{align*}
\end{gather*}
\caption{Typing rules of system $\LT$}
\label{TypingLT}
\end{figure}
\begin{figure}
\begin{gather*}
\begin{align*}
&\AXM{}
\UIM{\left(\lambda\LTVarA.\LTTermA\right)\LTTermB\LTRed\LTTermA\Subst{\LTVarA}{\LTTermB}}
\DP
&
&\AXM{}
\UIM{\LTNatIt\,\LTTermA\,\LTTermB\,\LTNatZ\LTRed\LTTermA}
\DP
&
&\AXM{}
\UIM{\LTLamIt\,\LTTermA\,\LTTermB\,\LTTermC\left(\LTLamVar\,\LTValA\right)\LTRed\LTTermA\,\LTValA}
\DP
&
\end{align*}
\\[5pt]
\begin{align*}
&\AXM{}
\UIM{\LTProj_1\,\LTPair{\LTTermA}{\LTTermB}\LTRed\LTTermA}
\DP
&
&\AXM{}
\UIM{\LTNatIt\,\LTTermA\,\LTTermB\left(\LTNatS\,\LTValA\right)\LTRed\LTTermB\left(\LTNatIt\,\LTTermA\,\LTTermB\,\LTValA\right)}
\DP
&
&\AXM{}
\UIM{\LTLamIt\,\LTTermA\,\LTTermB\,\LTTermC\left(\LTLamAbs\,\LTValA\right)\LTRed\LTTermB\left(\LTLamIt\,\LTTermA\,\LTTermB\,\LTTermC\,\LTValA\right)}
\DP
&
\end{align*}
\\[5pt]
\begin{align*}
&\AXM{}
\UIM{\LTProj_2\,\LTPair{\LTTermA}{\LTTermB}\LTRed\LTTermB}
\DP
&
&\AXM{}
\UIM{\LTLamIt\,\LTTermA\,\LTTermB\,\LTTermC\left(\LTLamApp\,\LTValA\,\LTValB\right)\LTRed\LTTermC\left(\LTLamIt\,\LTTermA\,\LTTermB\,\LTTermC\,\LTValA\right)\left(\LTLamIt\,\LTTermA\,\LTTermB\,\LTTermC\,\LTValB\right)}
\DP
&
\end{align*}
\\[5pt]
\begin{align*}
&\AXM{}
\UIM{\LTLamListIt\,\LTTermA\,\LTTermB\,\LTLamListNil\LTRed\LTTermA}
\DP
&
&\AXM{}
\UIM{\LTLamListIt\,\LTTermA\,\LTTermB\left(\LTLamListCons\,\LTValA\,\LTValB\right)\LTRed\LTTermB\left(\LTLamListIt\,\LTTermA\,\LTTermB\,\LTValA\right)\LTValB}
\DP
&
&\AXM{\LTTermA\LTRed\LTTermB}
\UIM{\LTEnvA{\LTTermA}\LTRed\LTEnvA{\LTTermB}}
\DP
&
\end{align*}
\\
\begin{flalign*}
&\text{where:}&
\end{flalign*}
\\
\LTValA,\LTValB\GramDef\LTNatZ\BarSep\LTNatS\,\LTValA\BarSep\LTLamVar\,\LTValA\BarSep\LTLamAbs\,\LTValA\BarSep\LTLamApp\,\LTValA\,\LTValB\BarSep\LTLamListNil\BarSep\LTLamListCons\,\LTValA\,\LTValB\\
\begin{aligned}
\LTEnvA{\_}\GramDef\_&\BarSep\LTEnvA{\_}\,\LTTermA\BarSep\LTProj_1\,\LTEnvA{\_}\BarSep\LTProj_2\,\LTEnvA{\_}\BarSep\LTNatIt\,\LTTermA\,\LTTermB\,\LTEnvA{\_}\BarSep\LTLamIt\,\LTTermA\,\LTTermB\,\LTTermC\,\LTEnvA{\_}\BarSep\LTLamListIt\,\LTTermA\,\LTTermB\,\LTEnvA{\_}\\
&\BarSep\LTNatS\,\LTEnvA{\_}\BarSep\LTLamVar\,\LTEnvA{\_}\BarSep\LTLamAbs\,\LTEnvA{\_}\BarSep\LTLamApp\,\LTEnvA{\_}\BarSep\LTLamApp\,\LTValA\,\LTEnvA{\_}\BarSep\LTLamListCons\,\LTEnvA{\_}\BarSep\LTLamListCons\,\LTValA\,\LTEnvA{\_}\\
\end{aligned}
\end{gather*}
\caption{Reductions in system $\LT$}
\label{RedLT}
\end{figure}
Note that we write lists with the most recent element at the end since the addition of an element to a list corresponds to the extension of an applicative context with one more argument. We use iterators rather than recursors for simplicity, recursors can nevertheless be defined using iterators and pairs. Our operational semantics involves values and evaluation contexts because we will interpret system $\LTbbc$ in a domain semantics that allows a priori non terminating computations, even though one can prove that all programs of system $\LTbbc$ terminate. This is the reason why our iterators first reduce their last argument to a value, so if this argument does not terminate then the iterator does not terminate either. Using iterators we define a generalized application $\LTLamListApp\Def\lambda\LTVarA.\LTLamListIt\,\LTVarA\,\LTLamApp:\LTTypeLam\to\LTTypeLamList\to\LTTypeLam$ such that if $\LTTermA,\LTTermB_0,\ldots,\LTTermB_{n-1}\in\LTTypeLam$:
\begin{equation*}
\LTLamListApp\,\LTTermA\left(\LTLamListCons\left(\LTLamListCons\left(\ldots\LTLamListCons\,\LTLamListNil\,\LTTermB_0\ldots\right)\LTTermB_{n-1}\right)\right)\LTRed^*\LTLamApp\left(\ldots\left(\LTLamApp\,\LTTermA\,\LTValA_0\right)\ldots\right)\LTValA_{n-1}
\end{equation*}
where $\LTTermB_i\LTRed^*\LTValA_i$. We can also define a shift operation on lists of terms $\LTLamListShift:\LTTypeLamList\to\LTTypeLamList$ implementing the operation $\FormalShift{}$ described in section~\ref{Subst}. Finally, we can define a substitution operation $\LTLamSubst:\LTTypeLam\to\LTTypeNat\to\LTTypeLamList\to\LTTypeLam$ implementing the operation $\_\FormalSubst{\_}{\_}$ described in section~\ref{Subst}.\par
We now extend $\LT$ with the BBC functional that interprets the axiom scheme of comprehension on $\lambda$-terms. The lack of a canonical ordering on $\lambda$-terms is our main motivation for choosing the BBC functional rather than modified bar recursion. The BBC functional builds a partial function on $\LTTypeLam$ piece by piece, so we need a type for such partial functions. We encode them in system $\LT$ as total functions to the type of pairs of a natural number and a value, with the convention that the function is defined when the natural number reduces to $\LTNatZ$. The type of partial functions to $\LTTypeA$ is therefore defined as:
\begin{equation*}
\LTTypeFinFun{\LTTypeA}\Def\LTTypeLam\to\LTTypeNat\times\LTTypeA
\end{equation*}
with the convention that $\LTTermA:\LTTypeFinFun{\LTTypeA}$ is defined at $\LTTermB:\LTTypeLam$ if $\LTProj_1\left(\LTTermA\,\LTTermB\right)\LTRed^*\LTNatZ$, in which case its value is $\LTProj_2\left(\LTTermA\,\LTTermB\right)$, and undefined otherwise.\par
In order to define the empty function we need to have a canonical element $\LTcan_\LTTypeA$ at every type $\LTTypeA$, defined inductively as follows:
\begin{flalign*}
&\LTcan_\LTTypeNat\Def\LTNatZ
&
&\LTcan_\LTTypeLam\Def\LTLamVar\,\LTNatZ
&
&\LTcan_\LTTypeLamList\Def\LTLamListNil
&
&\LTcan_{\LTTypeA\to\LTTypeB}\Def\lambda\_.\LTcan_\LTTypeB
&
&\LTcan_{\LTTypeA\times\LTTypeB}\Def\LTPair{\LTcan_\LTTypeA}{\LTcan_\LTTypeB}
\end{flalign*}
The strict partial function with empty support is defined as follows:
\begin{equation*}
\LTFinFunEmpty\Def\LTLamIt\left(\lambda\_.\LTPair{\LTNatS\,\LTNatZ}{\LTcan_\LTTypeA}\right)\left(\lambda\_.\LTPair{\LTNatS\,\LTNatZ}{\LTcan_\LTTypeA}\right)\left(\lambda\_\_.\LTPair{\LTNatS\,\LTNatZ}{\LTcan_\LTTypeA}\right):\LTTypeFinFun{\LTTypeA}
\end{equation*}
This partial function is such that $\LTFinFunEmpty\,\LTTermA\LTRed^*\LTPair{\LTNatS\,\LTNatZ}{\LTcan_\LTTypeA}$ for any (terminating) $\LTTermA:\LTTypeLam$, that is, $\LTFinFunEmpty$ is the everywhere undefined function. Again, even though all programs of system $\LTbbc$ terminate, we have to take into account non terminating arguments so it is important that $\LTFinFunEmpty$ is strict, that is, $\LTFinFunEmpty\,\LTTermA$ terminates only if $\LTTermA$ terminates. The strictness of $\LTFinFunEmpty$ will be necessary in the proof of correctness of the BBC functional in section~\ref{Comprehension}.\par
The extension/overwrite of a partial function $\LTTermA:\LTTypeFinFun{\LTTypeA}$ with a value $\LTTermC:\LTTypeA$ at input $\LTTermB:\LTTypeLam$ requires decidability on type $\LTTypeLam$, that is, the existence of a term $\LTTermD:\LTTypeLam\to\LTTypeLam\to\LTTypeNat$ such that for any $\LTTermE,\LTTermF:\LTTypeLam$, $\LTTermD\,\LTTermE\,\LTTermF\LTRed^*\LTNatZ$ if and only if $\LTTermE\LTRed^*\LTValA$ and $\LTTermF\LTRed^*\LTValA$ for some $\LTValA$. Such a term can indeed be defined in system $\LT$ using $\LTNatIt$, $\LTLamIt$ and pairs, so we can define $\LTFinFunExtend{\LTTermA}{\LTTermB}{\LTTermC}:\LTTypeFinFun{\LTTypeA}$ that behaves on terminating arguments as follows:
\begin{equation*}
\left(\LTFinFunExtend{\LTTermA}{\LTTermB}{\LTTermC}\right)\LTTermD\LTRed^*\left\{\begin{aligned}&\LTPair{\LTNatZ}{\LTTermC}&&\text{if }\LTTermB\LTRed^*\LTValA\text{ and }\LTTermD\LTRed^*\LTValA\text{ for some }\LTValA\\&\LTTermA\,\LTTermD&&\text{otherwise}\end{aligned}\right.
\end{equation*}\par
The BBC functional also combines a partial function $\LTTermA:\LTTypeFinFun{\LTTypeA}$ with a total function $\LTTermB:\LTTypeLam\to\LTTypeA$ to form the total function $\LTFinFunComplete{\LTTermA}{\LTTermB}:\LTTypeLam\to\LTTypeA$ that takes values from $\LTTermA$ when they are defined and values from $\LTTermB$ everywhere else. This combination is defined as:
\begin{equation*}
\LTFinFunComplete{\LTTermA}{\LTTermB}\Def\lambda\LTVarA.\LTNatIt\left(\LTProj_2\left(\LTTermA\,\LTVarA\right)\right)\left(\lambda\_.\LTTermB\,\LTVarA\right)\left(\LTProj_1\left(\LTTermA\,\LTVarA\right)\right)
\end{equation*}
and behaves as follows on terminating arguments:
\begin{equation*}
\left(\LTFinFunComplete{\LTTermA}{\LTTermB}\right)\LTTermC\LTRed^*\left\{\begin{aligned}&\LTProj_2\left(\LTTermA\,\LTTermC\right)&&\text{if }\LTProj_1\left(\LTTermA\,\LTTermC\right)\LTRed^*\LTNatZ\\&\LTTermB\,\LTTermC&&\text{otherwise}\end{aligned}\right.
\end{equation*}
With these new definitions we can now introduce the BBC functional:
\begin{equation*}
\LTBBC:\left(\left(\LTTypeA\to\LTTypeNat\right)\to\LTTypeA\right)\to\left(\left(\LTTypeLam\to\LTTypeA\right)\to\LTTypeNat\right)\to\LTTypeFinFun{\LTTypeA}\to\LTTypeNat
\end{equation*}
together with its reduction rule:
\begin{equation*}
\LTBBC\,\LTTermA\,\LTTermB\,\LTTermC\LTRed\LTTermB\left(\LTFinFunComplete{\LTTermC}{\lambda\LTVarB.\LTTermA\left(\lambda\LTVarC.\LTBBC\,\LTTermA\,\LTTermB\left(\LTFinFunExtend{\LTTermC}{\LTVarB}{\LTVarC}\right)\right)}\right)
\end{equation*}
System $\LTbbc$ is obtained by extending system $\LT$ with the constant $\LTBBC$ together with its reduction rule.
\subsection{Continuous semantics of system \texorpdfstring{$\LTbbc$}{lambda-T-bbc}}
\label{Semantics}
Interpreting the language in a model containing non-computable elements is a convenient way of proving correctness of the BBC functional. We follow this route and consider realizers that are elements of a continuous model of system $\LTbbc$ rather than mere programs. Since system $\LTbbc$ can be seen as a subset of PCF where recursion is restricted to primitive recursion and the BBC functional, it is natural to consider a domain-theoretic semantics. More precisely, we define a denotational semantics of system $\LTbbc$ in complete partial orders. We recall some basic definitions:
\begin{defi}[cpo]
A partial order $\left(\CPOA,\leq\right)$ is a complete partial order (cpo) if:
\begin{itemize}
\item$\CPOA$ has a least element $\CPObot$
\item Every directed subset $\CPOdir$ of $\CPOA$ has a least upper bound $\CPOsup\CPOdir$, where $\CPOdir\subseteq\CPOA$ is directed if it is non-empty and:
\begin{equation*}
\forall\CPOelA\in\CPOdir\,\forall\CPOelB\in\CPOdir\,\exists\CPOelC\in\CPOdir\left(\CPOelA\leq\CPOelC\wedge\CPOelB\leq\CPOelC\right)
\end{equation*}
\end{itemize}
\end{defi}
\begin{defi}[continuous function]
If $\left(\CPOA,\leq\right)$ and $\left(\CPOB,\leq\right)$ are cpos, a function $\CPOelA:\CPOA\to\CPOB$ is continuous if for every directed subset $\CPOdir$ of $\CPOA$, $\CPOelA\left(\CPOdir\right)$ is directed and:
\begin{equation*}
\CPOelA\left(\CPOsup\CPOdir\right)=\CPOsup\CPOelA\left(\CPOdir\right)
\end{equation*}
\end{defi}
\begin{defi}[product of cpos]
If $\left(\CPOA,\leq\right)$ and $\left(\CPOB,\leq\right)$ are cpos, then $\CPOA\times\CPOB$ is a cpo for the pointwise ordering:
\begin{equation*}
\left(\CPOelA,\CPOelB\right)\leq\left(\CPOelA',\CPOelB'\right)\quad\Longleftrightarrow\quad\CPOelA\leq\CPOelA'\wedge\CPOelB\leq\CPOelB'
\end{equation*}
\end{defi}
The projection functions from $\CPOA\times\CPOB$ to $\CPOA$ and $\CPOB$ will be written $\CPOproj_1$ and $\CPOproj_2$.
\begin{defi}[cpo of continuous functions]
If $\left(\CPOA,\leq\right)$ and $\left(\CPOB,\leq\right)$ are cpos, then the set of continuous functions from $\CPOA$ to $\CPOB$ is a cpo for the pointwise ordering:
\begin{equation*}
\CPOelA\leq\CPOelA'\quad\Longleftrightarrow\quad\forall\CPOelB\in\CPOA,\,\CPOelA\left(\CPOelB\right)\leq\CPOelA'\left(\CPOelB\right)
\end{equation*}
\end{defi}
\begin{defi}[flat cpo]
If $X$ is a set, then $X_\bot=X\cup\SetSuch{\CPObot}{}$ is a cpo for the following ordering:
\begin{equation*}
\CPOelA\leq\CPOelB\quad\Longleftrightarrow\quad\CPOelA=\CPOelB\vee\CPOelA=\bot
\end{equation*}
\end{defi}
The category of cpos and continuous functions is cartesian closed and provides a sound and computationally adequate semantics for PCF where the type of natural numbers is interpreted as $\mathbb{N}_\CPObot$, see e.g.~\cite{AmadioCurien}. We extend this semantics with a type of $\lambda$-terms interpreted as $\Lambda_\CPObot$, a type of lists of $\lambda$-terms interpreted as $\left(\Lambda^*\right)_\CPObot$ (where $\Lambda^*$ denotes the set of finite sequences of $\lambda$-terms) and product types interpreted with the categorical product of cpos. All the constants of system $\LTbbc$ can be interpreted with fixpoints and basic operations on flat domains, therefore the category of cpos and continuous functions forms a model of system $\LTbbc$.\par
We fix some notations. If $\LTTypeA$ is a type of system $\LTbbc$, then $\CPOinterp{\LTTypeA}$ denotes the cpo interpreting $\LTTypeA$. If $\LTVarA_0:\LTTypeA_0,\ldots,\LTVarA_{n-1}:\LTTypeA_{n-1}\Entails\LTTermA:\LTTypeB$  is the conclusion of a typing derivation in system $\LTbbc$ and if $\ValuatA$ is a valuation such that $\ValuatA\left(\LTVarA_i\right)\in\CPOinterp{\LTTypeA_i}$ for each $i$, then $\CPOinterp{\LTTermA}_\ValuatA\in\CPOinterp{\LTTypeB}$ is the denotation of $\LTTermA$ with valuation $\ValuatA$. The category of cpos and continuous functions provides a sound and computationally adequate model for system $\LTbbc$:
\begin{lem}
If $\LTTermA\LTRed\LTTermB$ in system $\LTbbc$ and if $\ValuatA$ is a valuation then:
\begin{equation*}
\CPOinterp{\LTTermA}_\ValuatA=\CPOinterp{\LTTermB}_\ValuatA
\end{equation*}
Moreover, if $\LTTermA:\LTTypeNat$ is a closed term and if $\CPOinterp{\LTTermA}$ is some $\RealNatA\in\mathbb{N}$ then:
\begin{equation*}
\LTTermA\LTRed^*\LTNatS^\RealNatA\,\LTNatZ
\end{equation*}
\end{lem}
These results are proved using standard techniques for continuous models of PCF, see e.g.~\cite{AmadioCurien}. In system $\LTbbc$, computational adequacy holds for every basic type but we only need it on the type $\LTTypeNat$ of natural numbers. Finally, we stress that the BBC functional is a total element in this model, and therefore system $\LTbbc$ is a total language: all computations terminate. In particular, if $\LTTermA$ is a closed term in system $\LTbbc$ then $\CPOinterp{\LTTermA}\neq\CPObot$. We do not prove totality of the BBC functional here but the proof is a straightforward simplification of its proof of adequacy (lemma~\ref{RealBarRec}).\par
Finally, we mention a result that will be useful for the proof of adequacy of the BBC functional and is a consequence of the properties of cpos:
\begin{lem}
\label{CPOCont}
Write $\CPOA^{X_\CPObot}$ for the cpo of continuous functions from $X_\CPObot$ to $\CPOA$. If $\CPOelA$ is a continuous function from $\CPOA^{X_\CPObot}$ to $Y_\CPObot$ and if $\CPOelB\in\CPOA^{X_\CPObot}$ is such that $\CPOelA\left(\CPOelB\right)\neq\CPObot$ and $\CPOelB\left(\CPObot\right)=\CPObot$, then there exists a finite set $F\subseteq X$ such that:
\begin{equation*}
\forall\CPOelB'\in\CPOA^{X_\CPObot}\left(\forall\CPOelC\in F\left(\CPOelB'\left(\CPOelC\right)=\CPOelB\left(\CPOelC\right)\right)\Rightarrow\CPOelA\left(\CPOelB'\right)=\CPOelA\left(\CPOelB\right)\right)
\end{equation*}
\end{lem}
\begin{proof}
Define for $F$ finite subset of $X$ the continuous function:
\begin{equation*}
\CPOelB_F\left(\CPOelC\right)=\left\{\begin{aligned}&\CPOelB\left(\CPOelC\right)&&\text{ if }\CPOelC\in F\\&\CPObot&&\text{ otherwise}\end{aligned}\right.
\end{equation*}
Then $\SetSuch{\CPOelB_F}{F\subseteq X\text{ finite}}$ is directed and:
\begin{equation*}
\CPOelB=\CPOsup\SetSuch{\CPOelB_F}{F\subseteq X\text{ finite}}
\end{equation*}
so the continuity of $\CPOelA$ implies that:
\begin{equation*}
\CPOelA\left(\CPOelB\right)=\CPOsup\SetSuch{\CPOelA\left(\CPOelB_F\right)}{F\subseteq X\text{ finite}}
\end{equation*}
By definition of the order on $Y_\CPObot$, this means that there must exist some finite $F\subseteq X$ such that $\CPOelA\left(\CPOelB_F\right)=\CPOelA\left(\CPOelB\right)$. If $\CPOelB'$ is such that $\CPOelB'\left(\CPOelC\right)=\CPOelB\left(\CPOelC\right)$ for every $\CPOelC\in F$, then $\CPOelB'\geq\CPOelB_F$ so $\CPOelA\left(\CPOelB'\right)\geq\CPOelA\left(\CPOelB_F\right)=\CPOelA\left(\CPOelB\right)$. Finally, since $\CPOelA\left(\CPOelB\right)\neq\CPObot$ we obtain $\CPOelA\left(\CPOelB'\right)=\CPOelA\left(\CPOelB\right)$.
\end{proof}
\section{Realizability}
\label{Realizability}
This section contains the main contribution of our work: a translation of system F into system $\LTbbc$ through a bar recursive interpretation of second-order arithmetic. Our realizability model follows the lines of Kreisel's modified realizability. A plain Dialectica interpretation of our logic would not be possible because the interpretation of contraction ($\LogFormA\Rightarrow\LogFormA\wedge\LogFormA$) requires the decidability of quantifier-free formulas, which we do not have in our logic: $\LogBoolIn{\LogTermA}{\RealPredA}$ is undecidable when $\RealPredA$ is the set of normalizing terms for example. However, the Diller-Nahm interpretation~\cite{DillerNahm} circumvents this difficulty and provides a finite set of potential witnesses with the property that one of them is correct. Since our translation relies on the extraction of an upper bound, taking the maximum of the set of potential witnesses would give an alternative translation of system F into system $\LTbbc$ that we plan to investigate.\par
We first define a syntactic mapping from our logic into system $\LTbbc$ and the realizability values of formulas. Then we show how we interpret classical logic, the axiom scheme of comprehension and the instantiation of a set variable with an arbitrary formula. Finally, we define the interpretation of the normalization proof of section~\ref{NormProof} and derive our translation from it.
\subsection{Mapping the logic into system \texorpdfstring{$\LTbbc$}{lambda-T-bbc}}
\label{Mapping}
Our interpretation is in the style of Kreisel's modified realizability in which realizers are typed. In our setting, we associate to each formula $\LogFormA$ a type $\LTinterp{\LogFormA}$ of system $\LTbbc$, so that realizers of $\LogFormA$ are elements of the cpo $\CPOinterp{\LTinterp{A}}$ interpreting $\LTinterp{\LogFormA}$. Moreover, our realizers will manipulate natural numbers, terms and lists of terms of our logic so we also associate to each element $\LogNatA$ (respectively $\LogTermA$, $\LogTermListA$) of the logic a program $\LTinterp{\LogNatA}:\LTTypeNat$ (respectively $\LTinterp{\LogTermA}:\LTTypeLam$, $\LTinterp{\LogTermListA}:\LTTypeLamList$).\par
The mapping $\LTinterp{\_}$ on formulas is defined as follows:
\begin{gather*}
\begin{align*}
&\LTinterp{\left(\LogFormA\Rightarrow\LogFormB\right)}=\LTinterp{\LogFormA}\to\LTinterp{\LogFormB}
&
&\LTinterp{\left(\LogFormA\wedge\LogFormB\right)}=\LTinterp{\LogFormA}\times\LTinterp{\LogFormB}
&
&\LTinterp{\left(\forall\LogPredVarA\,\LogFormA\right)}=\LTinterp{\left(\forall\LogBoolVarA\,\LogFormA\right)}=\LTinterp{\LogFormA}
&
\end{align*}
\\
\begin{align*}
&\LTinterp{\LogBoolA}=\LTTypeNat
&
&\LTinterp{\left(\forall\LogNatVarA\,\LogFormA\right)}=\LTTypeNat\to\LTinterp{\LogFormA}
&
&\LTinterp{\left(\forall\LogTermVarA\,\LogFormA\right)}=\LTTypeLam\to\LTinterp{\LogFormA}
&
&\LTinterp{\left(\forall\LogTermListVarA\,\LogFormA\right)}=\LTTypeLamList\to\LTinterp{\LogFormA}
&
\end{align*}
\end{gather*}
Atomic formulas are mapped to the type $\LTTypeNat$ of natural numbers because we want to extract natural numbers (bounds on the numbers of reduction steps for reaching a normal form) from proofs in classical logic. We perform the standard technique of defining the set of realizers of the false formula as a well-chosen subset of the natural numbers. This technique was already used in~\cite{BerardiBezemCoquand} and is the computational counterpart of Friedman's $A$-translation. The sorts of the logic are divided in two groups. The sorts of natural numbers, terms and lists of terms called computational: a realizer of a quantification on a computational sort takes an element of that sort as input and builds a realizer of the instantiation of the formula with that element. Conversely, the sorts of sets and booleans are not computational: a realizer of a quantification on a non-computational sort must be uniform, in the sense that it must realize all the instantiations regardless of the element the formula is instantiated with.\par
Since the type associated to a formula does not depend on the particular first-order elements in the formula, the type associated to an instance of a 1-formula $\LogFormA\left(\_,\ldots,\_\right)$ does not depend on the instance and will simply be written $\LTinterp{\LogFormA}$. On the other hand, the type associated to a 2-formula depends on its particular instance.\par
We now define the mapping from elements of a computational sort in our logic to system $\LTbbc$ programs of the corresponding type. For simplicity and without loss of generality we suppose that the variables $\LogNatVarA$, $\LogTermVarA$ and $\LogTermListVarA$ of the logic are also variables of system $\LTbbc$ with respective types $\LTTypeNat$, $\LTTypeLam$ and $\LTTypeLamList$. A first-order element $\LogNatA$, $\LogTermA$ or $\LogTermListA$ of the logic is then mapped to a program $\LTinterp{\LogNatA}$, $\LTinterp{\LogTermA}$ or $\LTinterp{\LogTermListA}$ with the same set of variables as follows:
\begin{gather*}
\begin{flalign*}
&\LTinterp{\LogNatVarA}=\LogNatVarA
&
&\LTinterp{\LogTermVarA}=\LogTermVarA
&
&\LTinterp{\LogTermListVarA}=\LogTermListVarA
&
&\LTinterp{\LogNatZ}=\LTNatZ
&
&\LTinterp{\left(\LogNatS{\LogNatA}\right)}=\LTNatS\,\LTinterp{\LogNatA}
&
&\LTinterp{\LogTermVar{\LogNatA}}=\LTLamVar\,\LTinterp{\LogNatA}
&
&\LTinterp{\left(\LogTermLam{\LogTermA}\right)}=\LTLamAbs\,\LTinterp{\LogTermA}
\end{flalign*}
\\
\begin{flalign*}
&\LTinterp{\left(\LogTermApp{\LogTermA}{\LogTermListA}\right)}=\LTLamListApp\,\LTinterp{\LogTermA}\,\LTinterp{\LogTermListA}
&
&\LTinterp{\LogTermListNil}=\LTLamListNil
&
&\LTinterp{\LogTermListCons{\LogTermListA}{\LogTermA}}=\LTLamListCons\,\LTinterp{\LogTermListA}\,\LTinterp{\LogTermA}
&
&\LTinterp{\left(\LogTermSubst{\LogTermA}{\LogTermListA}\right)}=\LTLamSubst\,\LTinterp{\LogTermA}\,\LTNatZ\,\LTinterp{\LogTermListA}
\end{flalign*}
\end{gather*}\par
\subsection{Realizability values}
We now define the realizability model that will ensure the correctness of our translation from system F to system $\LTbbc$. We define for each formula $\LogFormA$ the set $\RealVal{\LogFormA}\subseteq\CPOinterp{\LTinterp{\LogFormA}}$ of realizers of $\LogFormA$, where $\LTinterp{\_}$ is the mapping from formulas to types of system $\LTbbc$ defined in section~\ref{Mapping} and $\CPOinterp{\_}$ is the interpretation of system $\LTbbc$ in cpos defined in section~\ref{Semantics}.\par
Because $\LogFormA$ may contain free variables, its realizability value $\RealVal{\LogFormA}$ depends on a valuation, that is, a function $\ValuatA$ on the free variables of $\LogFormA$ such that:
\begin{align*}
&\ValuatA\left(\LogNatVarA\right)\in\mathbb{N}
&
&\ValuatA\left(\LogTermVarA\right)\in\Lambda
&
&\ValuatA\left(\LogTermListVarA\right)\in\Lambda^*
&
&\ValuatA\left(\LogPredVarA\right)\in\mathcal{P}\left(\Lambda\right)
&
&\ValuatA\left(\LogBoolVarA\right)\in\SetSuch{\RealBoolT;\RealBoolF}{}
\end{align*}
where $\Lambda^*$ denotes the set of finite sequences of $\lambda$-terms. Since $\mathbb{N}\subseteq\mathbb{N}_\CPObot=\CPOinterp{\LTTypeNat}$, $\Lambda\subseteq\Lambda_\CPObot=\CPOinterp{\LTTypeLam}$ and $\Lambda^*\subseteq\left(\Lambda^*\right)_\CPObot=\CPOinterp{\LTTypeLamList}$, we have that for any term $\LogNatA$, $\LogTermA$ or $\LogTermListA$ appearing in $\LogFormA$, a valuation on $\LogFormA$ is in particular a valuation on $\LTinterp{\LogNatA}$, $\LTinterp{\LogTermA}$ or $\LTinterp{\LogTermListA}$ in the sense of cpos, where $\LTinterp{\_}$ is the mapping from computational elements of the logic to programs of system $\LTbbc$ defined in section~\ref{Mapping}. Therefore, $\CPOinterp{\LTinterp{\LogNatA}}_\ValuatA\in\CPOinterp{\LTTypeNat}$, $\CPOinterp{\LTinterp{\LogTermA}}_\ValuatA\in\CPOinterp{\LTTypeLam}$ and $\CPOinterp{\LTinterp{\LogTermListA}}_\ValuatA\in\CPOinterp{\LTTypeLamList}$ are well-defined. Moreover, $\CPOinterp{\LTinterp{\LogNatA}}_\ValuatA\in\mathbb{N}$, $\CPOinterp{\LTinterp{\LogTermA}}_\ValuatA\in\Lambda$ and $\CPOinterp{\LTinterp{\LogTermListA}}_\ValuatA\in\Lambda^*$: they are different from $\CPObot$.\par
As explained in the previous section, we fix the set of realizers of false atomic formulas to a well-chosen set of natural numbers so we can extract computational content from proofs in classical logic. For now this set is a parameter of our realizability model:
\begin{equation*}
\RealBot\subseteq\mathbb{N}
\end{equation*}
From that parameter, we define the realizability value $\RealVal{\LogFormA}_\ValuatA\subseteq\CPOinterp{\LTinterp{\LogFormA}}$ of a formula $\LogFormA$ with valuation $\ValuatA$ in figure~\ref{RealVal}.
\begin{figure*}
\begin{gather*}
\begin{flalign*}
&\begin{aligned}
\RealVal{\LogBoolT}_\ValuatA&=\mathbb{N}_\CPObot\\
\RealVal{\LogBoolF}_\ValuatA&=\RealBot
\end{aligned}
&
\RealVal{\LogNNorm{\LogTermA}{\LogNatA}}_\ValuatA&=\left\{\begin{aligned}&\mathbb{N}_\CPObot&&\begin{aligned}[t]\text{if }\CPOinterp{\LTinterp{\LogTermA}}_\ValuatA\text{ can reduce for }\CPOinterp{\LTinterp{\LogNatA}}_\ValuatA\text{ steps of weak head}\\\text{reduction without reaching a normal form}\end{aligned}\\&\RealBot&&\text{otherwise}\end{aligned}\right.
\end{flalign*}
\\
\begin{flalign*}
\RealVal{\LogBoolVarA}_\ValuatA&=\left\{\begin{aligned}&\mathbb{N}_\CPObot&&\text{if }\ValuatA\left(\LogBoolVarA\right)=\RealBoolT\\&\RealBot&&\text{if }\ValuatA\left(\LogBoolVarA\right)=\RealBoolF\end{aligned}\right.
&
&\begin{aligned}
\RealVal{\LogFormA\Rightarrow\LogFormB}_\ValuatA&=\SetSuch{\CPOelA\in\CPOinterp{\LTinterp{\LogFormA}\to\LTinterp{\LogFormB}}}{\forall\CPOelB\in\RealVal{\LogFormA}_\ValuatA,\CPOelA\left(\CPOelB\right)\in\RealVal{\LogFormB}_\ValuatA}
\\
\RealVal{\LogFormA\wedge\LogFormB}_\ValuatA&=\SetSuch{\left(\CPOelA,\CPOelB\right)\in\CPOinterp{\LTinterp{\LogFormA}\times\LTinterp{\LogFormB}}}{\CPOelA\in\RealVal{\LogFormA}_\ValuatA\wedge\CPOelB\in\RealVal{\LogFormB}_\ValuatA}
\end{aligned}
\end{flalign*}
\\
\begin{flalign*}
\RealVal{\LogBoolIn{\LogTermA}{\LogPredVarA}}_\ValuatA&=\left\{\begin{alignedat}{2}&\mathbb{N}_\CPObot&&\text{if}\,\CPOinterp{\LTinterp{\LogTermA}}_\ValuatA\in\ValuatA\left(\LogPredVarA\right)\\&\RealBot&&\text{if}\,\CPOinterp{\LTinterp{\LogTermA}}_\ValuatA\notin\ValuatA\left(\LogPredVarA\right)\end{alignedat}\right.
&
\RealVal{\forall\LogNatVarA\,\LogFormA}_\ValuatA&=\SetSuch{\CPOelA\in\CPOinterp{\LTTypeNat\to\LTinterp{\LogFormA}}}{\forall\RealNatA\in\mathbb{N},\CPOelA\left(\RealNatA\right)\in\RealVal{\LogFormA}_{\ValuatA\uplus\SetSuch{\LogNatVarA\mapsto\RealNatA}{}}}
\end{flalign*}
\\
\begin{flalign*}
\RealVal{\forall\LogPredVarA\,\LogFormA}_\ValuatA&=\bigcap_{\RealPredA\in\mathcal{P}\left(\Lambda\right)}\RealVal{\LogFormA}_{\ValuatA\uplus\SetSuch{\LogPredVarA\mapsto\RealPredA}{}}
&
\RealVal{\forall\LogTermVarA\,\LogFormA}_\ValuatA&=\SetSuch{\CPOelA\in\CPOinterp{\LTTypeLam\to\LTinterp{\LogFormA}}}{\forall\RealTermA\in\Lambda,\CPOelA\left(\RealTermA\right)\in\RealVal{\LogFormA}_{\ValuatA\uplus\SetSuch{\LogTermVarA\mapsto\RealTermA}{}}}
\\
\RealVal{\forall\LogBoolVarA\,\LogFormA}_\ValuatA&=\bigcap_{\RealBoolA\in\SetSuch{\RealBoolT;\RealBoolF}{}}\RealVal{\LogFormA}_{\ValuatA\uplus\SetSuch{\LogBoolVarA\mapsto\RealBoolA}{}}
&
\RealVal{\forall\LogTermListVarA\,\LogFormA}_\ValuatA&=\SetSuch{\CPOelA\in\CPOinterp{\LTTypeLamList\to\LTinterp{\LogFormA}}}{\forall\RealTermListA\in\Lambda^*,\CPOelA\left(\RealTermListA\right)\in\RealVal{\LogFormA}_{\ValuatA\uplus\SetSuch{\LogTermListVarA\mapsto\RealTermListA}{}}}
\end{flalign*}
\end{gather*}
\caption{Realizability values}
\label{RealVal}
\end{figure*}
The realizability value of a boolean formula $\LogBoolA$ is either the whole set $\CPOinterp{\LTinterp{\LogBoolA}}=\mathbb{N}_\CPObot$ or the parameter $\RealBot$, which is a standard definition in realizability models for classical logic. In the definition of $\RealVal{\LogNNorm{\LogTermA}{\LogNatA}}_\ValuatA$, remember that $\CPOinterp{\LTinterp{\LogTermA}}_\ValuatA\in\Lambda$ and $\CPOinterp{\LTinterp{\LogNatA}}_\ValuatA\in\mathbb{N}$ (they are not $\CPObot$), so the definition is correct. The realizability values for universally quantified formulas depend on whether the sort of the quantified variable is computational or not. In the computational case, the realizer takes as input the element the formula is instantiated with, while in the non-computational case the realizer does not depend on the particular value the formula is instantiated with: the realizer is uniform. Realizability values of implications and conjunctions are standard.\par
As an alternative to valuations, we will also use terms and formulas with parameters. This means that we syntactically substitute elements of $\CPOinterp{\LTTypeA}$ for free variables of type $\LTTypeA$ in the interpretation of terms of system $\LTbbc$, and elements of $\mathbb{N}$, $\Lambda$, $\Lambda^*$, $\mathcal{P}\left(\Lambda\right)$ and $\SetSuch{\RealBoolT;\RealBoolF}{}$ for free variables of the corresponding sort in the realizability values of formulas. For example if $\CPOelA\in\CPOinterp{\LTTypeLam}$ then we can write $\CPOinterp{\lambda\LTVarA.\LTLamApp\,\CPOelA\,\LTVarA}$ instead of $\CPOinterp{\lambda\LTVarA.\LTLamApp\,\LTVarB\,\LTVarA}_{\SetSuch{\LTVarB\mapsto\CPOelA}{}}$, and we can write $\RealVal{\forall\LogTermVarA\,\LogNNorm{\LogTermVarA}{\LogNatS\,7}}$ for $\RealVal{\forall\LogTermVarA\,\LogNNorm{\LogTermVarA}{\LogNatS\,\LogNatVarA}}_{\SetSuch{\LogNatVarA\mapsto7}{}}$. A closed element with parameters is an element with parameters that does not have any free variables anymore.
\subsection{Classical logic}
\label{ClassicalLogic}
As explained in section~\ref{Logic}, we work in the target of G\"odel's negative translation so that classical principles can be realized. In particular, we can define realizers of double-negation elimination by induction on formulas:
\begin{gather*}
\begin{flalign*}
\LTdne_\LogBoolA&=\lambda\LTVarA.\LTVarA\left(\lambda\LTVarB.\LTVarB\right)
&
\LTdne_{\forall\LogBoolVarA\,\LogFormA}&=\LTdne_{\forall\LogPredVarA\,\LogFormA}=\LTdne_{\LogFormA}
\\
\LTdne_{\forall\LogRelVarA\,\LogFormA}&=\lambda\LTVarA\LogRelVarA.\LTdne_{\LogFormA}\left(\lambda\LTVarB.\LTVarA\left(\lambda\LTVarC.\LTVarB\left(\LTVarC\,\LogRelVarA\right)\right)\right)
&
\LTdne_{\LogFormA\Rightarrow\LogFormB}&=\lambda\LTVarA\LTVarB.\LTdne_{\LogFormB}\left(\lambda\LTVarC.\LTVarA\left(\lambda\LTVarD.\LTVarC\left(\LTVarD\,\LTVarB\right)\right)\right)
\end{flalign*}
\\
\LTdne_{\LogFormA\wedge\LogFormB}=\lambda\LTVarA.\LTPair{\LTdne_{\LogFormA}\left(\lambda\LTVarB.\LTVarA\left(\lambda\LTVarC.\LTVarB\left(\LTProj_1\,\LTVarC\right)\right)\right)}{\LTdne_{\LogFormB}\left(\lambda\LTVarB.\LTVarA\left(\lambda\LTVarC.\LTVarB\left(\LTProj_2\,\LTVarC\right)\right)\right)}
\end{gather*}
where $\LogRelVarA$ ranges over variables of a computational sort: $\LogNatVarA$, $\LogTermVarA$ and $\LogTermListVarA$. These terms indeed realize double-negation elimination:
\begin{lem}
If $\LogFormA$ is a closed formula with parameters then:
\begin{equation*}
\CPOinterp{\LTdne_{\LogFormA}}\in\RealVal{\neg\neg\LogFormA\Rightarrow\LogFormA}
\end{equation*}
\end{lem}
\begin{proof}
By induction:
\begin{itemize}
\item$\LogBoolA$: since by definition $\RealVal{\LogBoolA}$ is either $\RealVal{\LogBoolT}$ or $\RealVal{\LogBoolF}$, we only have to check these two cases:
\begin{itemize}
\item$\CPOinterp{\lambda\LTVarA.\LTVarA\left(\lambda\LTVarB.\LTVarB\right)}\in\RealVal{\neg\neg\LogBoolT\Rightarrow\LogBoolT}$: let $\CPOelA\in\RealVal{\left(\LogBoolT\Rightarrow\LogBoolF\right)\Rightarrow\LogBoolF}$. We show $\CPOinterp{\CPOelA\left(\lambda\LTVarB.\LTVarB\right)}\in\RealVal{\LogBoolT}$, but this is immediate since $\RealVal{\LogBoolT}=\mathbb{N}_\CPObot$
\item$\CPOinterp{\lambda\LTVarA.\LTVarA\left(\lambda\LTVarB.\LTVarB\right)}\in\RealVal{\neg\neg\LogBoolF\Rightarrow\LogBoolF}$: let $\CPOelA\in\RealVal{\left(\LogBoolF\Rightarrow\LogBoolF\right)\Rightarrow\LogBoolF}$. We show $\CPOinterp{\CPOelA\left(\lambda\LTVarB.\LTVarB\right)}\in\RealVal{\LogBoolF}$, which is true because $\CPOinterp{\lambda\LTVarB.\LTVarB}\in\RealVal{\LogBoolF\Rightarrow\LogBoolF}$
\end{itemize}
\item$\forall\LogBoolVarA\,\LogFormA$ or $\forall\LogPredVarA\,\LogFormA$: immediate by induction hypothesis.
\item$\forall\LogRelVarA\,\LogFormA$: we only prove the case $\LogRelVarA\equiv\LogNatVarA$. Let $\CPOelA\in\RealVal{\neg\neg\forall\LogNatVarA\,\LogFormA}$ and let $\RealNatA\in\CPOinterp{\LTTypeNat}$. By induction hypothesis it is sufficient to show that $\CPOinterp{\lambda\LTVarB.\CPOelA\left(\lambda\LTVarC.\LTVarB\left(\LTVarC\,\RealNatA\right)\right)}\in\RealVal{\neg\neg\LogFormA}_{\SetSuch{\LogNatVarA\mapsto\RealNatA}{}}$. Let $\CPOelB\in\RealVal{\neg\LogFormA}_{\SetSuch{\LogNatVarA\mapsto\RealNatA}{}}$. Since $\CPOelA\in\RealVal{\neg\neg\forall\LogNatVarA\,\LogFormA}$ we are left to prove that $\CPOinterp{\lambda\LTVarC.\CPOelB\left(\LTVarC\,\RealNatA\right)}\in\RealVal{\neg\forall\LogNatVarA\,\LogFormA}$. Indeed, if $\CPOelC\in\RealVal{\forall\LogNatVarA\,\LogFormA}$ then $\CPOelC\left(\RealNatA\right)\in\RealVal{\LogFormA}_{\SetSuch{\LogNatVarA\mapsto\RealNatA}{}}$ so $\CPOelB\left(\CPOelC\left(\RealNatA\right)\right)\in\RealVal{\LogBoolF}$.
\item$\LogFormA\Rightarrow\LogFormB$: let $\CPOelA\in\RealVal{\neg\neg\left(\LogFormA\Rightarrow\LogFormB\right)}$ and let $\CPOelB\in\RealVal{\LogFormA}$. By induction hypothesis it is sufficient to show that $\CPOinterp{\lambda\LTVarC.\CPOelA\left(\lambda\LTVarD.\LTVarC\left(\LTVarD\,\CPOelB\right)\right)}\in\RealVal{\neg\neg\LogFormB}$. Let $\CPOelC\in\RealVal{\neg\LogFormB}$. Since $\CPOelA\in\RealVal{\neg\neg\left(\LogFormA\Rightarrow\LogFormB\right)}$ we are left to prove that $\CPOinterp{\lambda\LTVarD.\CPOelC\left(\LTVarD\,\CPOelB\right)}\in\RealVal{\neg\left(\LogFormA\Rightarrow\LogFormB\right)}$. Indeed, if $\CPOelD\in\RealVal{\LogFormA\Rightarrow\LogFormB}$ then $\CPOelD\left(\CPOelB\right)\in\RealVal{\LogFormB}$ so $\CPOelC\left(\CPOelD\left(\CPOelB\right)\right)\in\RealVal{\LogBoolF}$.
\item$\LogFormA\wedge\LogFormB$: let $\CPOelA\in\RealVal{\neg\neg\left(\LogFormA\wedge\LogFormB\right)}$. By induction hypotheses it is sufficient to prove that $\CPOinterp{\left(\lambda\LTVarB.\CPOelA\left(\lambda\LTVarC.\LTVarB\left(\LTProj_1\,\LTVarC\right)\right)\right)}\in\RealVal{\neg\neg\LogFormA}$ and $\CPOinterp{\left(\lambda\LTVarB.\CPOelA\left(\lambda\LTVarC.\LTVarB\left(\LTProj_2\,\LTVarC\right)\right)\right)}\in\RealVal{\neg\neg\LogFormB}$. The two claims are similar so we prove only the first one. Let $\CPOelB\in\RealVal{\neg\LogFormA}$. Since $\CPOelA\in\RealVal{\neg\neg\left(\LogFormA\wedge\LogFormB\right)}$ it is sufficient to prove that $\CPOinterp{\lambda\LTVarC.\CPOelB\left(\LTProj_1\,\LTVarC\right)}\in\RealVal{\neg\left(\LogFormA\wedge\LogFormB\right)}$. Indeed, if $\CPOelC\in\RealVal{\LogFormA\wedge\LogFormB}$ then $\CPOproj_1\left(\CPOelC\right)\in\RealVal{\LogFormA}$ so $\CPOelB\left(\CPOproj_1\left(\CPOelC\right)\right)\in\RealVal{\LogBoolF}$.\qedhere
\end{itemize}
\end{proof}
Using $\LTdne_\LogFormA$ we define the following term:
\begin{equation*}
\LTexf_{\LogFormA}=\lambda\LTVarA.\LTdne_{\LogFormA}\left(\lambda\_.\LTVarA\right)
\end{equation*}
which immediately realizes the \textit{ex falso quodlibet} principle:
\begin{equation*}
\CPOinterp{\LTexf_{\LogFormA}}\in\RealVal{\LogBoolF\Rightarrow\LogFormA}
\end{equation*}\par
\subsection{Realizing the axiom scheme of comprehension}
\label{Comprehension}
The combination of the axiom of countable choice with classical logic implies the comprehension scheme on natural numbers. Indeed, classical logic provides a proof of $\forall\LogNatVarA\,\exists\LogBoolVarA\,\left(\LogBoolVarA\Leftrightarrow\LogFormA\left(\LogNatVarA\right)\right)$ and then the axiom of countable choice implies $\exists f\,\forall\LogNatVarA\,\left(f\left(\LogNatVarA\right)\Leftrightarrow\LogFormA\left(\LogNatVarA\right)\right)$, where $f$ is a function from natural numbers to booleans. Therefore, we can interpret second-order arithmetic through an encoding of sets of natural numbers as functions from natural numbers to booleans.\par
In the current setting we interpret the comprehension scheme on $\lambda$-terms rather than on natural numbers, so we interpret the following version of the axiom of countable choice:
\begin{equation*}
\forall\LogTermVarA\,\exists\LogBoolVarA\,\LogFormA\left(\LogBoolVarA,\LogTermVarA\right)\Rightarrow\exists\LogPredVarA\,\forall\LogTermVarA\,\LogFormA\left(\LogBoolIn{\LogTermVarA}{\LogPredVarA},\LogTermVarA\right)
\end{equation*}
We actually interpret a weaker version: we define a program that turns an element of $\bigcap_{\RealTermA\in\Lambda}\RealVal{\exists\LogBoolVarA\,\LogFormA\left(\LogBoolVarA,\RealTermA\right)}$ into an element of $\RealVal{\exists\LogPredVarA\,\forall\LogTermVarA\,\LogFormA\left(\LogBoolIn{\LogTermVarA}{\LogPredVarA},\LogTermVarA\right)}$. The difference is that a realizer of $\forall\LogTermVarA\,\exists\LogBoolVarA\,\LogFormA\left(\LogBoolVarA,\LogTermVarA\right)$ takes a term as input (since the sort of terms is computational), while in our particular case we can build a realizer of $\exists\LogBoolVarA\,\left(\LogBoolVarA\Leftrightarrow\LogFormA\left(\LogTermVarA\right)\right)$ that is uniform in $\LogTermVarA$. Because of that, the weaker version is sufficient for the comprehension scheme. The usual BBC functional~\cite{BerardiBezemCoquand} (where the first argument would be of type $\LTTypeLam\to\left(\LTTypeA\to\LTTypeNat\right)\to\LTTypeA$) can in fact realize the stronger version where the left quantification on $\LogTermVarA$ is relativized. Our version is weaker because the first argument is only of type $\left(\LTTypeA\to\LTTypeNat\right)\to\LTTypeA$. It is not clear yet whether the usual version is computationally strictly stronger than our version. Our proof of adequacy is inspired by~\cite{BergerBBCDomains} and uses Zorn's lemma:
\begin{lem}
\label{RealBarRec}
If $\LogFormA\left(\LogBoolA,\LogTermA\right)$ is a closed 1-formula with parameters and $\CPOelA\in\bigcap_{\RealTermA\in\Lambda}\RealVal{\exists\LogBoolVarA\,\LogFormA\left(\LogBoolVarA,\RealTermA\right)}$ then:
\begin{equation*}
\CPOinterp{\lambda\LTVarA.\LTBBC\left(\lambda\LTVarB.\LTexf_{\LogFormA}\left(\CPOelA\,\LTVarB\right)\right)\LTVarA\LTFinFunEmpty}\in\RealVal{\exists\LogPredVarA\,\forall\LogTermVarA\,\LogFormA\left(\LogBoolIn{\LogTermVarA}{\LogPredVarA},\LogTermVarA\right)}
\end{equation*}
\end{lem}
\begin{proof}
Remember that in our logic, $\exists$ is encoded as $\neg\forall\neg$. Let $\CPOelB\in\RealVal{\forall\LogPredVarA\,\neg\forall\LogTermVarA\,\LogFormA\left(\LogBoolIn{\LogTermVarA}{\LogPredVarA},\LogTermVarA\right)}$ and write $\CPOelC=\CPOinterp{\LTBBC\left(\lambda\LTVarB.\LTexf_{\LogFormA}\left(\CPOelA\,\LTVarB\right)\right)\CPOelB}$. We have to prove that:
\begin{equation*}
\CPOelC\left(\CPOinterp{\LTFinFunEmpty}\right)\in\RealVal{\LogBoolF}
\end{equation*}
First, we define the following set:
\begin{equation*}
E=\SetSuch{\CPOelD\in\CPOinterp{\LTTypeFinFun{\LTinterp{\LogFormA}}}}{\begin{aligned}
&\begin{aligned}
\CPOelD\left(\RealTermA\right)\in{}&\SetSuch{0}{}\times\RealVal{\LogFormA\left(\RealBoolT,\RealTermA\right)}\\
{}\cup{}&\SetSuch{0}{}\times\RealVal{\LogFormA\left(\RealBoolF,\RealTermA\right)}\\
{}\cup{}&\SetSuch{1}{}\times\SetSuch{\CPOinterp{\LTcan_{\LTinterp{\LogFormA}}}}{}
\end{aligned}\\
&\CPOelD\left(\CPObot\right)=\CPObot\\
&\CPOelC\left(\CPOelD\right)\notin\RealVal{\LogBoolF}
\end{aligned}}
\end{equation*}
In particular, elements of $E$ are strict partial functions that take values in $\RealVal{\LogFormA\left(\RealBoolT,\RealTermA\right)}$ or $\RealVal{\LogFormA\left(\RealBoolF,\RealTermA\right)}$ where they are defined. We define a partial order $\prec$ on $E$:
\begin{equation*}
\CPOelD\prec\CPOelD'\quad\Longleftrightarrow\quad\forall\RealTermA\in\Lambda\left(\CPOproj_1\left(\CPOelD\left(\RealTermA\right)\right)=0\Rightarrow\CPOelD'\left(\RealTermA\right)=\CPOelD\left(\RealTermA\right)\right)
\end{equation*}
That is, $\CPOelD\prec\CPOelD'$ if $\CPOelD'$ is more defined than $\CPOelD$. We now prove that every non-empty chain of $E$ has an upper bound in $E$ and that $E$ has no maximal element. Therefore by Zorn's lemma the empty set cannot have an upper bound in $E$ and so $E=\emptyset$. In particular $\CPOinterp{\LTFinFunEmpty}\notin E$ and so $\CPOelC\left(\CPOinterp{\LTFinFunEmpty}\right)\in\RealVal{\LogBoolF}$ because $\CPOinterp{\LTFinFunEmpty}$ satisfies all other conditions of $E$ (since $\LTFinFunEmpty$ is strict).
\begin{itemize}
\item Every non-empty chain of $E$ has an upper bound in $E$:\\
Let $C$ be a non-empty chain of $E$ and build $\CPOelD_{max}$ as follows:
\begin{align*}
&\CPOelD_{max}\left(\RealTermA\right)=\left\{\begin{aligned}&\CPOelD\left(\RealTermA\right)\text{ if }\CPOproj_1\left(\CPOelD\left(\RealTermA\right)\right)=0\text{ for some }\CPOelD\in C\\&\left(1,\CPOinterp{\LTcan_{\LTinterp{\LogFormA}}}\right)\text{ otherwise}\end{aligned}\right.\\
&\CPOelD_{max}\left(\CPObot\right)=\CPObot
\end{align*}
This function is well-defined because $C$ is a chain for $\prec$ so if $\CPOelD,\CPOelD'\in C$ are such that $\CPOproj_1\left(\CPOelD\left(\RealTermA\right)\right)=\CPOproj_1\left(\CPOelD'\left(\RealTermA\right)\right)=0$ for some $\RealTermA$, then $\CPOelD\left(\RealTermA\right)=\CPOelD'\left(\RealTermA\right)$. Also, if $\CPOproj_1\left(\CPOelD_{max}\left(\RealTermA\right)\right)=0$ then $\CPOproj_1\left(\CPOelD\left(\RealTermA\right)\right)=0$ for some $\CPOelD\in C$, and therefore:
\begin{equation*}
\CPOproj_2\left(\CPOelD_{max}\left(\RealTermA\right)\right)=\CPOproj_2\left(\CPOelD\left(\RealTermA\right)\right)\in\RealVal{\LogFormA\left(\RealBoolT,\RealTermA\right)}\cup\RealVal{\LogFormA\left(\RealBoolF,\RealTermA\right)}
\end{equation*}
The only non-trivial property left to prove in order to get $\CPOelD_{max}\in E$ is that $\CPOelC\left(\CPOelD_{max}\right)\notin\RealVal{\LogBoolF}$. Suppose $\CPOelC\left(\CPOelD_{max}\right)\in\RealVal{\LogBoolF}$. Then, $\CPOelC\left(\CPOelD_{max}\right)\neq\CPObot$ because $\RealVal{\LogBoolF}=\RealBot\subseteq\mathbb{N}$. We also have $\CPOelD_{max}\left(\CPObot\right)=\CPObot$ so we can apply lemma~\ref{CPOCont} with $X=\Lambda$, $\CPOA=\CPOinterp{\LTTypeNat\times\LTinterp{\LogFormA}}$ and $Y=\mathbb{N}$ to get a finite set $F\subseteq\Lambda$ such that:
\begin{equation*}
\forall\CPOelD\left(\forall\RealTermA\in F\left(\CPOelD\left(\RealTermA\right)=\CPOelD_{max}\left(\RealTermA\right)\right)\Rightarrow\CPOelC\left(\CPOelD\right)=\CPOelC\left(\CPOelD_{max}\right)\right)
\end{equation*}
For every $\RealTermA\in F$ there is some $\CPOelD_\RealTermA\in C$ such that $\CPOelD_\RealTermA\left(\RealTermA\right)=\CPOelD_{max}\left(\RealTermA\right)$. Indeed, if $\CPOproj_1\left(\CPOelD_{max}\left(\RealTermA\right)\right)=0$ then this is by definition of $\CPOelD_{max}$ and if $\CPOproj_1\left(\CPOelD_{max}\left(\RealTermA\right)\right)\neq0$ then any element of $C$ meets the condition (remember that $C$ is non-empty). $C$ is a non-empty chain and $\SetSuch{\CPOelD_\RealTermA}{\RealTermA\in F}$ is a finite subset of $C$ so it has an upper bound $\CPOelD_{\RealTermA_0}\in C$. Then it is easy to see that for any $\RealTermA\in F$, $\CPOelD_{\RealTermA_0}\left(\RealTermA\right)=\CPOelD_{max}\left(\RealTermA\right)$. Therefore $\CPOelC\left(\CPOelD_{max}\right)=\CPOelC\left(\CPOelD_{\RealTermA_0}\right)$, but $\CPOelC\left(\CPOelD_{\RealTermA_0}\right)\notin\RealVal{\LogBoolF}$ since $\CPOelD_{\RealTermA_0}\in C\subseteq E$, hence the contradiction.
\item $E$ has no maximal element:\\
Suppose for the sake of contradiction that $\CPOelD$ is some maximal element of $E$. By definition of the reduction rule for the BBC functional, we have the following equation:
\begin{equation*}
\CPOinterp{\CPOelC\,\CPOelD}=\CPOinterp{\CPOelB\left(\LTFinFunComplete{\CPOelD}{\lambda\LTVarB.\LTexf_{\LogFormA}\left(\CPOelA\left(\lambda\LTVarC.\CPOelC\left(\LTFinFunExtend{\CPOelD}{\LTVarB}{\LTVarC}\right)\right)\right)}\right)}
\end{equation*}
Let $\RealPredA=\SetSuch{\RealTermA\in\Lambda}{\CPOproj_2\left(\CPOelD\left(\RealTermA\right)\right)\in\RealVal{\LogFormA\left(\RealBoolT,\RealTermA\right)}}$. Since we have $\CPOelB\in\RealVal{\neg\forall\LogTermVarA\,\LogFormA\left(\LogBoolIn{\LogTermVarA}{\RealPredA},\LogTermVarA\right)}$ and $\CPOelC\left(\CPOelD\right)\notin\RealVal{\LogBoolF}$, we get:
\begin{equation*}
\CPOinterp{\LTFinFunComplete{\CPOelD}{\lambda\LTVarB.\LTexf_{\LogFormA}\left(\CPOelA\left(\lambda\LTVarC.\CPOelC\left(\LTFinFunExtend{\CPOelD}{\LTVarB}{\LTVarC}\right)\right)\right)}}\notin\RealVal{\forall\LogTermVarA\,\LogFormA\left(\LogBoolIn{\LogTermVarA}{\RealPredA},\LogTermVarA\right)}
\end{equation*}
Therefore there is some $\RealTermA\in\Lambda$ such that:
\begin{equation*}
\CPOinterp{\left(\LTFinFunComplete{\CPOelD}{\lambda\LTVarB.\LTexf_{\LogFormA}\left(\CPOelA\left(\lambda\LTVarC.\CPOelC\left(\LTFinFunExtend{\CPOelD}{\LTVarB}{\LTVarC}\right)\right)\right)}\right)\RealTermA}\notin\RealVal{\LogFormA\left(\LogBoolIn{\RealTermA}{\RealPredA},\RealTermA\right)}
\end{equation*}
If $\CPOproj_1\left(\CPOelD\left(\RealTermA\right)\right)=0$ then $\CPOproj_2\left(\CPOelD\left(\RealTermA\right)\right)\notin\RealVal{\LogFormA\left(\LogBoolIn{\RealTermA}{\RealPredA},\RealTermA\right)}$, but since $\CPOelD\in E$ we also have:
\begin{equation*}
\CPOproj_2\left(\CPOelD\left(\RealTermA\right)\right)\in\RealVal{\LogFormA\left(\RealBoolT,\RealTermA\right)}\cup\RealVal{\LogFormA\left(\RealBoolF,\RealTermA\right)}
\end{equation*}
and both cases lead to a contradiction by definition of $\RealPredA$. Therefore $\CPOproj_1\left(\CPOelD\left(\RealTermA\right)\right)\neq0$. Moreover $\CPOproj_1\left(\CPOelD\left(\RealTermA\right)\right)\neq\CPObot$ because $\CPOelD\in E$, so we obtain:
\begin{equation*}
\CPOinterp{\LTexf_{\LogFormA}\left(\CPOelA\left(\lambda\LTVarC.\CPOelC\left(\LTFinFunExtend{\CPOelD}{\RealTermA}{\LTVarC}\right)\right)\right)}\notin\RealVal{\LogFormA\left(\LogBoolIn{\RealTermA}{\RealPredA},\RealTermA\right)}
\end{equation*}
and therefore $\CPOinterp{\CPOelA\left(\lambda\LTVarC.\CPOelC\left(\LTFinFunExtend{\CPOelD}{\RealTermA}{\LTVarC}\right)\right)}\notin\RealVal{\LogBoolF}$. Finally, since $\CPOelA\in\RealVal{\neg\forall\LogBoolVarA\,\neg\LogFormA\left(\LogBoolVarA,\RealTermA\right)}$, we have:
\begin{equation*}
\CPOinterp{\lambda\LTVarC.\CPOelC\left(\LTFinFunExtend{\CPOelD}{\RealTermA}{\LTVarC}\right)}\notin\RealVal{\forall\LogBoolVarA\,\neg\LogFormA\left(\LogBoolVarA,\RealTermA\right)}
\end{equation*}
which means that there exists some:
\begin{equation*}
\CPOelE\in\RealVal{\LogFormA\left(\RealBoolT,\RealTermA\right)}\cup\RealVal{\LogFormA\left(\RealBoolF,\RealTermA\right)}
\end{equation*}
such that $\CPOinterp{\CPOelC\left(\LTFinFunExtend{\CPOelD}{\RealTermA}{\CPOelE}\right)}\notin\RealVal{\LogBoolF}$. It is then easy to check that $\CPOinterp{\LTFinFunExtend{\CPOelD}{\RealTermA}{\CPOelE}}\in E$ and $\CPOelD\prec\CPOinterp{\LTFinFunExtend{\CPOelD}{\RealTermA}{\CPOelE}}$, contradicting the maximality of $\CPOelD$.\qedhere
\end{itemize}
\end{proof}
As we explained before the lemma, the next step is the definition of an element of $\bigcap_{\RealTermA\in\Lambda}\RealVal{\exists\LogBoolVarA\left(\LogBoolVarA\Leftrightarrow\LogFormA\left(\RealTermA\right)\right)}$, so that its combination with the realizer above provides an interpretation of the comprehension scheme: $\exists\LogPredVarA\,\forall\LogTermVarA\left(\LogBoolIn{\LogTermVarA}{\LogPredVarA}\Leftrightarrow\LogFormA\left(\LogTermVarA\right)\right)$.
\begin{lem}
If $\LogFormA\left(\LogTermA\right)$ is a closed 1-formula with parameters such that $\LogBoolVarA\notin\FV{\LogFormA\left(\LogTermVarA\right)}$, then:
\begin{equation*}
\CPOinterp{\lambda\LTVarA.\LTVarA\LTPair{\LTexf_{\LogFormA}}{\lambda\LTVarB.\LTVarA\LTPair{\lambda\_.\LTVarB}{\lambda\_.\LTNatZ}}}
\in\bigcap_{\RealTermA\in\Lambda}\RealVal{\exists\LogBoolVarA\left(\LogBoolVarA\Leftrightarrow\LogFormA\left(\RealTermA\right)\right)}
\end{equation*}
\end{lem}
\begin{proof}
Let $\RealTermA\in\Lambda$ and $\CPOelA\in\RealVal{\forall\LogBoolVarA\,\neg\left(\LogBoolVarA\Leftrightarrow\LogFormA\left(\RealTermA\right)\right)}$. We have to prove that:
\begin{equation*}
\CPOinterp{\CPOelA\LTPair{\LTexf_{\LogFormA}}{\lambda\LTVarB.\CPOelA\LTPair{\lambda\_.\LTVarB}{\lambda\_.\LTNatZ}}}\in\RealVal{\LogBoolF}
\end{equation*}
Since $\CPOelA\in\RealVal{\neg\left(\RealBoolF\Leftrightarrow\LogFormA\left(\RealTermA\right)\right)}$, it is sufficient to prove:
\begin{align*}
\CPOinterp{\LTexf_{\LogFormA}}&\in\RealVal{\RealBoolF\Rightarrow\LogFormA\left(\RealTermA\right)}&\CPOinterp{\lambda\LTVarB.\CPOelA\LTPair{\lambda\_.\LTVarB}{\lambda\_.\LTNatZ}}&\in\RealVal{\neg\LogFormA\left(\RealTermA\right)}
\end{align*}
The first one is immediate. For the second, let $\CPOelB\in\RealVal{\LogFormA\left(\RealTermA\right)}$. Since $\CPOelA\in\RealVal{\neg\left(\RealBoolT\Leftrightarrow\LogFormA\left(\RealTermA\right)\right)}$, it is sufficient to prove:
\begin{align*}
\CPOinterp{\lambda\_.\CPOelB}&\in\RealVal{\RealBoolT\Rightarrow\LogFormA\left(\RealTermA\right)}&\CPOinterp{\lambda\_.\LTNatZ}&\in\RealVal{\LogFormA\left(\RealTermA\right)\Rightarrow\RealBoolT}
\end{align*}
The first one is immediate, and the second one follows from $\RealVal{\RealBoolT}=\mathbb{N}_\CPObot$.
\end{proof}
Combining the two realizers above, we can now define:
\begin{equation*}
\LTcomp_{\LogFormA}=\lambda\LTVarA.\LTBBC\left(\lambda\LTVarB.\LTexf_\LogFormA\left(\LTVarB\LTPair{\LTexf_\LogFormA}{\lambda\LTVarD.\LTVarB\LTPair{\lambda\_.\LTVarD}{\lambda\_.\LTNatZ}}\right)\right)\LTVarA\LTFinFunEmpty
\end{equation*}
which by construction realizes the axiom scheme of comprehension:
\begin{equation*}
\CPOinterp{\LTcomp_{\LogFormA}}\in\RealVal{\exists\LogPredVarA\,\forall\LogTermVarA\left(\LogBoolIn{\LogTermVarA}{\LogPredVarA}\Leftrightarrow\LogFormA\left(\LogTermVarA\right)\right)}
\end{equation*}
\subsection{Realizing second-order elimination}
We have now realized the axiom scheme of comprehension that asserts the existence of a first-order element of sort set witnessing any formula. However, we still need to interpret the equivalent of second-order elimination in our setting: subtitution of an arbitrary 1-formula for a first-order set variable. In other words, we have to interpret:
\begin{equation*}
\forall\LogPredVarA\LogFormA\left(\LogPredVarForm{\LogPredVarA}\right)\Longrightarrow\LogFormA\left(\LogFormB\right)
\end{equation*}
for arbitrary 2-formula $\LogFormA\left(\LogFormC\right)$ and 1-formula $\LogFormB\left(\LogTermA\right)$. The first step towards the interpretation of second-order elimination is the interpretation of the following formula:
\begin{equation*}
\forall\LogTermVarA\left(\LogFormB\left(\LogTermVarA\right)\Leftrightarrow\LogFormC\left(\LogTermVarA\right)\right)\Longrightarrow\left(\LogFormA\left(\LogFormB\right)\Leftrightarrow\LogFormA\left(\LogFormC\right)\right)
\end{equation*}
The combination of a realizer of that formula with $\LTcomp_\LogFormB$ will then provide an interpretation of second-order elimination.\par
Since we build the realizer $\LTrepl_\LogFormA$ of that formula by induction on $\LogFormA$, we need to explicitly take into account the free variables of $\LogFormA$: the free vriables of $\LTrepl_\LogFormA$ are the free variables of $\LogFormA$ that are of a computational sort, i.e. $\LogNatVarA$, $\LogTermVarA$ or $\LogTermListVarA$. In particular, if $\LogFormA$ is closed then $\LTrepl_\LogFormA$ is closed as well. For simplicity, we first define $\LTrepl'_\LogFormA$ such that $\FV{\LTrepl'_\LogFormA}=\FV{\LTrepl_\LogFormA}\cup\SetSuch{\LTVarA}{}$, and then define $\LTrepl_\LogFormA=\lambda\LTVarA.\LTrepl'_\LogFormA$. The definition of $\LTrepl'_\LogFormA$ is given in figure~\ref{DefRepl} and we can prove the intended result by induction on $\LogFormA$:
\begin{figure*}
\begin{gather*}
\begin{align*}
\LTrepl'_{\LogPredVarForm{\LogPredVarA}\mapsto\LogBoolIn{\LogTermA}{\LogPredVarA}}&=\LTVarA\,\LTinterp{\LogTermA}
&
\LTrepl'_{\LogPredVarForm{\LogPredVarA}\mapsto\LogBoolA}&=\LTPair{\lambda\LTVarB.\LTVarB}{\lambda\LTVarB.\LTVarB}\text{ if }\LogBoolA\not\equiv\LogBoolIn{\LogTermA}{\LogPredVarA}
&
\end{align*}
\\
\LTrepl'_{\LogFormA_1\Rightarrow\LogFormA_2}=\LTPair{\lambda\LTVarB\LTVarC.\LTProj_1\,\LTrepl'_{\LogFormA_2}\left(\LTVarB\left(\LTProj_2\,\LTrepl'_{\LogFormA_1}\,\LTVarC\right)\right)}{\lambda\LTVarB\LTVarC.\LTProj_2\,\LTrepl'_{\LogFormA_2}\left(\LTVarB\left(\LTProj_1\,\LTrepl'_{\LogFormA_1}\,\LTVarC\right)\right)}
\\
\LTrepl'_{\LogFormA_1\wedge\LogFormA_2}=\LTPair{\lambda\LTVarB.\LTPair{\LTProj_1\,\LTrepl'_{\LogFormA_1}\left(\LTProj_1\,\LTVarB\right)}{\LTProj_1\,\LTrepl'_{\LogFormA_2}\left(\LTProj_2\,\LTVarB\right)}\right.}{\hspace{150pt}\\\hspace{150pt}\left.\lambda\LTVarB.\LTPair{\LTProj_2\,\LTrepl'_{\LogFormA_1}\left(\LTProj_1\,\LTVarB\right)}{\LTProj_2\,\LTrepl'_{\LogFormA_2}\left(\LTProj_2\,\LTVarB\right)}}
\\
\begin{flalign*}
\LTrepl'_{\forall\LogRelVarA\,\LogFormA}&=\LTPair{\lambda\LTVarB\LogRelVarA.\LTProj_1\,\LTrepl'_\LogFormA\left(\LTVarB\,\LogRelVarA\right)}{\lambda\LTVarB\LogRelVarA.\LTProj_2\,\LTrepl'_\LogFormA\left(\LTVarB\,\LogRelVarA\right)}
&
\LTrepl'_{\forall\LogPredVarA\,\LogFormA}&=\LTrepl'_{\forall\LogBoolVarA\,\LogFormA}=\LTrepl'_{\LogFormA}
\end{flalign*}
\end{gather*}
\caption{Definition of $\LTrepl_\LogFormA$}
\label{DefRepl}
\end{figure*}
\begin{lem}
\label{RealRepl}
If $\LogFormA\left(\LogFormD\right)$ is a 2-formula, $\LogFormB\left(\LogTermA\right)$, $\LogFormC\left(\LogTermA\right)$ are closed 1-formulas with parameters and $\ValuatA$ is a valuation on $\LogFormA$ then:
\begin{equation*}
\CPOinterp{\LTrepl_\LogFormA}_\ValuatA\in\RealVal{\forall\LogTermVarA\left(\LogFormB\left(\LogTermVarA\right)\Leftrightarrow\LogFormC\left(\LogTermVarA\right)\right)\Rightarrow\left(\LogFormA\left(\LogFormB\right)\Leftrightarrow\LogFormA\left(\LogFormC\right)\right)}_\ValuatA
\end{equation*}
\end{lem}
\begin{proof}
We have to prove that if $\CPOelA\in\RealVal{\forall\LogTermVarA\left(\LogFormB\left(\LogTermVarA\right)\Leftrightarrow\LogFormC\left(\LogTermVarA\right)\right)}$, then:
\begin{equation*}
\CPOinterp{\lambda\LTVarA.\LTrepl'_\LogFormA}_\ValuatA\left(\CPOelA\right)=\CPOinterp{\LTrepl'_\LogFormA}_{\ValuatA\uplus\SetSuch{\LTVarA\mapsto\CPOelA}{}}\in\RealVal{\LogFormA\left(\LogFormB\right)\Leftrightarrow\LogFormA\left(\LogFormC\right)}_\ValuatA
\end{equation*}
We write $\ValuatA'=\ValuatA\uplus\SetSuch{\LTVarA\mapsto\CPOelA}{}$ and proceed by induction on $\LogFormA$:
\begin{itemize}
\item$\LogFormA\left(\LogFormD\right)\equiv\LogFormD\left(\LogTermA\right)$: since $\CPOinterp{\LTinterp{\LogTermA}}_\ValuatA\in\Lambda$, we have by hypothesis on $\CPOelA$:
\begin{align*}
\CPOinterp{\LTVarA\,\LTinterp{\LogTermA}}_{\ValuatA'}=\CPOelA\left(\CPOinterp{\LTinterp{\LogTermA}}_\ValuatA\right)\in\RealVal{\LogFormB\left(\LogTermVarA\right)\Leftrightarrow\LogFormC\left(\LogTermVarA\right)}_{\ValuatA\uplus\SetSuch{\LogTermVarA\mapsto\CPOinterp{\LTinterp{\LogTermA}}_\ValuatA}{}}&=\RealVal{\LogFormB\left(\LogTermA\right)\Leftrightarrow\LogFormC\left(\LogTermA\right)}_\ValuatA\\
&=\RealVal{\LogFormA\left(\LogFormB\right)\Leftrightarrow\LogFormA\left(\LogFormC\right)}_\ValuatA
\end{align*}
\item$\LogFormA\left(\LogFormD\right)\equiv\LogBoolA\not\equiv\LogFormD\left(\LogTermA\right)$: immediate since $\LogFormA\left(\LogFormB\right)\equiv\LogFormA\left(\LogFormC\right)$ in that case
\item$\LogFormA\left(\LogFormD\right)\equiv\LogFormA_1\left(\LogFormD\right)\Rightarrow\LogFormA_2\left(\LogFormD\right)$: we have by induction hypothesis:
\begin{equation*}
\CPOinterp{\LTProj_2\,\LTrepl'_{\LogFormA_1}}_{\ValuatA'}\in\RealVal{\LogFormA_1\left(\LogFormC\right)\Rightarrow\LogFormA_1\left(\LogFormB\right)}_\ValuatA
\end{equation*}
therefore if $\CPOelB\in\RealVal{\LogFormA\left(\LogFormB\right)}_\ValuatA$ and $\CPOelC\in\RealVal{\LogFormA_1\left(\LogFormC\right)}_\ValuatA$ we get:
\begin{equation*}
\CPOinterp{\CPOelB\left(\LTProj_2\,\LTrepl'_{\LogFormA_1}\,\CPOelC\right)}_{\ValuatA'}\in\RealVal{\LogFormA_2\left(\LogFormB\right)}_\ValuatA
\end{equation*}
but the second induction hypothesis gives:
\begin{equation*}
\CPOinterp{\LTProj_1\,\LTrepl'_{\LogFormA_2}}_{\ValuatA'}\in\RealVal{\LogFormA_2\left(\LogFormB\right)\Rightarrow\LogFormA_2\left(\LogFormC\right)}_\ValuatA
\end{equation*}
so we have:
\begin{equation*}
\CPOinterp{\LTProj_1\,\LTrepl'_{\LogFormA_2}\left(\CPOelB\left(\LTProj_2\,\LTrepl'_{\LogFormA_1}\,\CPOelC\right)\right)}_{\ValuatA'}\in\RealVal{\LogFormA_2\left(\LogFormC\right)}_\ValuatA
\end{equation*}
and therefore:
\begin{equation*}
\CPOinterp{\lambda\LTVarB\LTVarC.\LTProj_1\,\LTrepl'_{\LogFormA_2}\left(\LTVarB\left(\LTProj_2\,\LTrepl'_{\LogFormA_1}\,\LTVarC\right)\right)}_{\ValuatA'}\in\RealVal{\LogFormA\left(\LogFormB\right)\Rightarrow\LogFormA\left(\LogFormC\right)}_\ValuatA
\end{equation*}
similarly we have:
\begin{equation*}
\CPOinterp{\lambda\LTVarB\LTVarC.\LTProj_2\,\LTrepl'_{\LogFormA_2}\left(\LTVarB\left(\LTProj_1\,\LTrepl'_{\LogFormA_1}\,\LTVarC\right)\right)}_{\ValuatA'}\in\RealVal{\LogFormA\left(\LogFormC\right)\Rightarrow\LogFormA\left(\LogFormB\right)}_\ValuatA
\end{equation*}
and therefore $\CPOinterp{\LTrepl'_\LogFormA}_{\ValuatA'}\in\RealVal{\LogFormA\left(\LogFormB\right)\Leftrightarrow\LogFormA\left(\LogFormC\right)}_\ValuatA$
\item$\LogFormA\left(\LogFormD\right)\equiv\LogFormA_1\left(\LogFormD\right)\wedge\LogFormA_2\left(\LogFormD\right)$: we have by induction hypothesis:
\begin{equation*}
\CPOinterp{\LTProj_1\,\LTrepl'_{\LogFormA_1}}_{\ValuatA'}\in\RealVal{\LogFormA_1\left(\LogFormB\right)\Rightarrow\LogFormA_1\left(\LogFormC\right)}_\ValuatA
\end{equation*}
therefore if $\CPOelB\in\RealVal{\LogFormA\left(\LogFormB\right)}_\ValuatA$ we get:
\begin{equation*}
\CPOinterp{\LTProj_1\,\LTrepl'_{\LogFormA_1}\left(\LTProj_1\,\CPOelB\right)}_{\ValuatA'}\in\RealVal{\LogFormA_1\left(\LogFormC\right)}_\ValuatA
\end{equation*}
but the second induction hypothesis gives:
\begin{equation*}
\CPOinterp{\LTProj_1\,\LTrepl'_{\LogFormA_2}}_{\ValuatA'}\in\RealVal{\LogFormA_2\left(\LogFormB\right)\Rightarrow\LogFormA_2\left(\LogFormC\right)}_\ValuatA
\end{equation*}
so we have:
\begin{equation*}
\CPOinterp{\LTProj_1\,\LTrepl'_{\LogFormA_2}\left(\LTProj_2\,\CPOelB\right)}_{\ValuatA'}\in\RealVal{\LogFormA_2\left(\LogFormC\right)}_\ValuatA
\end{equation*}
and therefore:
\begin{equation*}
\CPOinterp{\lambda\LTVarB.\LTPair{\LTProj_1\,\LTrepl'_{\LogFormA_1}\left(\LTProj_1\,\LTVarB\right)}{\LTProj_1\,\LTrepl'_{\LogFormA_2}\left(\LTProj_2\,\LTVarB\right)}}_{\ValuatA'}\in\RealVal{\LogFormA\left(\LogFormB\right)\Rightarrow\LogFormA\left(\LogFormC\right)}_\ValuatA
\end{equation*}
similarly we have:
\begin{equation*}
\CPOinterp{\lambda\LTVarB.\LTPair{\LTProj_2\,\LTrepl'_{\LogFormA_1}\left(\LTProj_1\,\LTVarB\right)}{\LTProj_2\,\LTrepl'_{\LogFormA_2}\left(\LTProj_2\,\LTVarB\right)}}_{\ValuatA'}\in\RealVal{\LogFormA\left(\LogFormC\right)\Rightarrow\LogFormA\left(\LogFormB\right)}_\ValuatA
\end{equation*}
and therefore $\CPOinterp{\LTrepl'_\LogFormA}_{\ValuatA'}\in\RealVal{\LogFormA\left(\LogFormB\right)\Leftrightarrow\LogFormA\left(\LogFormC\right)}_\ValuatA$
\item$\LogFormA\left(\LogFormD\right)\equiv\forall\LogRelVarA\,\LogFormA_0\left(\LogFormD\right)$: we do the case $\LogRelVarA\equiv\LogTermVarA$, the other ones being similar. The induction hypothesis implies that for any $\RealTermA\in\Lambda$:
\begin{equation*}
\CPOinterp{\LTProj_1\,\LTrepl'_{\LogFormA_0}}_{\ValuatA'\uplus\SetSuch{\LogTermVarA\mapsto\RealTermA}{}}\in\RealVal{\LogFormA_0\left(\LogFormB\right)\Rightarrow\LogFormA_0\left(\LogFormC\right)}_{\ValuatA\uplus\SetSuch{\LogTermVarA\mapsto\RealTermA}{}}
\end{equation*}
if $\CPOelB\in\RealVal{\LogFormA\left(\LogFormB\right)}_\ValuatA$ and $\RealTermA\in\Lambda$ then:
\begin{equation*}
\CPOinterp{\CPOelB\,\LogTermVarA}_{\ValuatA'\uplus\SetSuch{\LogTermVarA\mapsto\RealTermA}{}}=\CPOelB\left(\RealTermA\right)\in\RealVal{\LogFormA_0\left(\LogFormB\right)}_{\ValuatA\uplus\SetSuch{\LogTermVarA\mapsto\RealTermA}{}}
\end{equation*}
so we have:
\begin{equation*}
\CPOinterp{\LTProj_1\,\LTrepl'_{\LogFormA_0}\left(\CPOelB\,\LogTermVarA\right)}_{\ValuatA'\uplus\SetSuch{\LogTermVarA\mapsto\RealTermA}{}}\in\RealVal{\LogFormA_0\left(\LogFormC\right)}_{\ValuatA\uplus\SetSuch{\LogTermVarA\mapsto\RealTermA}{}}
\end{equation*}
and therefore:
\begin{equation*}
\CPOinterp{\lambda\LTVarB\LogTermVarA.\LTProj_1\,\LTrepl'_{\LogFormA_0}\left(\LTVarB\,\LogTermVarA\right)}_{\ValuatA'}\in\RealVal{\LogFormA\left(\LogFormB\right)\Rightarrow\LogFormA\left(\LogFormC\right)}_\ValuatA
\end{equation*}
similarly we have:
\begin{equation*}
\CPOinterp{\lambda\LTVarB\LogTermVarA.\LTProj_2\,\LTrepl'_{\LogFormA_0}\left(\LTVarB\,\LogTermVarA\right)}_{\ValuatA'}\in\RealVal{\LogFormA\left(\LogFormC\right)\Rightarrow\LogFormA\left(\LogFormB\right)}_\ValuatA
\end{equation*}
and therefore $\CPOinterp{\LTrepl'_\LogFormA}_{\ValuatA'}\in\RealVal{\LogFormA\left(\LogFormB\right)\Leftrightarrow\LogFormA\left(\LogFormC\right)}_\ValuatA$
\item$\LogFormA\left(\LogFormD\right)\equiv\forall\LogPredVarA\,\LogFormA_0\left(\LogFormD\right)$ or $\LogFormA\left(\LogFormD\right)\equiv\forall\LogBoolVarA\,\LogFormA_0\left(\LogFormD\right)$: we treat only the case of $\LogPredVarA$ since the other one is similar. The induction hypothesis implies that for any $\RealPredA\subseteq\Lambda$:
\begin{equation*}
\CPOinterp{\LTrepl'_{\LogFormA_0}}_{\ValuatA'\uplus\SetSuch{\LogPredVarA\mapsto\RealPredA}{}}\in\RealVal{\LogFormA_0\left(\LogFormB\right)\Leftrightarrow\LogFormA_0\left(\LogFormC\right)}_{\ValuatA\uplus\SetSuch{\LogPredVarA\mapsto\RealPredA}{}}
\end{equation*}
but since $\LTrepl'_{\LogFormA_0}$ does not contain variable $\LogPredVarA$ we get:
\begin{equation*}
\CPOinterp{\LTrepl'_{\LogFormA_0}}_{\ValuatA'}\in\RealVal{\LogFormA_0\left(\LogFormB\right)\Leftrightarrow\LogFormA_0\left(\LogFormC\right)}_{\ValuatA\uplus\SetSuch{\LogPredVarA\mapsto\RealPredA}{}}
\end{equation*}
so we get:
\begin{equation*}
\CPOinterp{\LTrepl'_{\LogFormA_0}}_{\ValuatA'}\in\RealVal{\forall\LogPredVarA\left(\LogFormA_0\left(\LogFormB\right)\Leftrightarrow\LogFormA_0\left(\LogFormC\right)\right)}_\ValuatA
\end{equation*}
but since for any closed formulas with parameters $\LogFormD$ and $\LogFormD'$ we have:
\begin{gather*}
\RealVal{\forall\LogPredVarA\left(\LogFormD\wedge\LogFormD'\right)}=\RealVal{\forall\LogPredVarA\,\LogFormD\wedge\forall\LogPredVarA\,\LogFormD'}\\
\RealVal{\forall\LogPredVarA\left(\LogFormD\Rightarrow\LogFormD'\right)}\subseteq\RealVal{\forall\LogPredVarA\,\LogFormD\Rightarrow\forall\LogPredVarA\,\LogFormD'}
\end{gather*}
we then obtain $\CPOinterp{\LTrepl'_\LogFormA}_{\ValuatA'}\in\RealVal{\LogFormA\left(\LogFormB\right)\Leftrightarrow\LogFormA\left(\LogFormC\right)}_\ValuatA$\qedhere
\end{itemize}
\end{proof}
We can now interpret the instantiation of a set variable with an arbitrary 1-formula:
\begin{equation*}
\forall\LogPredVarA\,\LogFormA\left(\LogPredVarForm{\LogPredVarA}\right)\Rightarrow\LogFormA\left(\LogFormB\right)
\end{equation*}
Since the existential quantifier is not primitive in our logic, our version of the axiom scheme of comprehension is in fact:
\begin{equation*}
\neg\forall\LogPredVarA\,\neg\forall\LogTermVarA\left(\LogBoolIn{\LogTermVarA}{\LogPredVarA}\Leftrightarrow\LogFormB\left(\LogTermVarA\right)\right)
\end{equation*}
therefore, the elimination of such an existential quantifier will require classical logic. Our realizer $\LTelim_{\LogFormA,\LogFormB}$ of second-order elimination (where $\LogFormA\left(\LogFormC\right)$ is a 2-formula and $\LogFormB\left(\LogTermA\right)$ is a 1-formula) is such that $\FV{\LTelim_{\LogFormA,\LogFormB}}=\FV{\LogFormA}\cap\LogRelVarA$ (that is, the free variables of $\LTelim_{\LogFormA,\LogFormB}$ are the free variables of $\LogFormA$ which are of a computational sort) and is defined as:
\begin{equation*}
\LTelim_{\LogFormA,\LogFormB}=\lambda\LTVarA.\LTdne_{\LogFormA\left(\LogFormB\right)}\left(\lambda\LTVarB.\LTcomp_{\LogFormB}\left(\lambda\LTVarC.\LTVarB\left(\LTProj_1\left(\LTrepl_\LogFormA\,\LTVarC\right)\LTVarA\right)\right)\right)
\end{equation*}
Correctness of this realizer is then an easy consequence of the lemmas above:
\begin{lem}
\label{RealElim}
If $\LogFormA\left(\LogFormC\right)$ is a 2-formula such that $\LogPredVarA\notin\FV{\LogFormA}$ (meaning that $\LogPredVarA\notin\FV{\LogFormA\left(\LogFormC\right)}$ whenever $\LogPredVarA\notin\FV{\LogFormC}$), if $\LogFormB\left(\LogTermA\right)$ is a closed 1-formula with parameters and if $\ValuatA$ is a valuation on $\LogFormA$ then:
\begin{equation*}
\CPOinterp{\LTelim_{\LogFormA,\LogFormB}}_\ValuatA\in\RealVal{\forall\LogPredVarA\,\LogFormA\left(\LogPredVarForm{\LogPredVarA}\right)\Rightarrow\LogFormA\left(\LogFormB\right)}_\ValuatA
\end{equation*}
\end{lem}
\begin{proof}
Let $\CPOelA\in\RealVal{\forall\LogPredVarA\,\LogFormA\left(\LogPredVarForm{\LogPredVarA}\right)}_\ValuatA$. Since:
\begin{equation*}
\CPOinterp{\LTdne_{\LogFormA\left(\LogFormB\right)}}\in\RealVal{\neg\neg\LogFormA\left(\LogFormB\right)\Rightarrow\LogFormA\left(\LogFormB\right)}_\ValuatA
\end{equation*}
we are left to prove:
\begin{equation*}
\CPOinterp{\lambda\LTVarB.\LTcomp_\LogFormB\left(\lambda\LTVarC.\LTVarB\left(\LTProj_1\left(\LTrepl_\LogFormA\,\LTVarC\right)\CPOelA\right)\right)}_\ValuatA\in\RealVal{\neg\neg\LogFormA\left(\LogFormB\right)}_\ValuatA
\end{equation*}
Let $\CPOelB\in\RealVal{\neg\LogFormA\left(\LogFormB\right)}_\ValuatA$. Since:
\begin{equation*}
\CPOinterp{\LTcomp_\LogFormB}\in\RealVal{\neg\forall\LogPredVarA\,\neg\forall\LogTermVarA\left(\LogBoolIn{\LogTermVarA}{\LogPredVarA}\Leftrightarrow\LogFormB\left(\LogTermVarA\right)\right)}
\end{equation*}
we are left to prove:
\begin{equation*}
\CPOinterp{\lambda\LTVarC.\CPOelB\left(\LTProj_1\left(\LTrepl_\LogFormA\,\LTVarC\right)\CPOelA\right)}_\ValuatA\in\RealVal{\forall\LogPredVarA\,\neg\forall\LogTermVarA\left(\LogBoolIn{\LogTermVarA}{\LogPredVarA}\Leftrightarrow\LogFormB\left(\LogTermVarA\right)\right)}_\ValuatA
\end{equation*}
Let $\RealPredA\subseteq\Lambda$ and $\CPOelC\in\RealVal{\forall\LogTermVarA\left(\LogBoolIn{\LogTermVarA}{\RealPredA}\Leftrightarrow\LogFormB\left(\LogTermVarA\right)\right)}_\ValuatA$. We need to prove:
\begin{equation*}
\CPOinterp{\CPOelB\left(\LTProj_1\left(\LTrepl_\LogFormA\,\CPOelC\right)\CPOelA\right)}_\ValuatA\in\RealVal{\LogBoolF}_\ValuatA
\end{equation*}
but we have:
\begin{equation*}
\CPOinterp{\LTProj_1\left(\LTrepl_\LogFormA\,\CPOelC\right)}_\ValuatA\in\RealVal{\LogFormA\left(\LogPredVarForm{\RealPredA}\right)\Rightarrow\LogFormA\left(\LogFormB\right)}_\ValuatA
\end{equation*}
and since $\CPOelA\in\RealVal{\LogFormA\left(\LogPredVarForm{\RealPredA}\right)}_\ValuatA$ we get:
\begin{equation*}
\CPOinterp{\LTProj_1\left(\LTrepl_\LogFormA\,\CPOelC\right)\CPOelA}_\ValuatA\in\RealVal{\LogFormA\left(\LogFormB\right)}_\ValuatA
\end{equation*}
finally, since $\CPOelB\in\RealVal{\neg\LogFormA\left(\LogFormB\right)}_\ValuatA$ we obtain $\CPOinterp{\CPOelB\left(\LTProj_1\left(\LTrepl_\LogFormA\,\CPOelC\right)\CPOelA\right)}_\ValuatA\in\RealVal{\LogBoolF}_\ValuatA$
\end{proof}
\subsection{Realizing normalization of system F}
We now have an interpretation of full second-order arithmetic. Therefore, we can describe the details of our interpretation of the proof of normalization of system F given in section~\ref{NormProof} using the realizer $\LTelim$ of the previous section. The first realizer corresponds to lemma~\ref{NormRC} and is defined as follows:
\begin{equation*}
\LTnormrc=\LTPair{\LTPair{\LTnormrc^{(1)}}{\LTnormrc^{(2)}}}{\LTnormrc^{(3)}}
\end{equation*}
where:
\begin{align*}
\LTnormrc^{(1)}&=\lambda\LogTermListVarA\LTVarA.\LTVarA\,\LTNatZ
&
\LTnormrc^{(2)}&=\lambda\LogTermVarA\LTVarA.\LTVarA
&
\LTnormrc^{(3)}&=\lambda\LogTermVarA\LogTermVarB\LogTermListVarA\LTVarA\LTVarB.\LTVarA\left(\lambda\LogNatVarA.\LTVarB\left(\LTNatS\,\LogNatVarA\right)\right)
\end{align*}
$\LTnormrc$ can be shown to be the computational interpretation of the proof of lemma~\ref{NormRC}:
\begin{lem}
\label{RealNormRC}
$\CPOinterp{\LTnormrc}\in\RealVal{\LogRedCand{\LogNormForm}}$
\end{lem}
\begin{proof}
\begin{itemize}
\item$\CPOinterp{\LTnormrc^{(1)}}\in\RealVal{\forall\LogTermListVarA\,\neg\forall\LogNatVarA\,\LogNNorm{\LogTermApp{\LogTermVar{\LogNatZ}}{\LogTermListVarA}}{\LogNatVarA}}$: let $\RealTermListA\in\Lambda^*$ and let $\CPOelA\in\RealVal{\forall\LogNatVarA\,\LogNNorm{\LogTermApp{\LogTermVar{\LogNatZ}}{\RealTermListA}}{\LogNatVarA}}$. Then we have $\CPOinterp{\CPOelA\,\LTNatZ}=\CPOelA\left(0\right)\in\RealVal{\LogNNorm{\LogTermApp{\LogTermVar{\LogNatZ}}{\RealTermListA}}{0}}$. Since $\CPOinterp{\LTinterp{\left(\LogTermApp{\LogTermVar{\LogNatZ}}{\RealTermListA}\right)}}=\CPOinterp{\LTLamListApp\,\left(\LTLamVar\,\LTNatZ\right)\RealTermListA}$ is in head normal form, $\RealVal{\LogNNorm{\LogTermApp{\LogTermVar{\LogNatZ}}{\RealTermListA}}{0}}=\RealVal{\LogBoolF}$ and therefore $\CPOelA\left(0\right)\in\RealVal{\LogBoolF}$
\item$\CPOinterp{\LTnormrc^{(2)}}\in\RealVal{\forall\LogTermVarA\left(\LogNorm{\LogTermVarA}\Rightarrow\LogNorm{\LogTermVarA}\right)}$: immediate
\item$\CPOinterp{\LTnormrc^{(3)}}\in\RealVal{\forall\LogTermVarA\,\forall\LogTermVarB\,\forall\LogTermListVarA\left(\LogNorm{\left(\LogTermApp{\LogTermSubst{\LogTermVarA}{\LogTermVarB}}{\LogTermListVarA}\right)}\Rightarrow\LogNorm{\left(\LogTermApp{\LogTermApp{\left(\LogTermLam{\LogTermVarA}\right)}{\LogTermListAbbrv{\LogTermVarB}}}{\LogTermListVarA}\right)}\right)}$: let $\RealTermA\in\Lambda$, $\RealTermB\in\Lambda$, $\RealTermListA\in\Lambda^*$ and let $\CPOelA\in\RealVal{\LogNorm{\left(\LogTermApp{\LogTermSubst{\RealTermA}{\RealTermB}}{\RealTermListA}\right)}}$ and $\CPOelB\in\RealVal{\forall\LogNatVarA\,\LogNNorm{\left(\LogTermApp{\LogTermApp{\left(\LogTermLam{\RealTermA}\right)}{\LogTermListAbbrv{\RealTermB}}}{\RealTermListA}\right)}{\LogNatVarA}}$. We have to prove that:
\begin{equation*}
\CPOinterp{\CPOelA\left(\lambda\LogNatVarA.\CPOelB\left(\LTNatS\,\LogNatVarA\right)\right)}\in\RealVal{\LogBoolF}
\end{equation*}
but since $\CPOelA\in\RealVal{\LogNorm{\left(\LogTermApp{\LogTermSubst{\RealTermA}{\RealTermB}}{\RealTermListA}\right)}}$, this reduces to:
\begin{equation*}
\CPOinterp{\lambda\LogNatVarA.\CPOelB\left(\LTNatS\,\LogNatVarA\right)}\in\RealVal{\forall\LogNatVarA\,\LogNNorm{\left(\LogTermApp{\LogTermSubst{\RealTermA}{\RealTermB}}{\RealTermListA}\right)}{\LogNatVarA}}
\end{equation*}
Let $\RealNatA\in\mathbb{N}$, we need to prove:
\begin{equation*}
\CPOelB\left(\RealNatA+1\right)\in\RealVal{\LogNNorm{\left(\LogTermApp{\LogTermSubst{\RealTermA}{\RealTermB}}{\RealTermListA}\right)}{\RealNatA}}
\end{equation*}
But $\CPOelB\left(\RealNatA+1\right)\in\RealVal{\LogNNorm{\left(\LogTermApp{\LogTermApp{\left(\LogTermLam{\RealTermA}\right)}{\LogTermListAbbrv{\RealTermB}}}{\RealTermListA}\right)}{\RealNatA+1}}$ and:
\begin{equation*}
\CPOinterp{\LTLamListApp\left(\LTLamApp\left(\LTLamAbs\,\RealTermA\right)\RealTermB\right)\RealTermListA}\FRed\CPOinterp{\LTLamListApp\left(\LTLamSubst\,\RealTermA\,\LTNatZ\left(\LTLamListCons\,\LTLamListNil\,\RealTermB\right)\right)\RealTermListA}
\end{equation*}
for weak head reduction and therefore:
\begin{equation*}
\RealVal{\LogNNorm{\left(\LogTermApp{\LogTermApp{\left(\LogTermLam{\RealTermA}\right)}{\LogTermListAbbrv{\RealTermB}}}{\RealTermListA}\right)}{\RealNatA+1}}=\RealVal{\LogNNorm{\left(\LogTermApp{\LogTermSubst{\RealTermA}{\RealTermB}}{\RealTermListA}\right)}{\RealNatA}}
\end{equation*}
which concludes the proof.\qedhere
\end{itemize}
\end{proof}
We now give the interpretation of the proof of lemma~\ref{RCisRC}: if $\LogRedCand{\LogPredVarForm{\LogPredVarA}}$ for each $\LogPredVarA\in\FV{\FTypeA}$, then $\LogRedCand{\LogRC{\FTypeA}}$. For that we inductively define in figure~\ref{IsrcDef} for each type $\FTypeA$ of system F built from variables $\LogPredVarA$ of the logic a term:
\begin{equation*}
\LTisrc_\FTypeA=\LTPair{\LTPair{\LTisrc_\FTypeA^{(1)}}{\LTisrc_\FTypeA^{(2)}}}{\LTisrc_\FTypeA^{(3)}}
\end{equation*}
such that $\FV{\LTisrc_\FTypeA}=\SetSuch{\LTVarA_\LogPredVarA}{\LogPredVarA\in\FV{\FTypeA}}$.
\begin{figure*}
\begin{gather*}
\begin{align*}
\LTisrc_\LogPredVarA^{(1)}&=\LTProj_1\left(\LTProj_1\,\LTVarA_\LogPredVarA\right)
&
\LTisrc_\LogPredVarA^{(2)}&=\LTProj_2\left(\LTProj_1\,\LTVarA_\LogPredVarA\right)
&
\LTisrc_\LogPredVarA^{(3)}&=\LTProj_2\,\LTVarA_\LogPredVarA
&
\end{align*}
\\
\begin{flalign*}
\LTisrc_{\FTypeA\to\FTypeB}^{(1)}&=\lambda\LogTermListVarA\LogTermVarA\LTVarA.\LTisrc_\FTypeB^{(1)}\left(\LTLamListCons\,\LogTermListVarA\,\LogTermVarA\right)
&
\LTisrc_{\FTypeA\to\FTypeB}^{(3)}&=\lambda\LogTermVarA\LogTermVarB\LogTermListVarA\LTVarA\LogTermVarC\LTVarB.\LTisrc_\FTypeB^{(3)}\,\LogTermVarA\,\LogTermVarB\left(
\LTLamListCons\,\LogTermListVarA\,\LogTermVarC\right)\left(\LTVarA\,\LogTermVarC\,\LTVarB\right)
\end{flalign*}
\\
\LTisrc_{\FTypeA\to\FTypeB}^{(2)}=\lambda\LogTermVarA\LTVarA.\LTisrc_\FTypeB^{(2)}\left(\LTLamApp\,\LogTermVarA\left(\LTLamVar\,\LTNatZ\right)\right)\left(\LTVarA\left(\LTLamVar\,\LTNatZ\right)\left(\LTisrc_\FTypeA^{(1)}\,\LTLamListNil\right)\right)
\\
\LTisrc_{\forall\LogPredVarA\,\FTypeA}^{(1)}=\lambda\LogTermListVarA\LTVarA_\LogPredVarA.\LTisrc_\FTypeA^{(1)}\LogTermListVarA
\qquad
\LTisrc_{\forall\LogPredVarA\,\FTypeA}^{(3)}=\lambda\LogTermVarA\LogTermVarB\LogTermListVarA\LTVarB\LTVarA_\LogPredVarA.\LTisrc_{\FTypeA}^{(3)}\,\LogTermVarA\,\LogTermVarB\,\LogTermListVarA\left(\LTVarB\,\LTVarA_\LogPredVarA\right)
\\
\begin{flalign*}
&\LTisrc_{\forall\LogPredVarA\,\FTypeA}^{(2)}=\lambda\LogTermVarA\LTVarA.\LTelim_{\LogPredVarForm{\LogPredVarA}\mapsto\LogRedCand{\LogPredVarForm{\LogPredVarA}}\Rightarrow\forall\LogTermVarA\left(\LogRC{\FTypeA}\left(\LogTermVarA\right)\Rightarrow\LogNorm{\LogTermVarA}\right),\LogNormForm}\left(\lambda\LTVarA_\LogPredVarA.\LTisrc_\FTypeA^{(2)}\right)\LTnormrc\,\LogTermVarA&
\end{flalign*}
\\
\begin{flalign*}
&&\left(\LTelim_{\LogPredVarForm{\LogPredVarA}\mapsto\LogRedCand{\LogPredVarForm{\LogPredVarA}}\Rightarrow\LogRC{\FTypeA}\left(\LogTermVarA\right),\LogNormForm}\,\LTVarA\,\LTnormrc\right)
\end{flalign*}
\end{gather*}
\vspace{-15pt}
\caption{Definition of $\LTisrc_\FTypeA$}
\label{IsrcDef}
\end{figure*}
Our claim is then that if we substitute a realizer of $\LogRedCand{\LogPredVarForm{\LogPredVarA}}$ for each corresponding variable $\LTVarA_\LogPredVarA$, we obtain a realizer of $\LogRedCand{\LogRC{\FTypeA}}$:
\begin{lem}
\label{RealIsRC}
If $\FTypeA$ is a type of system F in which type variables are variables $\LogPredVarA$ of the logic, if $\ValuatA:\FV{\FTypeA}\to\mathcal{P}\left(\Lambda\right)$ is a valuation on $\LogRedCand{\LogRC{\FTypeA}}$ and if:
\begin{equation*}
\ValuatA':\SetSuch{\LTVarA_\LogPredVarA}{\LogPredVarA\in\FV{\FTypeA}}\to\CPOinterp{\LTinterp{\LogRedCand{\LogPredVarForm{\LogPredVarA}}}}
\end{equation*}
(this codomain does not depend on the particular $\LogPredVarA$ chosen) is a valuation on $\LTisrc_\FTypeA$ such that $\ValuatA'\left(\LTVarA_\LogPredVarA\right)\in\RealVal{\LogRedCand{\LogPredVarForm{\LogPredVarA}}}_\ValuatA$ for each $\LogPredVarA\in\FV{\FTypeA}$, then:
\begin{equation*}
\CPOinterp{\LTisrc_\FTypeA}_{\ValuatA'}\in\RealVal{\LogRedCand{\LogRC{\FTypeA}}}_\ValuatA
\end{equation*}
\end{lem}
\begin{proof}
\begin{itemize}
\item$\LogPredVarA$: we have by hypothesis:
\begin{equation*}
\ValuatA'\left(\LTVarA_\LogPredVarA\right)\in\RealVal{\LogRedCand{\LogPredVarForm{\LogPredVarA}}}_\ValuatA=\RealVal{\LogRedCand{\LogRC{\LogPredVarA}}}_\ValuatA
\end{equation*}
and therefore:
\begin{equation*}
\CPOinterp{\LTPair{\LTPair{\LTProj_1\left(\LTProj_1\,\LTVarA_\LogPredVarA\right)}{\LTProj_2\left(\LTProj_1\,\LTVarA_\LogPredVarA\right)}}{\LTProj_2\,\LTVarA_\LogPredVarA}}_{\ValuatA'}=\CPOinterp{\LTVarA_\LogPredVarA}_{\ValuatA'}\in\RealVal{\LogRedCand{\LogRC{\LogPredVarA}}}_\ValuatA
\end{equation*}
\item$\FTypeA\to\FTypeB$:
\begin{itemize}
\item$\CPOinterp{\LTisrc_{\FTypeA\to\FTypeB}^{(1)}}_{\ValuatA'}\in\RealVal{\forall\LogTermListVarA\,\LogRC{\FTypeA\to\FTypeB}\left(\LogTermApp{\LogTermVar{\LogNatZ}}{\LogTermListVarA}\right)}_\ValuatA$: let $\RealTermListA\in\Lambda^*$, $\RealTermA\in\Lambda$ and $\CPOelA\in\RealVal{\LogRC{\FTypeA}\left(\RealTermA\right)}_\ValuatA$. The induction hypothesis gives:
\begin{equation*}
\CPOinterp{\LTisrc_\FTypeB^{(1)}}_{\ValuatA'}\in\RealVal{\forall\LogTermListVarA\,\LogRC{\FTypeB}\left(\LogTermApp{\LogTermVar{\LogNatZ}}{\LogTermListVarA}\right)}
\end{equation*}
and therefore:
\begin{equation*}
\CPOinterp{\LTisrc_\FTypeB^{(1)}\left(\LTLamListCons\,\RealTermListA\,\RealTermA\right)}_{\ValuatA'}\in\RealVal{\LogRC{\FTypeB}\left(\LogTermApp{\LogTermVar{\LogNatZ}}{\LogTermListAbbrv{\RealTermListA,\RealTermA}}\right)}
\end{equation*}
\item$\CPOinterp{\LTisrc_{\FTypeA\to\FTypeB}^{(2)}}_{\ValuatA'}\in\RealVal{\forall\LogTermVarA\left(\LogRC{\FTypeA\to\FTypeB}\left(\LogTermVarA\right)\Rightarrow\LogNorm{\LogTermVarA}\right)}_\ValuatA$: let $\RealTermA\in\Lambda$ and $\CPOelA\in\RealVal{\LogRC{\FTypeA\to\FTypeB}\left(\RealTermA\right)}_\ValuatA$. The induction hypothesis implies:
\begin{equation*}
\CPOinterp{\LTisrc_\FTypeA^{(1)}\,\LTLamListNil}_{\ValuatA'}\in\RealVal{\LogRC{\FTypeA}\left(\LogTermApp{\LogTermVar{\LogNatZ}}{\LogTermListNil}\right)}_\ValuatA=\RealVal{\LogRC{\FTypeA}\left(\LogTermVar{\LogNatZ}\right)}_\ValuatA
\end{equation*}
so since $\CPOinterp{\CPOelA\left(\LTLamVar\,\LTNatZ\right)}\in\RealVal{\LogRC{\FTypeA}\left(\LogTermVar{\LogNatZ}\right)\Rightarrow\LogRC{\FTypeB}\left(\LogTermApp{\RealTermA}{\LogTermListAbbrv{\LogTermVar{\LogNatZ}}}\right)}_\ValuatA$ we get:
\begin{equation*}
\CPOinterp{\CPOelA\left(\LTLamVar\,\LTNatZ\right)\left(\LTisrc_\FTypeA^{(1)}\,\LTLamListNil\right)}_{\ValuatA'}\in\RealVal{\LogRC{\FTypeB}\left(\LogTermApp{\RealTermA}{\LogTermListAbbrv{\LogTermVar{\LogNatZ}}}\right)}_\ValuatA
\end{equation*}
but the second induction hypothesis implies:
\begin{equation*}
\CPOinterp{\LTisrc_\FTypeB^{(2)}\left(\LTLamApp\,\RealTermA\left(\LTLamVar\,\LTNatZ\right)\right)}_{\ValuatA'}\in\RealVal{\LogRC{\FTypeB}\left(\LogTermApp{\RealTermA}{\LogTermListAbbrv{\LogTermVar{\LogNatZ}}}\right)\Rightarrow\LogNorm{\LogTermApp{\RealTermA}{\LogTermListAbbrv{\LogTermVar{\LogNatZ}}}}}_\ValuatA
\end{equation*}
and therefore:
\begin{equation*}
\CPOinterp{\LTisrc_\FTypeB^{(2)}\left(\LTLamApp\,\RealTermA\left(\LTLamVar\,\LTNatZ\right)\right)\left(\CPOelA\left(\LTLamVar\,\LTNatZ\right)\left(\LTisrc_\FTypeA^{(1)}\,\LTLamListNil\right)\right)}_{\ValuatA'}\in\RealVal{\LogNorm{\LogTermApp{\RealTermA}{\LogTermListAbbrv{\LogTermVar{\LogNatZ}}}}}_\ValuatA
\end{equation*}
and we conclude by proving that $\RealVal{\LogNorm{\LogTermApp{\RealTermA}{\LogTermListAbbrv{\LogTermVar{\LogNatZ}}}}}_\ValuatA\subseteq\RealVal{\LogNorm{\RealTermA}}_\ValuatA$. This inclusion is a consequence of $\RealVal{\forall\LogNatVarA\,\LogNNorm{\RealTermA}{\LogNatVarA}}_\ValuatA\subseteq\RealVal{\forall\LogNatVarA\,\LogNNorm{\LogTermApp{\RealTermA}{\LogTermListAbbrv{\LogTermVar{\LogNatZ}}}}{\LogNatVarA}}_\ValuatA$, which follows from $\RealVal{\LogNNorm{\RealTermA}{\RealNatA}}_\ValuatA\subseteq\RealVal{\LogNNorm{\LogTermApp{\RealTermA}{\LogTermListAbbrv{\LogTermVar{\LogNatZ}}}}{\RealNatA}}_\ValuatA$. This last inclusion comes from the fact that if $\RealTermA$ does not reach a normal form in $\RealNatA$ steps, then $\RealTermA\,0$ does not reach a normal form in $\RealNatA$ steps either.
\item$\CPOinterp{\LTisrc_{\FTypeA\to\FTypeB}^{(3)}}_{\ValuatA'}\in\RealVal{\forall\LogTermVarA\,\forall\LogTermVarB\,\forall\LogTermListVarA\left(\LogRC{\FTypeA\to\FTypeB}\left(\LogTermApp{\LogTermSubst{\LogTermVarA}{\LogTermVarB}}{\LogTermListVarA}\right)\Rightarrow\LogRC{\FTypeA\to\FTypeB}\left(\LogTermApp{\LogTermApp{\LogTermLam{\LogTermVarA}}{\LogTermVarB}}{\LogTermListVarA}\right)\right)}_\ValuatA$: let $\RealTermA\in\Lambda$, $\RealTermB\in\Lambda$, $\RealTermListA\in\Lambda^*$, $\CPOelA\in\RealVal{\LogRC{\FTypeA\to\FTypeB}\left(\LogTermApp{\LogTermSubst{\RealTermA}{\RealTermB}}{\RealTermListA}\right)}_\ValuatA$, $\RealTermC\in\Lambda$ and $\CPOelB\in\RealVal{\LogRC{\FTypeA}\left(\RealTermC\right)}_\ValuatA$. The induction hypothesis implies:
\begin{equation*}
\CPOinterp{\LTisrc_\FTypeB^{(3)}\,\RealTermA\,\RealTermB\left(\LTLamListCons\,\RealTermListA\,\RealTermC\right)}_{\ValuatA'}\in\RealVal{\LogRC{\FTypeB}\left(\LogTermApp{\LogTermSubst{\RealTermA}{\RealTermB}}{\LogTermListCons{\RealTermListA}{\RealTermC}}\right)\Rightarrow\LogRC{\FTypeB}\left(\LogTermApp{\LogTermApp{\left(\LogTermLam{\RealTermA}\right)}{\LogTermListAbbrv{\RealTermB}}}{\LogTermListCons{\RealTermListA}{\RealTermC}}\right)}_\ValuatA
\end{equation*}
and we also have:
\begin{equation*}
\CPOinterp{\CPOelA\,\RealTermC\,\CPOelB}\in\RealVal{\LogRC{\FTypeB}\left(\LogTermApp{\LogTermApp{\LogTermSubst{\RealTermA}{\RealTermB}}{\RealTermListA}}{\LogTermListAbbrv{\RealTermC}}\right)}_\ValuatA
\end{equation*}
so since:
\begin{equation*}
\CPOinterp{\LTLamApp\left(\LTLamListApp\left(\LTLamSubst\,\RealTermA\,\LTNatZ\left(\LTLamListCons\,\LTLamListNil\,\RealTermB\right)\right)\RealTermListA\right)\RealTermC}=\CPOinterp{\LTLamListApp\left(\LTLamSubst\,\RealTermA\,\LTNatZ\left(\LTLamListCons\,\LTLamListNil\,\RealTermB\right)\right)\left(\LTLamListCons\,\RealTermListA\,\RealTermC\right)}
\end{equation*}
we have:
\begin{equation*}
\RealVal{\LogRC{\FTypeB}\left(\LogTermApp{\LogTermApp{\LogTermSubst{\RealTermA}{\RealTermB}}{\RealTermListA}}{\LogTermListAbbrv{\RealTermC}}\right)}_\ValuatA=\RealVal{\LogRC{\FTypeB}\left(\LogTermApp{\LogTermSubst{\RealTermA}{\RealTermB}}{\LogTermListCons{\RealTermListA}{\RealTermC}}\right)}_\ValuatA
\end{equation*}
and therefore:
\begin{equation*}
\CPOinterp{\LTisrc_\FTypeB^{(3)}\,\RealTermA\,\RealTermB\left(\LTLamListCons\,\RealTermListA\,\RealTermC\right)\left(\CPOelA\,\RealTermC\,\CPOelB\right)}_{\ValuatA'}\in\RealVal{\LogRC{\FTypeB}\left(\LogTermApp{\LogTermApp{\left(\LogTermLam{\RealTermA}\right)}{\LogTermListAbbrv{\RealTermB}}}{\LogTermListCons{\RealTermListA}{\RealTermC}}\right)}_\ValuatA
\end{equation*}
but finally since:
\begin{equation*}
\CPOinterp{\LTLamListApp\left(\LTLamApp\left(\LTLamAbs\,\RealTermA\right)\RealTermB\right)\left(\LTLamListCons\,\RealTermListA\,\RealTermC\right)}=\CPOinterp{\LTLamApp\left(\LTLamListApp\left(\LTLamApp\left(\LTLamAbs\,\RealTermA\right)\RealTermB\right)\RealTermListA\right)\RealTermC}
\end{equation*}
we obtain:
\begin{equation*}
\CPOinterp{\LTisrc_\FTypeB^{(3)}\,\RealTermA\,\RealTermB\left(\LTLamListCons\,\RealTermListA\,\RealTermC\right)\left(\CPOelA\,\RealTermC\,\CPOelB\right)}_{\ValuatA'}\in\RealVal{\LogRC{\FTypeB}\left(\LogTermApp{\LogTermApp{\LogTermApp{\left(\LogTermLam{\RealTermA}\right)}{\LogTermListAbbrv{\RealTermB}}}{\RealTermListA}}{\LogTermListAbbrv{\RealTermC}}\right)}_\ValuatA
\end{equation*}
\end{itemize}
\item$\forall\LogPredVarA\,\FTypeA$:
\begin{itemize}
\item$\CPOinterp{\LTisrc_{\forall\LogPredVarA\,\FTypeA}^{(1)}}_{\ValuatA'}\mkern-1mu\in\RealVal{\forall\LogTermListVarA\,\LogRC{\forall\LogPredVarA\,\FTypeA}\left(\LogTermApp{\LogTermVar{\LogNatZ}}{\LogTermListVarA}\right)}_\ValuatA$: let $\RealTermListA\in\Lambda^*$, $\RealPredA\subseteq\Lambda$ and $\CPOelA\in\RealVal{\LogRedCand{\LogPredVarForm{\LogPredVarA}}}_{\ValuatA\uplus\SetSuch{\LogPredVarA\mapsto\RealPredA}{}}$. Then by induction hypothesis:
\begin{equation*}
\CPOinterp{\LTisrc_\FTypeA^{(1)}\,\RealTermListA}_{\ValuatA'\uplus\SetSuch{\LTVarA_\LogPredVarA\mapsto\CPOelA}{}}\in\RealVal{\LogRC{\FTypeA}\left(\LogTermApp{\LogTermVar{\LogNatZ}}{\RealTermListA}\right)}_{\ValuatA\uplus\SetSuch{\LogPredVarA\mapsto\RealPredA}{}}
\end{equation*}
therefore:
\begin{equation*}
\CPOinterp{\lambda\LTVarA_\LogPredVarA.\LTisrc_\FTypeA^{(1)}\,\RealTermListA}_{\ValuatA'}\in\RealVal{\LogRedCand{\LogPredVarForm{\LogPredVarA}}\Rightarrow\LogRC{\FTypeA}\left(\LogTermApp{\LogTermVar{\LogNatZ}}{\RealTermListA}\right)}_{\ValuatA\uplus\SetSuch{\LogPredVarA\mapsto\RealPredA}{}}
\end{equation*}
and finally:
\begin{equation*}
\CPOinterp{\lambda\LTVarA_\LogPredVarA.\LTisrc_\FTypeA^{(1)}\,\RealTermListA}_{\ValuatA'}\in\RealVal{\LogRC{\forall\LogPredVarA\,\FTypeA}\left(\LogTermApp{\LogTermVar{\LogNatZ}}{\RealTermListA}\right)}_\ValuatA
\end{equation*}
\item$\CPOinterp{\LTisrc_{\forall\LogPredVarA\,\FTypeA}^{(2)}}_{\ValuatA'}\in\RealVal{\forall\LogTermVarA\left(\LogRC{\forall\LogPredVarA\,\FTypeA}\left(\LogTermVarA\right)\Rightarrow\LogNorm{\LogTermVarA}\right)}_\ValuatA$: let $\RealTermA\in\Lambda$ and $\CPOelA\in\RealVal{\LogRC{\forall\LogPredVarA\,\FTypeA}\left(\LogTermVarA\right)}_{\ValuatA\uplus\SetSuch{\LogTermVarA\mapsto\RealTermA}{}}$. The induction hypothesis implies that for any $\RealPredA\subseteq\Lambda$:
\begin{equation*}
\CPOinterp{\lambda\LTVarA_\LogPredVarA.\LTisrc_\FTypeA^{(2)}}_{\ValuatA'}\in\RealVal{\LogRedCand{\LogPredVarForm{\LogPredVarA}}\Rightarrow\forall\LogTermVarA\left(\LogRC{\FTypeA}\left(\LogTermVarA\right)\Rightarrow\LogNorm{\LogTermVarA}\right)}_{\ValuatA\uplus\SetSuch{\LogPredVarA\mapsto\RealPredA}{}}
\end{equation*}
therefore:
\begin{equation*}
\CPOinterp{\lambda\LTVarA_\LogPredVarA.\LTisrc_\FTypeA^{(2)}}_{\ValuatA'}\in\RealVal{\forall\LogPredVarA\left(\LogRedCand{\LogPredVarForm{\LogPredVarA}}\Rightarrow\forall\LogTermVarA\left(\LogRC{\FTypeA}\left(\LogTermVarA\right)\Rightarrow\LogNorm{\LogTermVarA}\right)\right)}_\ValuatA
\end{equation*}
so by lemma~\ref{RealElim}:
\begin{multline*}
\CPOinterp{\LTelim_{\LogPredVarForm{\LogPredVarA}\mapsto\LogRedCand{\LogPredVarForm{\LogPredVarA}}\Rightarrow\forall\LogTermVarA\left(\LogRC{\FTypeA}\left(\LogTermVarA\right)\Rightarrow\LogNorm{\LogTermVarA}\right),\LogNormForm}\left(\lambda\LTVarA_\LogPredVarA.\LTisrc_\FTypeA^{(2)}\right)}_{\ValuatA'}
\\
\in\RealVal{\LogRedCand{\LogPredVarForm{\LogNormForm}}\Rightarrow\forall\LogTermVarA\left(\left(\LogPredVarForm{\LogPredVarA}\mapsto\LogRC{\FTypeA}\left(\LogTermVarA\right)\right)\left(\LogNormForm\right)\Rightarrow\LogNorm{\LogTermVarA}\right)}_\ValuatA
\end{multline*}
and then by lemma~\ref{RealNormRC}:
\begin{multline*}
\CPOinterp{\LTelim_{\LogPredVarForm{\LogPredVarA}\mapsto\LogRedCand{\LogPredVarForm{\LogPredVarA}}\Rightarrow\forall\LogTermVarA\left(\LogRC{\FTypeA}\left(\LogTermVarA\right)\Rightarrow\LogNorm{\LogTermVarA}\right),\LogNormForm}\left(\lambda\LTVarA_\LogPredVarA.\LTisrc_\FTypeA^{(2)}\right)\LTnormrc\,\LogTermVarA}_{\ValuatA\uplus\SetSuch{\LogTermVarA\mapsto\RealTermA}{}}
\\
\in\RealVal{\left(\LogPredVarForm{\LogPredVarA}\mapsto\LogRC{\FTypeA}\left(\LogTermVarA\right)\right)\left(\LogNormForm\right)\Rightarrow\LogNorm{\LogTermVarA}}_{\ValuatA\uplus\SetSuch{\LogTermVarA\mapsto\RealTermA}{}}
\end{multline*}
but on the other hand, since:
\begin{equation*}
\CPOelA\in\RealVal{\forall\LogPredVarA\left(\LogRedCand{\LogPredVarForm{\LogPredVarA}}\Rightarrow\LogRC{\FTypeA}\left(\LogTermVarA\right)\right)}_{\ValuatA\uplus\SetSuch{\LogTermVarA\mapsto\RealTermA}{}}
\end{equation*}
and since by lemma~\ref{RealElim}, for any $\RealTermA\in\lambda$:
\begin{multline*}
\CPOinterp{\LTelim_{\LogPredVarForm{\LogPredVarA}\mapsto\LogRedCand{\LogPredVarForm{\LogPredVarA}}\Rightarrow\LogRC{\FTypeA}\left(\LogTermVarA\right),\LogNormForm}}_{\ValuatA\uplus\SetSuch{\LogTermVarA\mapsto\RealTermA}{}}\\*
\in\RealVal{\forall\LogPredVarA\left(\LogRedCand{\LogPredVarForm{\LogPredVarA}}\Rightarrow\LogRC{\FTypeA}\left(\LogTermVarA\right)\right)\Rightarrow\LogRedCand{\LogPredVarForm{\LogNormForm}}\Rightarrow\left(\LogPredVarForm{\LogPredVarA}\mapsto\LogRC{\FTypeA}\left(\LogTermVarA\right)\right)\left(\LogNormForm\right)}_{\ValuatA\uplus\SetSuch{\LogTermVarA\mapsto\RealTermA}{}}
\end{multline*}
we have:
\begin{equation*}
\CPOinterp{\LTelim_{\LogPredVarForm{\LogPredVarA}\mapsto\LogRedCand{\LogPredVarForm{\LogPredVarA}}\Rightarrow\LogRC{\FTypeA}\left(\LogTermVarA\right),\LogNormForm}\,\CPOelA\,\LTnormrc}_{\ValuatA\uplus\SetSuch{\LogTermVarA\mapsto\RealTermA}{}}\in\RealVal{\left(\LogPredVarForm{\LogPredVarA}\mapsto\LogRC{\FTypeA}\left(\LogTermVarA\right)\right)\left(\LogNormForm\right)}_{\ValuatA\uplus\SetSuch{\LogTermVarA\mapsto\RealTermA}{}}
\end{equation*}
and therefore we get:
\begin{gather*}
\CPOinterp{\LTelim_{\LogPredVarForm{\LogPredVarA}\mapsto\LogRedCand{\LogPredVarForm{\LogPredVarA}}\Rightarrow\forall\LogTermVarA\left(\LogRC{\FTypeA}\left(\LogTermVarA\right)\Rightarrow\LogNorm{\LogTermVarA}\right),\LogNormForm}\left(\lambda\LTVarA_\LogPredVarA.\LTisrc_\FTypeA^{(2)}\right)\LTnormrc\,\LogTermVarA
\right.\hspace{80pt}\\\hspace{80pt}\left.
\left(\LTelim_{\LogPredVarForm{\LogPredVarA}\mapsto\LogRedCand{\LogPredVarForm{\LogPredVarA}}\Rightarrow\LogRC{\FTypeA}\left(\LogTermVarA\right),\LogNormForm}\,\CPOelA\,\LTnormrc\right)}_{\ValuatA\uplus\SetSuch{\LogTermVarA\mapsto\RealTermA}{}}\in\RealVal{\LogNorm{\LogTermVarA}}_{\ValuatA\uplus\SetSuch{\LogTermVarA\mapsto\RealTermA}{}}=\RealVal{\LogNorm{\RealTermA}}
\end{gather*}
\item$\CPOinterp{\LTisrc_{\forall\LogPredVarA\,\FTypeA}^{(3)}}_{\ValuatA'}\in\RealVal{\forall\LogTermVarA\,\forall\LogTermVarB\,\forall\LogTermListVarA\left(\LogRC{\forall\LogPredVarA\,\FTypeA}\left(\LogTermApp{\LogTermSubst{\LogTermVarA}{\LogTermVarB}}{\LogTermListVarA}\right)\Rightarrow\LogRC{\forall\LogPredVarA\,\FTypeA}\left(\LogTermApp{\LogTermApp{\LogTermLam{\LogTermVarA}}{\LogTermVarB}}{\LogTermListVarA}\right)\right)}_\ValuatA$: let $\RealTermA\in\Lambda$, $\RealTermB\in\Lambda$, $\RealTermListA\in\Lambda^*$, $\CPOelA\in\RealVal{\LogRC{\forall\LogPredVarA\,\FTypeA}\left(\LogTermApp{\LogTermSubst{\RealTermA}{\RealTermB}}{\RealTermListA}\right)}_\ValuatA$, $\RealPredA\subseteq\Lambda$ and $\CPOelB\in\RealVal{\LogRedCand{\LogPredVarForm{\LogPredVarA}}}_{\ValuatA\uplus\SetSuch{\LogPredVarA\mapsto\RealPredA}{}}$. The induction hypothesis implies:
\begin{equation*}
\CPOinterp{\LTisrc_\FTypeA^{(3)}\,\RealTermA\,\RealTermB\,\RealTermListA}_{\ValuatA\uplus\SetSuch{\LTVarA_\LogPredVarA\mapsto\CPOelB}{}}\in\RealVal{\LogRC{\FTypeA}\left(\LogTermApp{\LogTermSubst{\RealTermA}{\RealTermB}}{\RealTermListA}\right)\Rightarrow\LogRC{\FTypeA}\left(\LogTermApp{\LogTermApp{\LogTermLam{\RealTermA}}{\RealTermB}}{\RealTermListA}\right)}_{\ValuatA\uplus\SetSuch{\LogPredVarA\mapsto\RealPredA}{}}
\end{equation*}
but we have also:
\begin{equation*}
\CPOinterp{\CPOelA\,\CPOelB}\in\RealVal{\LogRC{\FTypeA}\left(\LogTermApp{\LogTermSubst{\RealTermA}{\RealTermB}}{\RealTermListA}\right)}_{\ValuatA\uplus\SetSuch{\LogPredVarA\mapsto\RealPredA}{}}
\end{equation*}
therefore:
\begin{equation*}
\CPOinterp{\LTisrc_\FTypeA^{(3)}\,\RealTermA\,\RealTermB\,\RealTermListA\left(\CPOelA\,\CPOelB\right)}_{\ValuatA\uplus\SetSuch{\LTVarA_\LogPredVarA\mapsto\CPOelB}{}}\in\RealVal{\LogRC{\FTypeA}\left(\LogTermApp{\LogTermApp{\LogTermLam{\RealTermA}}{\RealTermB}}{\RealTermListA}\right)}_{\ValuatA\uplus\SetSuch{\LogPredVarA\mapsto\RealPredA}{}}
\end{equation*}
and therefore:
\begin{equation*}
\CPOinterp{\lambda\LTVarA_\LogPredVarA.\LTisrc_\FTypeA^{(3)}\,\RealTermA\,\RealTermB\,\RealTermListA\left(\CPOelA\,\LTVarA_\LogPredVarA\right)}_{\ValuatA'}\in\RealVal{\LogRedCand{\LogPredVarForm{\LogPredVarA}}\Rightarrow\LogRC{\FTypeA}\left(\LogTermApp{\LogTermApp{\LogTermLam{\RealTermA}}{\RealTermB}}{\RealTermListA}\right)}_{\ValuatA\uplus\SetSuch{\LogPredVarA\mapsto\RealPredA}{}}
\end{equation*}
and since this holds for any $\RealPredA\subseteq\Lambda$ we obtain:
\begin{equation*}
\CPOinterp{\lambda\LTVarA_\LogPredVarA.\LTisrc_\FTypeA^{(3)}\,\RealTermA\,\RealTermB\,\RealTermListA\left(\CPOelA\,\LTVarA_\LogPredVarA\right)}_{\ValuatA'}\in\RealVal{\LogRC{\forall\LogPredVarA\,\FTypeA}\left(\LogTermApp{\LogTermApp{\LogTermLam{\RealTermA}}{\RealTermB}}{\RealTermListA}\right)}_\ValuatA
\end{equation*}
which concludes the proof.\qedhere
\end{itemize}
\end{itemize}
\end{proof}
The last step of the interpretation of normalization of system F is the interpretation of lemma~\ref{FNorm}, which is given in figure~\ref{AdeqDef}. Despite the fact that each term defined there depends on a full typing derivation in system F, we use the informal notation $\LTadeq_{\Gamma\Entails\LogTermA:\FTypeA}$, refering to the full derivation only by its conclusion. In order to ease our definition, the terms $\LTadeq_{\Gamma\Entails\LogTermA:\FTypeA}$ contain the following free variables:
\begin{equation*}
\SetSuch{\LTVarA_\LogPredVarA}{\LogPredVarA\in\FV{\Gamma,\FTypeA}}\cup\SetSuch{\LogTermVarA_\FTypeB}{\FTypeB\in\Gamma}\cup\SetSuch{\LTVarB_\FTypeB}{\FTypeB\in\Gamma}
\end{equation*}
where $\LTVarA_\LogPredVarA$ is meant to be replaced with a realizer of $\LogRedCand{\LogPredVarForm{\LogPredVarA}}$, $\LogTermVarA_\FTypeB$ is meant to be replaced with some term $\RealTermA_\FTypeB\in\Lambda$ and $\LTVarB_\FTypeB$ is meant to be replaced with a realizer of $\LogRC{\FTypeB}\left(\RealTermA_\FTypeB\right)$. In the notations $\LogTermVarA_\FTypeB$ and $\LTVarB_\FTypeB$, $\FTypeB$ refers to an occurence of $\FTypeB$ in $\Gamma$, rather than to $\FTypeB$ itself. The notation $\vec{\LogTermVarA_\Gamma}$ in figure~\ref{AdeqDef} and in the lemma is a shorthand for $\LTLamListCons\left(\LTLamListCons\left(\ldots\LTLamListCons\,\LTLamListNil\,t_{\FTypeB_0}\ldots\right)t_{\FTypeB_{n-1}}\right)$ if $\Gamma=\FTypeB_{n-1},\ldots,\FTypeB_0$.
\begin{figure*}
\begin{gather*}
\begin{align*}
\LTadeq_{\Gamma\Entails\LogTermVar{\LogNatA}:\FTypeB}&=\LTVarB_\FTypeB
&
\LTadeq_{\Gamma\Entails\LogTermA:\forall\LogPredVarA\,\FTypeA}&=\lambda\LTVarA_\LogPredVarA.\LTadeq_{\Gamma\Entails\LogTermA:\FTypeA}
&
\end{align*}
\\
\LTadeq_{\Gamma\Entails\LogTermA:\FTypeA\Subst{\LogPredVarA}{\FTypeB}}=\LTelim_{\LogPredVarForm{\LogPredVarA}\mapsto\LogRedCand{\LogPredVarForm{\LogPredVarA}}\Rightarrow\LogRC{\FTypeA}\left(\LogTermSubst{\LogTermA}{\vec{\LogTermVarA_\Gamma}}\right),\LogRC{\FTypeB}}\,\LTadeq_{\Gamma\Entails\LogTermA:\forall\LogPredVarA\,\FTypeA}\,\LTisrc_\FTypeB
\\
\LTadeq_{\Gamma\Entails\LogTermLam{\LogTermA}:\FTypeB\to\FTypeA}=\lambda\LogTermVarA_\FTypeB\LTVarB_\FTypeB.\LTisrc_\FTypeA^{(3)}\left(\LTLamSubst\,\LTinterp{\LogTermA}\left(\LTNatS\,\LTNatZ\right)\left(\LTLamListShift\,\vec{\LogTermVarA_\Gamma}\right)\right)\LogTermVarA_\FTypeB\,\LTLamListNil\,\LTadeq_{\Gamma,\FTypeB\Entails\LogTermA:\FTypeA}
\\
\LTadeq_{\Gamma\Entails\LogTermApp{\LogTermA}{\LogTermB}:\FTypeA}=\LTadeq_{\Gamma\Entails\LogTermA:\FTypeB\to\FTypeA}\left(\LTLamSubst\,\LTinterp{\LogTermB}\,\LTNatZ\,\vec{\LogTermVarA_\Gamma}\right)\LTadeq_{\Gamma\Entails\LogTermB:\FTypeB}
\end{gather*}
\vspace{-15pt}
\caption{Definition of $\LTadeq_{\Gamma\Entails\LogTermA:\FTypeA}$}
\label{AdeqDef}
\end{figure*}
It can then be shown that the terms $\LTadeq_{\Gamma\Entails\LogTermA:\FTypeA}$ satisfy the intended property:
\begin{thm}
\label{RealAdeq}
If $\Gamma\Entails\LogTermA:\FTypeA$ is a valid typing judgement in system F, and if $\ValuatA$ is a valuation such that:
\begin{itemize}
\item$\ValuatA\left(\LogPredVarA\right)\subseteq\Lambda$ and $\ValuatA\left(\LTVarA_\LogPredVarA\right)\in\RealVal{\LogRedCand{\LogPredVarForm{\LogPredVarA}}}_\ValuatA$ for each variable $\LogPredVarA\in\FV{\Gamma,\FTypeA}$
\item$\ValuatA\left(\LogTermVarA_\FTypeB\right)\in\Lambda$ and $\ValuatA\left(\LTVarB_\FTypeB\right)\in\RealVal{\LogRC{\FTypeB}\left(\LogTermVarA_\FTypeB\right)}_\ValuatA$ for each $\FTypeB\in\Gamma$
\end{itemize}
then $\CPOinterp{\LTadeq_{\Gamma\Entails\LogTermA:\FTypeA}}_\ValuatA\in\RealVal{\LogRC{\FTypeA}\left(\LogTermSubst{\LogTermA}{\vec{\LogTermVarA_\Gamma}}\right)}_\ValuatA$
\end{thm}
\begin{proof}
\begin{itemize}
\item$\Gamma\Entails\LogTermVar{\LogNatA}:\FTypeB$: the hypothesis gives:
\begin{equation*}
\ValuatA\left(\LTVarB_\FTypeB\right)\in\RealVal{\LogRC{\FTypeB}\left(\LogTermVarA_\FTypeB\right)}_\ValuatA
\end{equation*}
but since $\CPOinterp{\LTLamSubst\left(\LTLamVar\left(\LTNatS^m\,\LTNatZ\right)\right)\LTNatZ\,\vec{\LogTermVarA_\Gamma}}_\ValuatA=\CPOinterp{\LTinterp{\LogTermVarA_\FTypeB}}_\ValuatA$ we obtain:
\begin{equation*}
\RealVal{\LogRC{\FTypeB}\left(\LogTermSubst{\LogTermVar{\LogNatA}}{\vec{\LogTermVarA_\Gamma}}\right)}_\ValuatA=\RealVal{\LogRC{\FTypeB}\left(\LogTermVarA_\FTypeB\right)}_\ValuatA
\end{equation*}
\item$\Gamma\Entails\LogTermLam{\LogTermA}:\FTypeB\to\FTypeA$: let $\RealTermA\in\Lambda$ and $\CPOelA\in\RealVal{\LogRC{\FTypeB}\left(\RealTermA\right)}_\ValuatA$. If we write:
\begin{equation*}
\RealTermB=\CPOinterp{\LTLamSubst\,\LTinterp{\LogTermA}\left(\LTNatS\,\LTNatZ\right)\left(\LTLamListShift\,\vec{\LogTermVarA_\Gamma}\right)}_\ValuatA\in\Lambda
\end{equation*}
then lemma~\ref{RealIsRC} implies:
\begin{equation*}
\CPOinterp{\LTisrc_\FTypeA^{(3)}\,\RealTermB\,\RealTermA\,\LTLamListNil}_\ValuatA\in\RealVal{\LogRC{\FTypeA}\left(\LogTermApp{\LogTermSubst{\RealTermB}{\RealTermA}}{\LogTermListNil}\right)\Rightarrow\LogRC{\FTypeA}\left(\LogTermApp{\LogTermApp{\left(\LogTermLam{\RealTermB}\right)}{\RealTermA}}{\LogTermListNil}\right)}_\ValuatA
\end{equation*}
and on the other hand, the induction hypothesis implies:
\begin{equation*}
\CPOinterp{\LTadeq_{\Gamma,\FTypeB\Entails\LogTermA:\FTypeA}}_{\ValuatA\uplus\SetSuch{\LogTermVarA_\FTypeB\mapsto\RealTermA;\LTVarA_\FTypeB\mapsto\CPOelA}{}}\in\RealVal{\LogRC{\FTypeA}\left(\LogTermSubst{\LogTermA}{\vec{\LogTermVarA_{\Gamma,\FTypeB}}}\right)}_{\ValuatA\uplus\SetSuch{\LogTermVarA_\FTypeB\mapsto\RealTermA}{}}
\end{equation*}
but since we have by a version of lemma~\ref{SubstLemma} in system $\LTbbc$:
\begin{align*}
\CPOinterp{\LTinterp{\left(\LogTermApp{\LogTermSubst{\RealTermB}{\RealTermA}}{\LogTermListNil}\right)}}_\ValuatA
&=\CPOinterp{\LTLamSubst\left(\LTLamSubst\,\LTinterp{\LogTermA}\left(\LTNatS\,\LTNatZ\right)\left(\LTLamListShift\,\vec{\LogTermVarA_\Gamma}\right)\right)\LTNatZ\left(\LTLamListCons\,\LTLamListNil\,\RealTermA\right)}_\ValuatA\\*
&=\CPOinterp{\LTLamSubst\,\LTinterp{\LogTermA}\,\LTNatZ\,\vec{\LogTermVarA_{\Gamma,\FTypeB}}}_{\ValuatA\uplus\SetSuch{\LogTermVarA_\FTypeB\mapsto\RealTermA}{}}\\*
&=\CPOinterp{\LTinterp{\left(\LogTermSubst{\LogTermA}{\vec{\LogTermVarA_{\Gamma,\FTypeB}}}\right)}}_{\ValuatA\uplus\SetSuch{\LogTermVarA_\FTypeB\mapsto\RealTermA}{}}
\end{align*}
we also have:
\begin{equation*}
\RealVal{\LogRC{\FTypeA}\left(\LogTermApp{\LogTermSubst{\RealTermB}{\RealTermA}}{\LogTermListNil}\right)}_\ValuatA=\RealVal{\LogRC{\FTypeA}\left(\LogTermSubst{\LogTermA}{\vec{\LogTermVarA_{\Gamma,\FTypeB}}}\right)}_{\ValuatA\uplus\SetSuch{\LogTermVarA_\FTypeB\mapsto\RealTermA}{}}
\end{equation*}
and therefore:
\begin{multline*}
\CPOinterp{\LTisrc_\FTypeA^{(3)}\left(\LTLamSubst\,\LTinterp{\LogTermA}\left(\LTNatS\,\LTNatZ\right)\left(\LTLamListShift\,\vec{\LogTermVarA_\Gamma}\right)\right)\LogTermVarA_\FTypeB\,\LTLamListNil\,\LTadeq_{\Gamma,\FTypeB\Entails\LogTermA:\FTypeA}}_{\ValuatA\uplus\SetSuch{\LogTermVarA_\FTypeB\mapsto\RealTermA;\LTVarA_\FTypeB\mapsto\CPOelA}{}}\\
\in\RealVal{\LogRC{\FTypeA}\left(\LogTermApp{\left(\LogTermLam{\RealTermB}\right)}{\RealTermA}\right)}_\ValuatA
\end{multline*}
\item$\Gamma\Entails\LogTermApp{\LogTermA}{\LogTermB}:\FTypeA$: the first induction hypothesis implies:
\begin{equation*}
\CPOinterp{\LTadeq_{\Gamma\Entails\LogTermA:\FTypeB\to\FTypeA}\left(\LTLamSubst\,\LTinterp{\LogTermB}\,\LTNatZ\,\vec{\LogTermVarA_\Gamma}\right)}_\ValuatA\in\RealVal{\LogRC{\FTypeB}\left(\LogTermSubst{\LogTermB}{\vec{\LogTermVarA_\Gamma}}\right)\Rightarrow\LogRC{\FTypeA}\left(\LogTermApp{\LogTermSubst{\LogTermA}{\vec{\LogTermVarA_\Gamma}}}{\LogTermListAbbrv{\LogTermSubst{\LogTermB}{\vec{\LogTermVarA_\Gamma}}}}\right)}_\ValuatA
\end{equation*}
and the second induction hypothesis gives:
\begin{equation*}
\CPOinterp{\LTadeq_{\Gamma\Entails\LogTermB:\FTypeB}}_\ValuatA\\*
\in\RealVal{\LogRC{\FTypeB}\left(\LogTermSubst{\LogTermB}{\vec{\LogTermVarA_\Gamma}}\right)}_\ValuatA
\end{equation*}
so since:
\begin{equation*}
\CPOinterp{\LTLamApp\left(\LTLamSubst\,\LTinterp{\LogTermA}\,\LTNatZ\,\vec{\LogTermVarA_\Gamma}\right)\left(\LTLamSubst\,\LTinterp{\LogTermB}\,\LTNatZ\,\vec{\LogTermVarA_\Gamma}\right)}_\ValuatA=\CPOinterp{\LTLamSubst\left(\LTLamApp\,\LTinterp{\LogTermA}\,\LTinterp{\LogTermA}\right)\LTNatZ\,\vec{\LogTermVarA_\Gamma}}_\ValuatA
\end{equation*}
we obtain:
\begin{equation*}
\CPOinterp{\LTadeq_{\Gamma\Entails\LogTermA:\FTypeB\to\FTypeA}\left(\LTLamSubst\,\LTinterp{\LogTermB}\,\LTNatZ\,\vec{\LogTermVarA_\Gamma}\right)\LTadeq_{\Gamma\Entails\LogTermB:\FTypeB}}_\ValuatA\in\RealVal{\LogRC{\FTypeA}\left(\LogTermSubst{\left(\LogTermApp{\LogTermA}{\LogTermListAbbrv{\LogTermB}}\right)}{\vec{\LogTermVarA_\Gamma}}\right)}_\ValuatA
\end{equation*}
\item$\Gamma\Entails\LogTermA:\forall\LogPredVarA\,\FTypeA$: let $\RealPredA\subseteq\Lambda$ and $\CPOelA\in\RealVal{\LogRedCand{\LogPredVarForm{\LogPredVarA}}}_{\ValuatA\uplus\SetSuch{\LogPredVarA\mapsto\RealPredA}{}}$. The induction hypothesis gives immediately:
\begin{equation*}
\CPOinterp{\LTadeq_{\Gamma\Entails\LogTermA:\FTypeA}}_{\ValuatA\uplus\SetSuch{\LogPredVarA\mapsto\RealPredA;\LTVarA_\LogPredVarA\mapsto\CPOelA}{}}\in\RealVal{\LogRC{\FTypeA}\left(\LogTermSubst{\LogTermA}{\vec{\LogTermVarA_\Gamma}}\right)}_{\ValuatA\uplus\SetSuch{\LogPredVarA\mapsto\RealPredA;\LTVarA_\LogPredVarA\mapsto\CPOelA}{}}
\end{equation*}
\item$\Gamma\Entails\LogTermA:\FTypeA\Subst{\LogPredVarA}{\FTypeB}$: the induction hypothesis gives:
\begin{equation*}
\CPOinterp{\LTadeq_{\Gamma\Entails\LogTermA:\forall\LogPredVarA\,\FTypeA}}_\ValuatA\in\RealVal{\forall\LogPredVarA\left(\LogRedCand{\LogPredVarForm{\LogPredVarA}}\Rightarrow\LogRC{\FTypeA}\left(\LogTermSubst{\LogTermA}{\vec{\LogTermVarA_\Gamma}}\right)\right)}_\ValuatA
\end{equation*}
therefore lemma~\ref{RealElim} implies:
\begin{multline*}
\CPOinterp{\LTelim_{\LogPredVarForm{\LogPredVarA}\mapsto\LogRedCand{\LogPredVarForm{\LogPredVarA}}\Rightarrow\LogRC{\FTypeA}\left(\LogTermSubst{\LogTermA}{\vec{\LogTermVarA_\Gamma}}\right),\LogRC{\FTypeB}}\LTadeq_{\Gamma\Entails\LogTermA:\forall\LogPredVarA\,\FTypeA}}_\ValuatA\\
\in\RealVal{\LogRedCand{\LogRC{\FTypeB}}\Rightarrow\left(\LogPredVarForm{\LogPredVarA}\mapsto\LogRC{\FTypeA}\left(\LogTermSubst{\LogTermA}{\vec{\LogTermVarA_\Gamma}}\right)\right)\left(\LogRC{\FTypeB}\right)}_\ValuatA
\end{multline*}
and since by lemma~\ref{RealIsRC} we have:
\begin{equation*}
\CPOinterp{\LTisrc_\FTypeB}_\ValuatA\in\RealVal{\LogRedCand{\LogRC{\FTypeB}}}_\ValuatA
\end{equation*}
and moreover for any $\LogTermA$:
\begin{equation*}
\left(\LogPredVarForm{\LogPredVarA}\mapsto\LogRC{\FTypeA}\left(\LogTermA\right)\right)\left(\LogRC{\FTypeB}\right)\equiv\LogRC{\FTypeA\Subst{\LogPredVarA}{\FTypeB}}\left(\LogTermA\right)
\end{equation*}
we obtain:
\begin{equation*}
\CPOinterp{\LTelim_{\LogPredVarForm{\LogPredVarA}\mapsto\LogRedCand{\LogPredVarForm{\LogPredVarA}}\Rightarrow\LogRC{\FTypeA}\left(\LogTermSubst{\LogTermA}{\vec{\LogTermVarA_\Gamma}}\right),\LogRC{\FTypeB}}\LTadeq_{\Gamma\Entails\LogTermA:\forall\LogPredVarA\,\FTypeA}\,\LTisrc_\FTypeB}_\ValuatA\in\RealVal{\LogRC{\FTypeA\Subst{\LogPredVarA}{\FTypeB}}\left(\LogTermSubst{\LogTermA}{\vec{\LogTermVarA_\Gamma}}\right)}_\ValuatA
\end{equation*}
which concludes the proof.\qedhere
\end{itemize}
\end{proof}
Finally, if a closed term $\LogTermA$ is of closed type $\FTypeA$ in system F we can define:
\begin{equation*}
\LTnorm_{\Entails\LogTermA:\FTypeA}=\LTisrc_{\FTypeA}^{(2)}\,\LTinterp{\LogTermA}\,\LTadeq_{\Entails\LogTermA:\FTypeA}
\end{equation*}
Immediately, we have:
\begin{equation*}
\CPOinterp{\LTnorm_{\Entails\LogTermA:\FTypeA}}\in\RealVal{\LogNorm{\LogTermA}}=\RealVal{\neg\forall\LogNatVarA\,\LogNNorm{\LogTermA}{\LogNatVarA}}
\end{equation*}
As a final step we extract a witness $\RealNatA\in\mathbb{N}$ such that $\LogTermA$ normalizes in at most $\RealNatA$ steps of weak head reduction. The technique is standard in realizability for classical logic and requires that we fix the set of realizers of false boolean formulas to a well-chosen set:
\begin{thm}
If a closed term $\LogTermA$ is of closed type $\FTypeA$ in system F, then $\LTnorm_{\Entails\LogTermA:\FTypeA}\left(\lambda\LTVarA.\LTVarA\right)$ reduces to some $\LTNatS^\RealNatA\LTNatZ$ where $\RealNatA$ is such that $\LogTermA$ reaches a weak head normal form in at most $\RealNatA$ steps.
\end{thm}
\begin{proof}
We first fix the set of realizers of false boolean formulas:
\begin{equation*}
\RealBot=\SetSuch{\RealNatA\in\mathbb{N}}{\LogTermA\text{ reaches a normal form in at most $\RealNatA$ steps}}
\end{equation*}
Now we prove that:
\begin{equation*}
\CPOinterp{\lambda\LTVarA.\LTVarA}\in\RealVal{\forall\LogNatVarA\,\LogNNorm{\LogTermA}{\LogNatVarA}}
\end{equation*}
Indeed, let $\RealNatA\in\mathbb{N}$ and let show that $\RealNatA\in\RealVal{\LogNNorm{\LogTermA}{\RealNatA}}$. If $\RealNatA\in\RealBot$ then $\CPOinterp{\LTinterp{\LogTermA}}=\LogTermA$ reaches a normal form in at most $\RealNatA$ steps so $\RealVal{\LogNNorm{\LogTermA}{\RealNatA}}=\RealBot$ and therefore $\RealNatA\in\RealVal{\LogNNorm{\LogTermA}{\RealNatA}}$. If $\RealNatA\notin\RealBot$ then $\CPOinterp{\LTinterp{\LogTermA}}=\LogTermA$ can reduce for $\RealNatA$ steps without reaching a normal form so $\RealVal{\LogNNorm{\LogTermA}{\RealNatA}}=\mathbb{N}_\CPObot$ and therefore $\RealNatA\in\RealVal{\LogNNorm{\LogTermA}{\RealNatA}}$ trivially. Using that result, we obtain:
\begin{equation*}
\CPOinterp{\LTnorm_{\Entails\LogTermA:\FTypeA}\left(\lambda\LTVarA.\LTVarA\right)}\in\RealVal{\LogBoolF}=\RealBot
\end{equation*}
Now, computational adequacy of the model with respect to system $\LTbbc$ implies that $\LTnorm_{\Entails\LogTermA:\FTypeA}\left(\lambda\LTVarA.\LTVarA\right)$ reduces to some $\LTNatS^\RealNatA\LTNatZ$ where $\RealNatA\in\RealBot$, so $\RealNatA$ is such that $\LogTermA$ reaches a weak head normal form in at most $\RealNatA$ steps.
\end{proof}
It is easy to implement one-step weak head reduction in system $\LTbbc$, that is, there exists a term $\LTred:\LTTypeLam\to\LTTypeLam$ such that for every $\lambda$-term $\LogTermA$:
\begin{itemize}
\item if $\LogTermA\FRed\LogTermB$, then $\LTred\,\LTinterp{\LogTermA}\LTRed^*\LTinterp{\LogTermB}$
\item if $\LogTermA$ is in weak head normal form then $\LTred\,\LTinterp{\LogTermA}\LTRed^*\LTinterp{\LogTermA}$
\end{itemize}
Therefore, using our extracted bound we can compute the normal form of any closed $\lambda$-term $\LogTermA$ of closed type $\FTypeA$ in system F:
\begin{equation*}
\LTNatIt\,\LTinterp{\LogTermA}\,\LTred\left(\LTnorm_{\Entails\LogTermA:\FTypeA}\left(\lambda\LTVarA.\LTVarA\right)\right)\LTRed^*\LTValA
\end{equation*}
where $\LTValA$ is the representation of the normal form of $\LogTermA$ in system $\LTbbc$.
\bibliographystyle{plain}
\bibliography{paper}
\end{document}